\pdfoutput=1
\documentclass[11pt]{article}

\usepackage[numbers]{natbib}
\usepackage{macros}
\usepackage{xspace}
\usepackage{xcolor}
\definecolor{darkblue}{rgb}{0,0,.5}
\usepackage[colorlinks=true,allcolors=darkblue]{hyperref}
\usepackage{bbm}
\usepackage{tikz}
\usetikzlibrary{matrix}
\usepackage{pgfplots}
\usepackage{overpic}
\usepackage{makecell}
\usepgfplotslibrary{fillbetween}
\usetikzlibrary{patterns}
\usepackage{authblk}
\usepackage{enumerate}
\usepackage{amssymb}
\usepackage{amsbsy}
\usepackage{amsmath}
\usepackage{amsthm}
\usepackage{amsfonts}
\usepackage{latexsym}
\usepackage{graphicx}
\usepackage{color}
\usepackage{xifthen}
\usepackage{algorithm}
\usepackage[noend]{algpseudocode}

\usepackage{latexsym}

\usepackage{subcaption}

\newcommand{\reals}{\mathbb{R}} 
\newcommand{\R}{\mathbb{R}}

\newcommand{\naturals}{\mathbb{N}} 
\newcommand{\N}{\mathbb{N}}

\newtheorem{theorem}{Theorem}
\newtheorem{lemma}{Lemma}[section]
\newtheorem{corollary}{Corollary}[section]

\newtheorem{definition}{Definition}[section]

\newtheorem{prop}{Proposition}

\renewcommand\log{\ln}

\usepackage{fullpage}
\begin{document}

\title{Bounding, Concentrating, and Truncating: Unifying Privacy Loss Composition for Data Analytics}
\author[1]{Mark Cesar}
\author[1]{Ryan Rogers}
\affil{Data Science Applied Research, LinkedIn}

\maketitle 

\begin{abstract}
Differential privacy (DP) provides rigorous privacy guarantees on individual's data while also allowing for accurate statistics to be conducted on the overall, sensitive dataset.  To design a private system, first private algorithms must be designed that can quantify the privacy loss of each outcome that is released.  However, private algorithms that inject noise into the computation are not sufficient to ensure individuals' data is protected due to many noisy results ultimately concentrating to the true, non-privatized result.  Hence there have been several works providing precise formulas for how the privacy loss accumulates over multiple interactions with private algorithms.  However, these formulas either provide very general bounds on the privacy loss, at the cost of being overly pessimistic for certain types of private algorithms, or they can be too narrow in scope to apply to general privacy systems. 
In this work, we unify existing privacy loss composition bounds for special classes of differentially private (DP) algorithms along with general DP composition bounds. In particular, we provide strong privacy loss bounds when an analyst may select pure DP, bounded range (e.g. exponential mechanisms), or concentrated DP mechanisms in any order. We also provide optimal privacy loss bounds that apply when an analyst can select pure DP and bounded range mechanisms in a batch, i.e. non-adaptively. Further, when an analyst selects mechanisms within each class adaptively, we show a difference in privacy loss between different, predetermined orderings of pure DP and bounded range mechanisms.  Lastly, we compare the composition bounds of Laplace and Gaussian mechanisms based on histogram datasets and present new top-$k$ private algorithms based on truncated Gaussian noise.
\end{abstract}
 
 \clearpage
 \tableofcontents
 \clearpage
 
\section{Introduction} 
Differential privacy (DP) provides a mathematical formalism to an intuitive notion of what it means for a computation to be private --- the computation should produce similar results with or without any one individual's data.  With this formalism, we can quantify the privacy loss of a computation that injects noise when evaluated on a sensitive dataset.  This allows us to determine which DP algorithms are more private than others, i.e. which has smaller privacy loss.  Furthermore, if multiple computations are done on the same dataset we can still quantify the privacy loss over the entire interaction with the dataset, i.e. DP composes.  As opposed to measuring utility empirically, e.g. prediction accuracy of a classification task, privacy loss in DP requires analytical bounds over worst case datasets and outcomes.  Improvements to the privacy loss bounds show that a given algorithm might actually be more private than originally proposed, with no changes to the algorithm itself.  

Hence, there have been many works in precisely bounding the overall privacy loss.  There are multiple composition bounds to use, including bounds that hold for any DP algorithms \cite{DworkMcNiSm06, DworkRoVa10, KairouzOhVi17, MurtaghVa16}, as well as improved composition bounds that only apply to specific types of DP algorithms \cite{Abadietal16, Mironov17, BunSt16, DongRoSu19, DurfeeRo19, DongDuRo19}.  In the design of a DP system, we would like to provide the best possible bounds on the privacy loss that apply for combinations of general DP algorithms as well as specific types of DP algorithms that enjoy much better composition bounds.  One can simply use the most general formulations to provide a loose bound on the overall privacy loss, but this neglects the improvements that can be made, which allow for more queries or more accurate results.  

Consider a privacy system for data analytics, supporting tasks such as counting queries and exploratory analysis.  These general tasks typically use the Laplace mechanism \cite{DworkMcNiSm06} or Gaussian mechanism \cite{DworkKeMcMiNa06} to provide noisy counts,  as well as exponential mechanisms \cite{McSherryTa07}, a general class of DP algorithms that have been shown to achieve much improved composition bounds.  In particular, \citet{DongDuRo19} showed that one can query nearly four times more exponential mechanisms for the same overall privacy bound as what can be achieved with using the general, optimal DP composition bounds \cite{KairouzOhVi17}.  This improvement is because \citet{DurfeeRo19} defined a stronger condition that exponential mechanisms satisfy, called \emph{bounded range} (BR).  In particular $\diffp$-BR implies $\diffp$-DP, whereas $\diffp$-DP implies $2\diffp$-BR.  Further, the Gaussian mechanism does not satisfy (pure) $\diffp$-DP, but rather a slight relaxation called \emph{concentrated} DP (CDP) \cite{DworkRo16} or zero-mean concentrated DP (zCDP) \cite{BunSt16}.  

To see that combining DP and BR mechanisms can arise naturally, consider a privacy system that allows for general top-$k$ queries.  One would typically use a two phase approach to ensure DP.  The first phase would use a series of exponential mechanisms to \emph{discover} the domain of elements in the top-$k$.  Given the discovered set, the second phase goes back to the dataset to add noise to the true counts of the discovered elements via the Laplace mechanism and then release the noisy counts.    In fact, most DP top-$k$ algorithms use this two phase approach, see for example \citet{BhaskarLaSmTh10} and \citet{DurfeeRo19}.  One approach to bounding the privacy loss of such an interaction would be to simply use the general DP composition bounds, but this ignores the improved composition bounds that are possible via the BR analysis.  Another approach would be to analyze the composition bounds via BR, but this results in doubling the privacy parameter for each Laplace mechanism, as was done in LinkedIn's privacy system that handles top-$k$ queries \cite{RogersSuPeDuLeKaSaAh20}. We then follow a line of research proposed in \citet{DongDuRo19}, studying the privacy loss bounds that combine both general DP bounds and improved BR bounds.

We make several contributions in this work.  First, we provide a bound on the overall privacy guarantee in a much more general setting than has been studied before. Specifically, we allow bounds on the overall privacy loss when an analyst can select at each round of interaction a privacy parameter from a set of preregistered privacy parameters $\cE$, without replacement, and then selects a corresponding mechanism that can be selected as a function of previous outcomes.  It seems more natural for an analyst's choice of mechanism to also be dependent on a privacy parameter, since that determines the level of accuracy of the mechanism.  Previous composition bounds required knowing the ordering of the privacy parameters in advance, so could not be selected as an analyst interacts with a private system, with the exception of the \emph{privacy odometer} bounds from \citet{RogersRoUlVa16} which does not require a preregistered set of privacy parameters but defines a different privacy guarantee than traditional DP.  Although loose, we provide bounds to the overall privacy loss in this more general setting which are able to beat optimal bounds for general pure-DP mechanisms by incorporating improved bounds for BR and concentrated DP mechanisms.  These results allow for much more freedom in how the analyst interacts with a privacy system; see Section~\ref{sec:concentration}.

Our next contribution is that we provide the optimal privacy loss bound when $k-m$ of the $k$ $\diffp$-DP mechanisms are $\diffp$-BR and are non-adaptively selected in the homogenous privacy parameter setting, i.e. all privacy parameters are $\diffp$.  With these bounds, we can interpolate between the two extremes of only composing $\diffp$-DP mechanisms \cite{KairouzOhVi17, MurtaghVa16} and composing only $\diffp$-BR mechanisms \cite{DongDuRo19}.  Note that these bounds allow for worst case orderings of $\diffp$-DP and $\diffp$-BR mechanisms, as long as the ordering is selected in advance.

We then demonstrate that ordering between $\diffp$-BR and $\diffp$-DP mechanisms matters when the mechanisms are allowed to be adaptively selected.  Hence the privacy loss can differ between an analyst adaptively selecting exponential mechanisms after using Laplace mechanisms and an analyst that alternates between exponential mechanisms and Laplace mechanisms (as one would do with a top-$k$ DP system).

One omission from our consideration of the optimal privacy parameters of exponential mechanisms and general pure DP mechanisms is the Gaussian mechanism, which adds Gaussian noise to return noisy counts.  We address the Gaussian mechanism in Section~\ref{sec:GM}, which gives a comparison of the overall DP parameters when Laplace or Gaussian noise is added to histograms.  
We see that for reasonable parameter settings, they give roughly the same privacy loss, but when the number of privatized results grow, Gaussian noise gives smaller privacy loss.  Based on this, we then propose Gaussian based variants of mechanisms presented in LinkedIn's recently deployed privacy system \cite{RogersSuPeDuLeKaSaAh20}.
We see the inclusion of the Gaussian mechanism to the optimal DP composition bounds with pure DP and BR as a fruitful direction for future work.

\section{Preliminaries} 
We begin by defining differential privacy, which considers two \emph{neighboring} datasets $x,x'$ from some data universe $\cX$, i.e. $x$ is the same as $x'$ except one user's data has been removed or added, sometimes denoted as $x \sim x'$.  Note that DP is parameterized by the \emph{privacy loss} parameter $\diffp>0$ and a small probability of privacy failure $\delta \in [0,1]$.

\begin{definition}[\citet{DworkMcNiSm06}]
A randomized algorithm $M: \cX \to \cY$ that maps input set $\cX$ to some arbitrary outcome set $\cY$ is $(\diffp,\delta)$-differentially private (DP) if for any neighboring datasets $x,x'$ and outcome sets $S \subseteq \cY$,
\[ 
\Pr\left[ M(x) \in S\right] \leq e^\diffp \Pr\left[ M(x') \in S\right] + \delta.
\]
When $\delta = 0$, we typically say that $M$ is $\diffp$-DP or \emph{pure} DP.  
\end{definition}

One of the most useful properties of DP is that it composes, meaning that if one were to repeatedly use different DP algorithms, which can be adaptively selected at each round, then the result will be DP, although with a slightly worse privacy loss parameter.  A class of mechanisms that enjoys improved composition bounds are bounded range (BR) mechanisms, which is similar to pure DP  \cite{DurfeeRo19, DongDuRo19}.  Roughly, composing $\diffp$-BR mechanisms in a batch (i.e. non-adaptively) is nearly the same, in terms of the accumulated privacy loss, as composing $\diffp/2$-DP mechanisms.
\begin{definition}[\citet{DurfeeRo19}]
A randomized algorithm $M: \cX \to \cY$ that maps input set $\cX$ to some arbitrary outcome set $\cY$ is $\diffp$-bounded range (BR) if for any neighboring datasets $x,x'$ and outcomes $y_1,y_2 \in \cY$, 
\[ 
\frac{\Pr\left[ M(x) = y_1 \right]}{\Pr\left[ M(x') = y_1\right]} \leq e^\diffp \frac{\Pr\left[ M(x) = y_2 \right]}{\Pr\left[ M(x') = y_2\right] }
\]
\end{definition}

We then have the following connection between pure DP and BR.
\begin{lemma}[\citet{DurfeeRo19}]
If $M$ is $\diffp$-DP then it is $2 \diffp$-BR.  Alternatively, if $M$ is $\diffp$-BR then $M$ is $\diffp$-DP.  Furthermore, if $M$ is $\diffp$-BR then for each pair of neighbors $x,x'$, there exists a $t \defeq t(x,x')\in [0,\diffp]$ such that 
\[
t - \diffp \leq \log\left(\frac{\Pr\left[ M(x) = y \right]}{\Pr\left[ M(x') = y\right]}\right) \leq t.
\]
\end{lemma}

It turns out that one of the canonical mechanisms in DP, the exponential mechanism \cite{McSherryTa07}, with parameter $\diffp$ is $\diffp$-BR.  Thus, when working with the exponential mechanism, one should consider composition for BR mechanisms, rather than DP mechanisms. We then define the exponential mechanism, which provides a general way to construct DP mechanisms.  It takes a quality score $u: \cX \times \cY \to \R$ and samples an outcome $y$ based on the quality score $u(x,y)$ computed on the input data.  We will also need to define the \emph{range} of the quality score
\[
\tilde\Delta u \defeq \sup_{x \sim x'} \left\{ \max_{y \in \cY} \{ u(x,y) - u(x', y) \} - \min_{y' \in \cY}\{ u(x,y') - u(x',y') \} \right\}.
\]

The original definition is due to \citet{McSherryTa07}, but then it was recently modified with a slight refinement, based on the range of the quality score rather than the \emph{global sensitivity}, in \citet{DongDuRo19}.  
\begin{definition}[Exponential Mechanism \cite{McSherryTa07, DongDuRo19}]
For a given quality score $u: \cX \times \cY \to \R$ the \emph{exponential mechanism} $M_u: \cX \to \cY$ with input $x \in \cX$ returns $y \in \cY$ with probability proportional to 
\[ 
\exp\left(  \frac{\diffp u(x,y)}{\tilde\Delta u}\right).
\]
\end{definition}
\citet{McSherryTa07} showed that $M_u$ is $\diffp$-DP and \citet{DurfeeRo19} showed that it satisfies the stronger condition of $\diffp$-BR.   
In order to prove the optimal composition bounds for DP or BR mechanisms, it is sufficient to consider a simpler mechanism, based on randomized response \cite{Warner65}.  This mechanism takes a single bit and returns a bit, subject to DP or BR.  We provide the generalized version of randomized response below, which is from \citet{DongDuRo19}.
\begin{definition}[Generalized Random Response]\label{defn:gen_rr}
For any $\diffp \geq 0$ and $t \in [0,\diffp]$,
let $\grr{\diffp,t}: \{0,1\} \rightarrow \{0,1\}$ be a randomized mechanism in terms of probabilities $q_{\diffp,t}$ and $p_{\diffp,t}$ such that

\begin{align*}
\grr{\diffp,t}(0) = 0 \text{ w.p. } \frac{1 - e^{t-\diffp}}{1 - e^{-\diffp}} \eqdef q_{\diffp,t} \qquad & \text{ and } \qquad \grr{\diffp,t}(0) = 1 \text{ w.p. } \frac{e^{t-\diffp} - e^{-\diffp}}{1 - e^{-\diffp}} \eqdef 1 - q_{\diffp,t}
\\
\grr{\diffp,t}(1) = 0 \text{ w.p. } \frac{e^{-t} - e^{-\diffp}}{1 - e^{-\diffp}} \eqdef p_{\diffp,t} \qquad & \text{ and } \qquad \grr{\diffp,t}(1) = 1 \text{ w.p. } \frac{1 - e^{-t}}{1 - e^{-\diffp}} \eqdef 1 - p_{\diffp,t}.
\end{align*}
\end{definition}

Note that $\grr{2\diffp,\diffp}(\cdot)$ is the standard randomized response mechanism from DP and is $\diffp$-DP.  When considering the worst case optimal composition bounds of $k$ adaptively or nonadaptively chosen $\diffp$-DP mechanisms, we need only consider the overall privacy loss of $k$ repeated instances of $\grr{2\diffp,\diffp}(\cdot)$ \cite{KairouzOhVi17}.  For obtaining the worst case optimal composition bound of $\diffp$-BR mechanisms, it gets more complicated.  The main difficulty is in how each $t_i \in [0,\diffp]$ is selected at each round $i \in [k]$.  If all the $t_i$'s are preselected, corresponding to non-adaptively selecting $\diffp$-BR mechanisms in advance, then \citet{DongDuRo19} provide an efficiently computable formula for the optimal composition bound.  However, if each $t_i$ can be selected as a function of previous outcomes, i.e. the adaptive setting, then they provide a recursive formula for computing the optimal bound on the privacy loss, which is conjectured to be hard to compute, even in the homogenous privacy parameter setting.  

We will write $\mbr(\diffp)$ to denote the class of all $\diffp$-BR mechanisms, and similarly $\mdp(\diffp)$ will denote the class of all $\diffp$-DP mechanisms.  As discussed earlier, we will need to differentiate our bounds when the mechanisms can be adaptively selected at each round or not.  We then write $\cM_1 \times \cM_2 = \{ (M_1(\cdot), M_2(\cdot)) : M_1 \in \cM_1, M_2 \in \cM_2 \}$ as the non-adaptively selected class of mechanisms from $\cM_1$ and $\cM_2$. Alternatively, we will write $(\cM_1, \cM_2) = \{ M_1(\cdot), M_2(\cdot , M_1(\cdot))  : M_1 \in \cM_1, M_2 \in \cM_2   \}$ to denote the class of mechanisms that can be adaptively selected at each round, based on outcomes of previously selected mechanisms.

Since DP has two privacy parameters, we will typically fix $\diffp_g$ and write the best possible $\delta$ as a function of $\diffp_g$.  Hence, we will use the following definition for the optimal privacy parameter as a function of $\diffp_g$, which was also used in \cite{DongDuRo19}.  

\begin{definition}[Optimal Privacy Parameters]\label{def:optimal_delta}
Given a mechanism $M: \cX \to \cY$ and any $\diffp_g \in \R$, we define the optimal $\delta$ to be
\[
\delta_{\opt}(M;\diffp_g) \defeq \inf \big\{ \delta\geq 0: \text{ M is } (\diffp_g,\delta)\text{-DP}\big\}.
\]
Further, if $\cM$ is a class of mechanisms $M: \cX \to \cY$, then for any $\diffp_g \in \R$, we define 
\[
\delta_{\opt}(\cM;\diffp_g) \defeq \sup_{M \in \cM} \delta_{\opt}(M;\diffp_g).
\]
\end{definition}

We have the following formula for the optimal composition bound for homogenous, adaptively selected pure DP mechanisms.
\begin{theorem}[Optimal Homogeneous DP Composition \cite{KairouzOhVi17}]\label{thm:opt_DPCOMP}
For every $\diffp>0$ and $\diffp_g \geq 0$, we have the following where $\diffp_i = \diffp$ for all $i \in [k]$
\[
\deltaopt((\mdp(\diffp_1), \cdots, \mdp(\diffp_k)) ; \diffp_g) = \frac{1}{(1 + e^{\diffp})^k} \sum_{\ell = \left\lceil \tfrac{\diffp_g + k \diffp}{2 \diffp} \right\rceil }^k {k \choose \ell} \left( e^{\ell\diffp} - e^{\diffp_g + (k-\ell) \diffp} \right).
\]
\end{theorem}

We have the following result from Theorem 3 in \citet{DongDuRo19}, which results in improved DP composition bounds when compared with the optimal DP composition bounds for general $\diffp$-DP mechanisms, but only applies in the non-adaptive setting.
\begin{theorem}[Optimal Homogeneous Non-adaptive BR Composition \cite{DongDuRo19}]
\label{thm:non-adaptive}
Consider the case where over $k$ rounds, each mechanism is a non-adaptively selected $\diffp$-BR mechanism.  We set $t_\ell^* = \frac{\diffp_g + (\ell + 1) \diffp}{k + 1}$ where if $t_{\ell}^* \notin [0,\diffp]$, then we round it to the closest point in $[0,\diffp]$.  We then have the following, where we write $[y]_+ = \max\{ 0,y\}$, and write $\diffp_i = \diffp$ for $i \in [k]$
\[
\delta_{\opt}(\mbr(\diffp_1)\times \cdots \times \mbr(\diffp_k);\diffp_g) = \max_{ \ell \in \{0,\cdots,  k\}}  \sum_{i = 0}^k {k \choose i} \frac{(e^{-t_{\ell}^*} - e^{-\diffp})^{k-i}( 1- e^{-t_{\ell}^*} )^{i}}{(1 - e^{-\diffp})^k} \left[ e^{k t_{\ell}^*  - i\diffp} - e^{\diffp_g} \right]_+.
\]
\end{theorem}

We will also use concentrated differential privacy (CDP) \cite{DworkRo16} as a privacy definition, which is typically used to analyze mechanisms that add Gaussian noise to statistics.  Note that there are other variants related to CDP, including zero-mean CDP (zCDP) \cite{BunSt16} and R\'enyi DP (RDP) \cite{Mironov17}.  We first define the privacy loss random variable in terms of two random variables $Y$ and $Z$ over the same support
\[
L_{Y|| Z} \defeq \log\left(  \frac{\Pr[Y = y]}{\Pr[Z= y]}\right) \qquad \text{ where } y \sim Y.
\]
\begin{definition}[Concentrated Differential Privacy]
A randomized algorithm $M$ is $(\mu,\tau)$-CDP if for all neighboring inputs $x,x'$, we have $\E[L_{M(x)||M(x')}] \leq \mu$ and for any $\lambda \in \R$ we have
$
\E\left[\exp\left( \lambda \left(L_{M(x) || M(x') } - \E[L_{M(x)||M(x')}] \right) \right)\right] $$\leq $$e^{\lambda^2\cdot \tau^2/2}.
$
\end{definition}
Note that if a statistic $f: \cX \to \R$ can change by at most $\Delta_f$ on neighboring datasets, i.e. the sensitivity of $f$, then $M(x) = \Normal{f(x)}{\Delta_f^2 \sigma^2}$ satisfies $(\tfrac{1}{2\sigma^2}, \tfrac{1}{\sigma})$-CDP \cite{DworkRo16}.

We also provide the definition of zero-mean CDP (zCDP) from \citet{BunSt16}, which is useful for us in a later section due to its approximate version.  It is based on the R\'{e}nyi divergence of order $\alpha >1$ between two distributions $P$ and $Q$ over the same domain, denoted as $D_\alpha(P || Q)$ where
\[
D_{\alpha}(P||Q) \defeq \frac{1}{\alpha - 1} \log \E_{z \sim P} \left[ \left( \frac{P(z)}{Q(z)} \right)^{\alpha - 1}\right].
\]

\begin{definition}[Zero-mean Concentrated Differential Privacy]
A randomized algorithm $M: \cX \to \cY$ is $\delta$-approximately $(\xi,\rho)$-zCDP if for any neighbors $x, x' \in \cX$, there exists events $E$ and $E'$, such that $\Pr[E], \Pr[E'] \geq 1-\delta$ and for every $\alpha > 1$ we have the following bound in terms of the R\'enyi divergence $D_\alpha(\cdot || \cdot)$ of order $\alpha$
\[ 
D_\alpha(M(x)|_E || M(x')|_{E'} ) \leq \alpha \rho +\xi \qquad \text{ and } \qquad D_\alpha(M(x')|_{E'} || M(x)|_{E} ) \leq \alpha \rho + \xi.
\]
where $M(x)|_E$ is the distribution of $M(x)$ conditioned on event $E$ and similarly for $M(x') |_{E'}$,  If $\delta = 0$, then we say $M$ is $(\xi,\rho)$-zCDP and if $\xi = 0$ we write $\rho$-zCDP.
\label{defn:zCDP}
\end{definition}

We then have the following connection between CDP and zCDP.
\begin{lemma}[\citet{BunSt16}]\label{lem:CDPtozCDP}
If $M$ is $(\mu,\tau)$-CDP, then $M$ is also $(\mu - \tau^2/2, \tau^2/2)$-zCDP.
\end{lemma}
Note that we have the following connection between zCDP and DP.
\begin{lemma}[\citet{BunSt16}] 
If $M$ is $\rho$-zCDP then $M$ is also $(\rho + 2 \sqrt{\rho \log(1/\delta)} , \delta)$-DP for any $\delta >0$.
\end{lemma}

\section{Set-wise Adaptive Composition \label{sec:concentration}} 

We now present a general way to interact with a privacy system that allows the analyst lots of freedom in selecting different private mechanisms.  We unify the variants of DP (BR, DP, CDP) by analyzing with CDP, since it is the weakest of the three.  We allow the analyst to select CDP parameters from a set $\cE \defeq \{ (\mu_i, \tau_i) : i \in [k] \}$ at each round, but the order need not be predetermined. The analyst selects CDP parameters from $\cE$ without replacement so that after $k$ rounds each parameter has been used.   This extra freedom to the analyst was not considered in previous composition bounds and stops short of the complete freedom of adaptive parameter selection in \emph{privacy odometers} \cite{RogersRoUlVa16}.  Allowing an analyst to adaptively select the order of parameters arises naturally in practice.  Consider a setting where an analyst does heavy hitter selection (with BR mechanisms) and then decides whether they want counts as well with Laplace noise (DP), with Gaussian noise (CDP) to do some post processing (see later section that describes such a post processing mechanism), or simply a new heavy hitter query. 

We now detail the experiment between the analyst and the private mechanisms.  We start with the analyst selecting CDP parameters $(\mu_1,\tau_1)$ and then a CDP mechanism with those parameters, while also updating $\cE \gets \cE \setminus \{(\mu_1,\tau_1)\}$.  Similar to other works on privacy loss composition, the analyst may then select neighboring datasets $x_1^{(0)}, x_1^{(1)}$.  Once the analyst sees the outcome from the selected mechanism evaluated on the (unknown) $b \in \{0,1\}$ dataset, the analyst gets to adaptively select CDP parameters from the remaining set $\cE$, a corresponding mechanism, and neighboring datasets $x_2^{(0)}, x_2^{(1)}$.  We then delete the selected privacy parameters and continue.  The interaction proceeds until $\cE = \emptyset$.  The information that is hidden from the analyst is the bit $b \in \{0,1 \}$ that determines which of the neighboring datasets is used at each round.

We refer to this protocol between the analyst and the privacy system as the \emph{$\cE$-set-wise adaptive composition experiment $b$}.  As in \citet{DworkRoVa10}, we refer to $V^{(b)}$ as the \emph{view} of the analyst in the $\cE$-set-wise adaptive composition experiment $b$.  Hence, we want to be able to show that the two experiments with $b$ and $1-b$ are similar.\footnote{As presented, we do not allow for the analyst to randomize over different choices of mechanisms at each round in the experiment.  Hence, the analyst deterministically selects the next mechanism based on the previous outcomes.  However, the differential privacy guarantees for the experiment would not change if we allow the analyst to randomize over choices of mechanisms at each round, as was shown in Lemma~4.3 in \cite{DongDuRo19}}  

\begin{definition}
 We say that $\cE$ is $(\diffp_g,\delta)$-\emph{differentially private under set-wise adaptive composition} if for any outcome set $S$ for the $\cE$-set-wise adaptive composition experiment $b$, we have $\Pr[V^{(b)} \in S ] \leq e^{\diffp_g} \Pr[V^{(1-b)} \in S] + \delta$.
\end{definition}

In order to provide a privacy bound, we start with a general result that follows from \citet{HowardRaMcSe18}, but we include a proof in Appendix~\ref{app:ConcentrationProof} for completeness.  
\begin{lemma}\label{lem:concentration}
Let $(\Omega, \cF, \Pr)$ be a probability triple where $\emptyset \defeq \cF_0\subseteq \cF_1, \subseteq \cdots \subseteq \cF$ be an increasing sequence of $\sigma$-algebras.  Let $X_i$ be a real-valued $\cF_i$-measurable random variable 
and $X_0 = 0$.  Assume that there exists random variables $B_i$ that are $\cF_{i-1}$ measurable for $i \geq 1$ and $B_0 = 0$ such that $\sum_{i = 0}^k B_i^2 \leq b^2$ a.s. for some constant $b$.  If for all $\lambda > 0$, we have
$
\E[ e^{\lambda X_i} \mid \cF_{i-1}] \leq e^{\lambda^2 B_i^2/2 } \text{ a.s. } \forall i,
$
then we have for all $\beta > 0 $
\[
\Pr\left[\sum_{ i = 1}^k  X_i \geq \beta \right] \leq e^{\tfrac{-\beta^2}{2b^2}}.
\]
\end{lemma}

Many concentration bounds rely on a subgaussian bound.  This particular result allows us to get a subgaussian bound on the sum of mean zero terms when we know a bound on the sum of subgaussian parameters for each term and each individual subgaussian parameter need not be a fixed constant.  This allows the analyst to adaptively select different classes of mechanisms at each round, where the subgaussian parameter changes based on the class that is selected.

We then define the privacy loss random variable, which we will bound with high probability to get a bound on the overall privacy loss.  At each round $i$, the analyst has all the necessary information from the previous outcomes to decide on which neighboring datasets to use, which class of mechanism to select, which of the remaining privacy parameters to choose, and which specific mechanism $M_i$ to pick.  The previous outcomes and choices of the analyst up to round $i$ are random variables from the sigma algebra $\cF_{i-1}$. The selected randomized mechanism $M_i$ takes the input to some outcome set $\cY_i$ and for each $y_i \in \cY_i$ we define the following for neighboring datasets $x,x'$
\[ 
L_i(y_i) \defeq \log\left( \frac{\Pr[M_i(x) = y_i \mid \cF_{i-1}]}{\Pr[M_i(x') = y_i \mid \cF_{i-1}]} \right) \qquad L_i \defeq L_i(Y_i) \text{ where } Y_i \sim M_i(x) \mid \cF_{i-1}.
\]
We then consider the full privacy loss over the entire $\cE$-set-wise adaptive composition experiment, $\sum_{i=1}^k L_i$.  Similar to \citet{DworkRoVa10}, we aim to bound the accumulated privacy loss with high probability, i.e. $\Pr[ \sum_{i=1}^k L_i \geq \diffp_g ] \leq \delta$, so that the $\cE$-set-wise adaptive composition experiment is $(\diffp_g, \delta)$-DP.

We then want to use Lemma~\ref{lem:concentration}, so we define the following random variables 
\begin{equation}
X_i = L_i - \E[L_i \mid \cF_{i-1}].
\label{eq:Xmartingale}
\end{equation}
Next, we need bounds on $\E[L_i \mid \cF_{i-1}]$ and on the subgaussian parameter for each $X_i$.  

\begin{lemma}\label{lem:subgaussBounds}
For $X_i$ given in \eqref{eq:Xmartingale}, if the analyst selects an $\diffp$-DP mechanism $M_i$ given $\cF_{i-1}$, then for all $\lambda \in \R$ we have,
\[
\E[\exp(\lambda X_i) \mid \cF_{i-1}] \leq e^{\lambda^2 \diffp^2 /2} \qquad \& \qquad \E[L_i \mid \cF_{i-1} ] \leq \diffp \left( \frac{e^{\diffp} - 1}{e^{\diffp} + 1} \right).
\]
If the analyst selects an $\alpha$-BR mechanism $\cM_i$ given $\cF_{i-1}$ then for all $\lambda \in \R$ we have

\[
\E[\exp(\lambda X_i) \mid \cF_{i-1}] \leq e^{\lambda^2 \alpha^2 /8} \qquad \& \qquad \E[L_i \mid \cF_{i-1} ] \leq \frac{\alpha}{e^\alpha - 1} - 1 - \log\left( \frac{\alpha}{e^\alpha - 1} \right) .
\]

If the analyst selects a $(\mu,\tau)$-CDP mechanism $M_i$ given $\cF_{i-1}$, then for all $\lambda \in \R$, we have,
\[
\E[\exp(\lambda X_i) \mid \cF_{i-1}] \leq e^{\lambda^2 \tau^2 /2} \qquad \& \qquad \E[L_i \mid \cF_{i-1} ] \leq \mu.
\]
\end{lemma}
\begin{proof}
The statement about the CDP mechanism is simply due to the definition of $(\mu,\tau)$-CDP \cite{DworkRo16}.   The expectations of the privacy losses for $\diffp$-DP and $\alpha$-BR are from \cite{KairouzOhVi17} and \cite{DongDuRo19}, respectively.  From \citet{Hoeffding63}, we know that for bounded random variables $X_i \in [a,b]$ that $X$ is subgaussian with parameter $(b-a)^2/4$.  Hence, for $\diffp$-DP, we have $b- a \leq 2 \diffp$ and for $\alpha$-BR, we have $b-a \leq \alpha$.
\end{proof}

Note that the \emph{advanced composition} bound from \cite{DworkRoVa10} uses Azuma's inequality for a concentration bound on the privacy loss and \cite{DurfeeRo19} uses the more general Azuma-Hoeffding bound.  Here, we use Lemma~\ref{lem:concentration} to get a bound on the privacy loss in a more general setting.

\begin{lemma}\label{lem:genbound}
Let $\cE = \{ (\mu_i,\tau_i ) : i \in [k] \} \}$.  The set $\cE$ is $(\diffp_g, \delta)$-differentially private under set-wise adaptive composition for any $\delta >0$, where
\begin{equation}
\diffp_g =  \sum_{i \in [k] }\mu_i + \sqrt{ 2 \sum_{i \in [k]} \tau_i^2  \log(1/\delta)}.
\label{eq:setWiseCompEps}
\end{equation}
\end{lemma}

Consider the setting where the analyst is allowed to select $m_{\texttt{DP}}$ pure-DP mechanisms, $m_{\texttt{BR}}$ BR mechanisms, and $m_{\texttt{CDP}}$ concentrated DP mechanisms \cite{DworkRo16}.  Furthermore, there are preregistered privacy parameters $\cE^{\texttt{DP}} = \{\diffp_i = \diffp >0 : i \in [m_{\texttt{DP}}] \}$ for pure DP mechanisms, $\cE^{\texttt{BR}} = \{\alpha_i = \alpha>0 : i \in [m_{\texttt{BR}}] \}$ for BR mechanisms, and $\cE^{\texttt{CDP}} = \{(\mu_i,\tau_i) = (\mu,\tau) > (0,0) : i \in [m_{\texttt{CDP}}]\}$ for CDP mechanisms.  We allow the analyst to adaptively select the class of mechanism at each round $i$ adaptively and to also select the privacy parameter from the corresponding class of mechanisms, without replacement.  We then apply Lemma~\ref{lem:genbound} and Lemma~\ref{lem:concentration} to replace the formula in \eqref{eq:setWiseCompEps} to get
\begin{align}
\diffp_g & =  m_{\texttt{DP}}\diffp \left( \frac{e^\diffp - 1}{e^\diffp + 1} \right) + m_{\texttt{BR}} \left(\frac{\alpha}{e^\alpha - 1} - 1 - \log\left( \frac{\alpha}{e^\alpha - 1} \right) \right) + m_{\texttt{CDP}} \mu \nonumber \\
& \qquad + \sqrt{ 2 \left(  m_{\texttt{DP}} \diffp^2 + \frac{m_{\texttt{BR}}}{4}\alpha^2 + m_{\texttt{CDP}} \tau^2 \right) \log(1/\delta)}.
\label{eq:setWiseCompEps2}
\end{align}

Note that when $m_{\texttt{DP}} = k$ , we get the traditional \emph{advanced} composition bound from \citet{DworkRoVa10} with refinement from \citet{KairouzOhVi17} and when $m_{\texttt{BR}} = k$, we get the bound for composing $\diffp$-BR mechanisms from Corollary 3.1 in \citet{DongDuRo19}.  In Figure~\ref{fig:AdvancedCompBRDP}, we present the family of curves from \eqref{eq:setWiseCompEps2} for different values of $m=m_{\texttt{DP}}$, i.e. the number of $\diffp$-DP mechanisms, $k-m = m_{\texttt{BR}}$ while fixing $m_{\texttt{CDP}} = 0$, and compare it with the optimal DP composition bound from Theorem~\ref{thm:opt_DPCOMP}.  Note that applying the optimal DP bound is almost the same as using the given formula when half are BR mechanisms and the remaining half are DP.

\begin{figure}[h]
\centering
\includegraphics[width=0.48\textwidth]{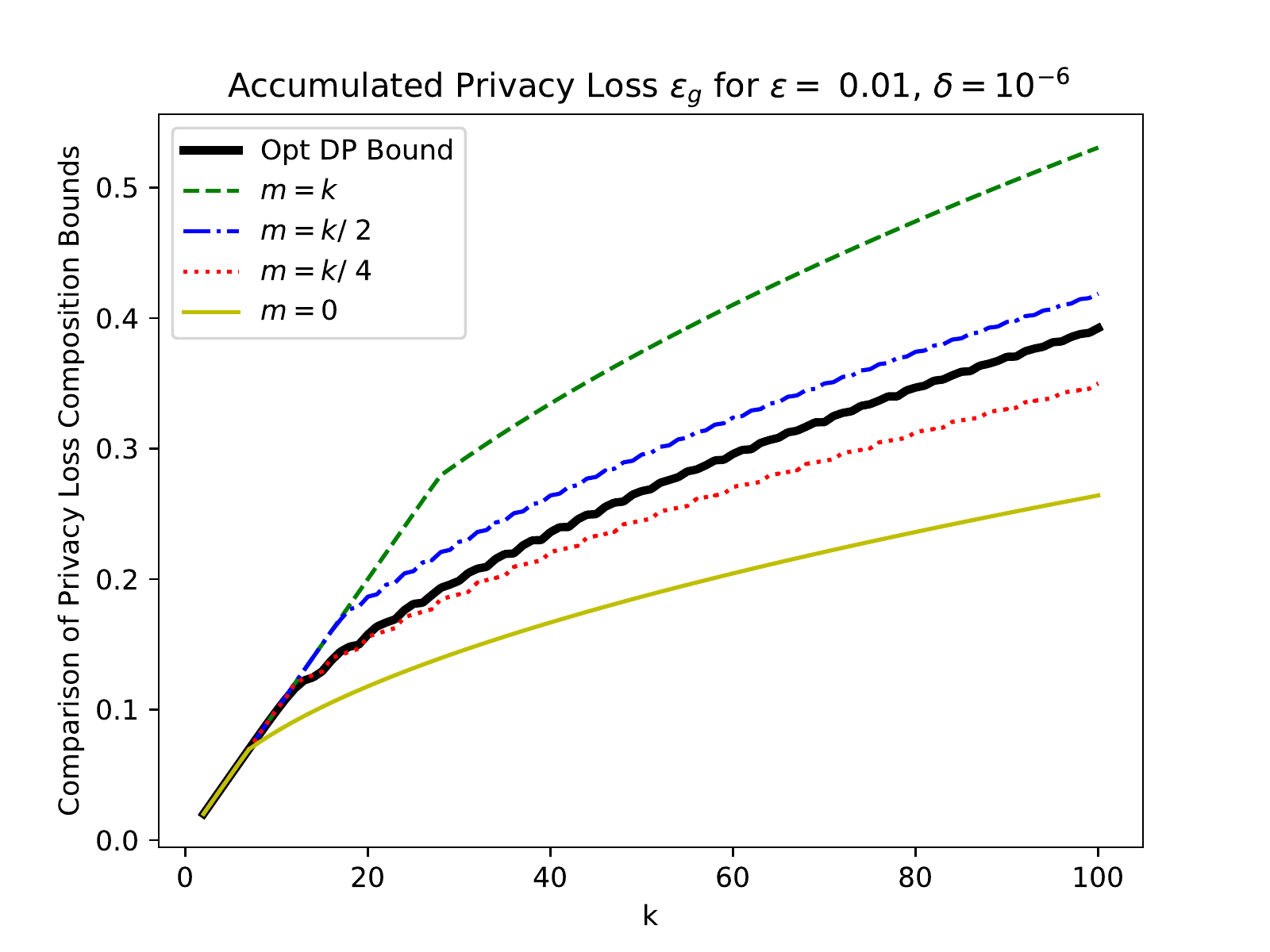}
\includegraphics[width=0.48\textwidth]{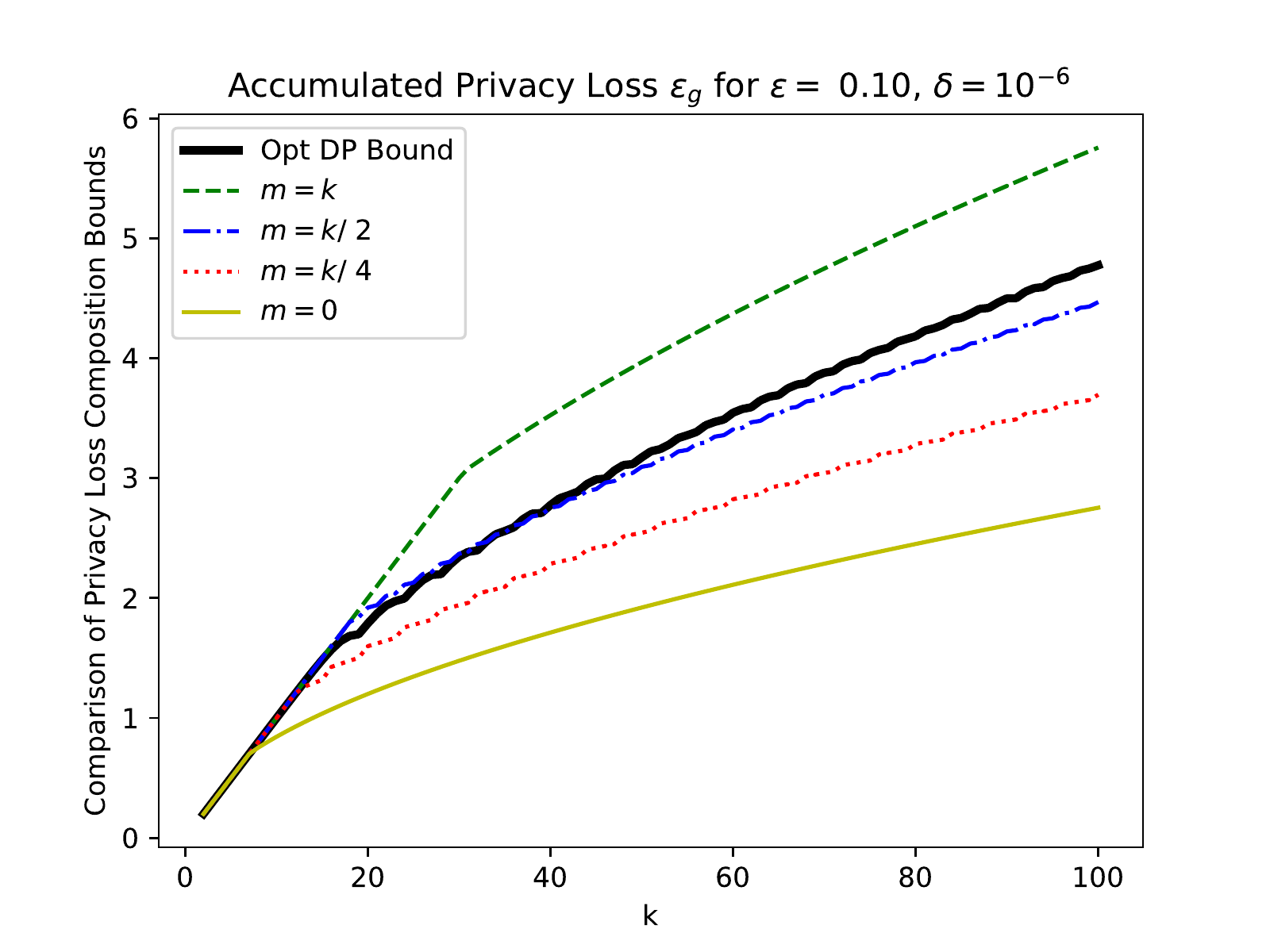}
\caption{We plot the bounds on $\diffp_g$ from \eqref{eq:setWiseCompEps2} for $m$ many $\diffp$-DP mechanisms, of which $k-m$ are $\diffp$-BR.  We also give the comparison to the optimal DP bound from \citet{KairouzOhVi17}.  Note that these bounds apply even when mechanisms can be adaptively selected at each round.  \label{fig:AdvancedCompBRDP}}
\end{figure}

We can also do a similar analysis when all the mechanisms are zCDP (Definition~\ref{defn:zCDP}), rather than CDP, which will be useful when we use the approximate ($\delta>0$) version of zCDP in Section~\ref{sec:GM}.  We modify the above experiment by using the set $\cE_{zCDP} = \{(\delta_i,\xi_i,\rho_i): i \in [k] \}$, to denote the approximate zCDP parameters that the analyst can select from at each round and then, with the selected parameters, can choose a zCDP mechanism (as opposed to CDP as above). We have the following result which again uses Lemma~\ref{lem:concentration}.
\begin{lemma}\label{lem:zCDPgenbound}
Let $\cE_{zCDP} = \{ (\delta_i,\xi_i,\rho_i ) : i \in [k] \} \}$.  The set $\cE_{zCDP}$ is $(\diffp_g, \delta + \sum_{i=1}^k\delta_i)$-differentially private under set-wise adaptive composition, where zCDP mechanisms are selected at each round, for any $\delta >0$, where
\begin{align*}
\diffp_g & =  \sum_{i \in [k] } (\xi_i + \rho_i) + 2\sqrt{  \sum_{i \in [k]} \rho_i  \log(1/\delta)}.
\end{align*}
\end{lemma}
\begin{proof}
Note that the result here is analogous to Lemma~\ref{lem:genbound} except we do not subtract off the mean of the privacy loss to form $X_i$, yet we use the same filtration $\{\cF_i\}$.  Consider the privacy loss $L_i^{(b)}$ at each round $i$ conditioned on the events $E_i^{(0)}, E_i^{(1)}$ where $\Pr[E_i^{(b)}] \geq 1-\delta_i$ for $b\in \{0,1\}$ and each round $i \in [k]$.  
\[
L_i^{(b)}(y_i) \defeq \log\left( \frac{\Pr[M_i(x_i^{(b)}) = y_i \mid E_i^{(b)}, \cF_{i-1}]}{\Pr[M_i(x_i^{(1-b)} = y_i \mid E_i^{(1-b)},\cF_{i-1}]} \right), \  L_i^{(b)} \defeq L_i^{(b)}(Y_i) \text{ where } Y_i \sim M_i(x_i^{(b)}) \mid E_i^{(b)},\cF_{i-1}.
\]
We use Lemma~\ref{lem:concentration} with $X_i = L_i^{(b)} - \xi_i - \rho_i$ and $B_i^2 = 2 \rho_i$.  From the definition of zCDP we have the following for any $\lambda = \alpha - 1 > 0$
\begin{align*}
& \E[\exp\left(\lambda (L_i^{(b)} - \xi_i - \rho_i)\right) \mid \cF_{i-1}] \leq \exp\left(\lambda^2\rho_i \right) \\
\implies & \E\left[\exp\left(\lambda \sum_{i=1}^k(L_i^{(b)} - \xi_i - \rho_i) - \lambda^2 \sum_{i=1}^k \rho_i \right) \right] \leq 1.
\end{align*}
Recall that at each round $i \in [k]$ the index of the parameters  is a random variable, based on the outcomes of previous results, but the resulting sum of parameters is predetermined.  Hence, we have 
\[
\E\left[\exp\left(\lambda \sum_{i=1}^k(L_i^{(b)} - \xi_i - \rho_i) \right) \right] \leq \exp\left( \lambda^2 \sum_{i=1}^k \rho_i\right).
\]
The rest of the analysis is identical to the proof of Lemma 3.5 from \citet{BunSt16}, which gives the conversion of zCDP to approximate DP.  Hence, we have
\[
\Pr[V^{(b)} \in S \mid \cup_{i=1}^k E_i^{(b)} ] \leq e^{\diffp_g} \Pr[V^{(1-b)} \in S\mid \cup_{i=1}^k E_i^{(1-b)}] + \delta
\]
We then follow Lemma~8.8 from \cite{BunSt16}.  Without loss of generality we set $\Pr[\cup_i E_i^{(b)}] = 1-\delta'$ and $\Pr[\cup_i E_i^{(1-b)}] \geq 1 - \delta'$ where $\delta' \leq \sum_i \delta_i$.  We then have
\[
\Pr[V^{(b)} \in S] \leq \Pr[V^{(b)} \in S \mid \cup_{i=1}^k E_i^{(b)} ] (1-\delta')+\delta'  
\]
and
\[
\Pr[V^{(1-b)} \in S] \geq \Pr[V^{(1-b)} \in S \mid \cup_{i=1}^k E_i^{(1-b)} ] (1-\delta')  
\]
Putting everything together, we have
\begin{align*}
 \Pr[V^{(b)} \in S] & \leq  \Pr[V^{(b)} \in S \mid \cup_{i=1}^k E_i^{(b)} ] (1-\delta')+\delta'   \\
& \leq \left(e^{\diffp_g}\Pr[V^{(1-b)} \in S \mid \cup_{i=1}^k E_i^{(1-b)} ]+\delta\right) (1-\delta')+\delta' \\
 & \leq e^{\diffp_g} \Pr[V^{(1-b)} \in S] + \delta(1-\delta') + \delta' \\
 & \leq e^{\diffp_g} \Pr[V^{(1-b)} \in S] + \delta +  \sum_{i=1}^k \delta_i.
\end{align*}
\end{proof}

Recall that with Lemma~\ref{lem:CDPtozCDP}, we can convert CDP parameters into zCDP parameters.  Hence, \eqref{eq:setWiseCompEps2} can also be derived from Lemma~\ref{lem:zCDPgenbound} and Lemma~\ref{lem:subgaussBounds}.  From Lemma~\ref{lem:subgaussBounds}, we see that $\diffp$-DP implies $( 0, \diffp^2/2)$-zCDP and $\alpha$-BR implies $(\xi, \alpha^2/8)$-zCDP, where $\xi = \frac{\alpha}{e^\alpha - 1} - 1 - \log\left( \frac{\alpha}{e^\alpha - 1} \right)  - \alpha^2/8\leq 0$.


\section{Optimal Non-adaptive Composition Bounds \label{sec:optimal_nonadaptive_omitted_proofs}} 
In this section we will present optimal bounds on the privacy loss when we combine both $\diffp$-BR and $\diffp$-DP mechanisms and the BR mechanisms are all preselected, prior to any interaction.  Note that our analysis does allow for an analyst to select any order of $\diffp$-BR and $\diffp$-DP mechanisms in advance, and the $\diffp$-DP mechanisms can be adaptively selected, but the $\diffp$-BR mechanisms cannot.  

We now want to generalize Theorem~\ref{thm:non-adaptive} when $m$ of the $k$ mechanisms can be $\diffp$-DP, rather than all being $\diffp$-BR.  
We will proceed in a similar way to the analysis in Section 5 of \citet{DongDuRo19}.  We start with the following result that allows for heterogeneous $\diffp_i$ at each $i \in [k]$.  

\begin{lemma}[Lemma 3.1 in \citet{DongDuRo19}]\label{lem:non_interactive_GRR}
Recall from Definition~\ref{defn:gen_rr}
we have $p_{\diffp_i, t_i}, q_{\diffp_i,t_i}$.  We then have
\begin{align*}
\delta_{\opt}( \mbr(\diffp_1)\times \cdots \times \mbr(\diffp_k);\diffp_g) & \\
= \sup_{\bbt \in \prod_{i \in [k]} [0,\diffp_i]} \sum_{S \subseteq \{1,...,k\}} & \left[\prod_{i \notin S} q_{\diffp_i,t_i} \prod_{i \in S} (1 - q_{\diffp_i,t_i}) - e^{\diffp_g}\prod_{i \notin S}p_{\diffp_i,t_i} \prod_{i \in S} (1 - p_{\diffp_i,t_i}) \right]_+.
\end{align*}
\end{lemma}

We then define the elements that we take the $\sup$ over for general $(t_1,\cdots, t_k)$ and corresponding privacy parameters $(\diffp_1, \cdots, \diffp_k)$ 
\begin{equation}
\delta( (t_1,\diffp_1) \times \cdots \times (t_k,\diffp_k) ; \diffp_g) \defeq  \sum_{S \subseteq \{1,...,k\}}  \left[\prod_{i \notin S} q_{\diffp_i,t_i} \prod_{i \in S} (1 - q_{\diffp_i,t_i}) - e^{\diffp_g}\prod_{i \notin S}p_{\diffp_i,t_i} \prod_{i \in S} (1 - p_{\diffp_i,t_i}) \right]_+.
\label{eq:gen_delta_t}
\end{equation}

We then have the following result that shows that the ordering of mechanisms does not modify the value of $\delta_\opt$.
\begin{lemma}\label{lem:ordering_dont_mattter}
Let $\pi: [k] \to [k]$ be a permutation on the indices $[k]$, we then have
\[
\delta( (t_1,\diffp_1) \times \cdots \times (t_k,\diffp_k) ; \diffp_g) = \delta( (t_{\pi(1)},\diffp_{\pi(1)}), \cdots, (t_{\pi(k)},\diffp_{\pi(k)}) ; \diffp_g)   
\]
\end{lemma}
\begin{proof}
Note that the expression in \eqref{eq:gen_delta_t} takes a summation over all possible subsets of indices.  Hence, if we permute the indices, the full summation has the same terms.  
\end{proof}

We will focus only on the homogeneous case, where all the privacy parameters are the same and leave the heterogeneous case to future work.  We have the immediate result from Lemma~\ref{lem:ordering_dont_mattter}.
\begin{lemma}
Let $\vec{\cM}, \vec{\cM'}$ be two sequences of non-adaptively selected mechanisms where $m$ are $\mdp(\diffp)$ and $k-m$ are from $\mbr(\diffp)$.  We then have for any $\diffp_g> 0$ and $\diffp>0$, 
$
\delta_{\opt}( \vec{\cM};\diffp_g)  = \delta_{\opt}(\vec{\cM'} ;\diffp_g) .
$
\end{lemma}

Recall that the big difference between $\diffp$-BR and $\diffp$-DP mechanisms, is that $\diffp$-BR mechanisms use the generalized randomized response $\grr{\diffp,t}$ for a worst case $t \in [0,\diffp]$ as opposed to $\diffp$-DP mechanisms which use $\grr{2\diffp,\diffp}$, so that $t = \diffp$.  We then define the following function that fixes $m$ of the $t_i$ values to be $\diffp$ which have corresponding $\diffp_i = 2 \diffp$ and corresponding sets $\cS_1 \defeq [k-m]$ and $\cS_2 \defeq \{ k-m+1, \cdots, k \}$.
\begin{align*}
\delta((t_{1},\diffp), &\cdots, (t_{k-m} ,\diffp), (\diffp,2\diffp),\cdots, (\diffp,2\diffp) ; \diffp_g) \\
& = \sum_{S_1 \subseteq \cS_1} \sum_{S_2 \subseteq \cS_2} \left[\prod_{i_1 \in \cS_1 \setminus S_1} q_{\diffp,t_{i_1}} \prod_{i_2 \in \cS_2 \setminus S_2} q_{2\diffp,\diffp} \prod_{j_1 \in S_1} (1 - q_{\diffp,t_{j_1}}) \prod_{j_2 \in S_2} (1 - q_{2\diffp,\diffp}) \right.\\
& \qquad\qquad\qquad\qquad\qquad\qquad\left. - e^{\diffp_g}\prod_{i_1 \in \cS_1 \setminus S_1} p_{\diffp,t_{i_1}} \prod_{i_2 \in \cS_2 \setminus S_2} p_{2\diffp,\diffp} \prod_{j_1 \in S_1} (1 - p_{\diffp,t_{j_1}}) \prod_{j_2 \in S_2} (1 - p_{2\diffp,\diffp}) \right]_+ \\
& = \sum_{S_1 \subseteq \cS_1} \sum_{S_2 \subseteq \cS_2} \left[q_{2\diffp,\diffp}^{m - |S_2|}  (1 - q_{2\diffp,\diffp})^{|S_2|} \prod_{i_1 \in \cS_1 \setminus S_1} q_{\diffp,t_{i_1}}  \prod_{j_1 \in S_1} (1 - q_{\diffp,t_{j_1}})\right.\\
& \qquad\qquad\qquad\qquad\qquad\qquad\left. - e^{\diffp_g}p_{2\diffp,\diffp}^{m - |S_2|}(1 - p_{2\diffp,\diffp})^{|S_2|}\prod_{i_1 \in \cS_1 \setminus S_1} p_{\diffp,t_{i_1}} \prod_{j_1 \in S_1} (1 - p_{\diffp,t_{j_1}})\right]_+ \\
\end{align*}

Note that we have the simple connection with this function and the optimal $\delta$,
\begin{align*}
\delta_{\opt}& \left( \underbrace{ \mbr(\diffp)\times \cdots \times \mbr(\diffp)}_{k-m} \times  \underbrace{\mdp(\diffp)\times \cdots\times \mdp(\diffp)}_{m};\diffp_g\right) \\
& \qquad =  \sup_{\bbt \in \prod_{i \in \{1, \cdots, k-m \}} [0,\diffp]}\delta((t_{1},\diffp), \cdots, (t_{k-m},\diffp) , (\diffp,2\diffp),\cdots, (\diffp,2\diffp)) 
\end{align*}

We then use a key result from \citet{DongDuRo19} that shows that despite BR mechanisms having separate $t_i \in [0,\diffp]$ for each BR mechanism, the worst way to set the $t_i$'s is to set them all equal. The following is a slight generalization of Lemma~5.4 from \citet{DongDuRo19}.
\begin{lemma}\label{lem:gen_convexity} 
For any $\diffp,\diffp' > 0$, $\diffp_g \in \R$, and $\bbt \in [0,\diffp]^{k-m} \times [0,\diffp']^m$,
\begin{align*}
& \delta((t_1,\diffp), (t_2,\diffp), \cdots (t_{k-m}, \diffp), (t_{k-m+1} , \diffp'),\cdots, (t_k,\diffp');\diffp_g) \\
& \quad \leq \delta\left( \left(\frac{t_1 + t_2}{2},\diffp \right), \left(\frac{t_1 + t_2}{2},\diffp \right) , (t_3,\diffp), \cdots,  (t_{k-m}, \diffp), (t_{k-m+1} , \diffp'),\cdots, (t_k,\diffp');\diffp_g\right)
\end{align*}
\end{lemma}
Hence, we can simplify our expression for $\delta_{\opt}$ substantially to the following where we use the fact that $p_{\diffp,t} = e^{-t} q_{\diffp,t}$ and $1 - p_{\diffp,t} = e^{\diffp-t} ( 1 - q_{\diffp,t} ) $
\begin{align}
& \delta_{\opt} (\left( \mbr(\diffp), \cdots, \mbr(\diffp), \mdp(\diffp), \cdots, \mdp(\diffp) \right);\diffp_g) \\
& = \sup_{t \in [0,\diffp]} \sum_{S_1 \subseteq \cS_1} \sum_{S_2 \subseteq \cS_2} \left[q_{2\diffp,\diffp}^{m - |S_2|}  (1 - q_{2\diffp,\diffp})^{|S_2|} q_{\diffp,t}^{k-m - |S_1|}  (1 - q_{\diffp,t})^{|S_1|}\right. \nonumber\\
& \qquad\qquad\qquad\qquad\qquad\qquad\left. - e^{\diffp_g}p_{2\diffp,\diffp}^{m - |S_2|}(1 - p_{2\diffp,\diffp})^{|S_2|}p_{\diffp,t}^{k - m - |S_1|}  (1 - p_{\diffp,t})^{|S_1|}\right]_+ \nonumber \\
& = \sup_{t \in [0,\diffp]} \sum_{i = 0}^{k-m} \sum_{j = 0}^m {k-m \choose i} {m \choose j} q_{2\diffp,\diffp}^{m - j}(1 - q_{2\diffp,\diffp})^{j}q_{\diffp,t}^{k - m - i}  (1 - q_{\diffp,t})^{i} \left[ 1- e^{\diffp_g - \diffp ( m - 2j - i) - t (k - m )}\right]_+. \label{eq:nonadapt1}
\end{align}

We now want to show that the $\sup_{t \in [0,\diffp]}$ can be decomposed into a max over a finite number of values for $t \in [0,\diffp]$.  We define the following function for $t \in [0,\diffp]$:
\begin{equation}\label{eq:delta_mk}
\delta^{m,k}(t; \diffp_g) \defeq \delta(\underbrace{(t,\diffp), \cdots, (t,\diffp)}_{k-m}, \underbrace{(\diffp,2\diffp), \cdots, (\diffp,2\diffp)}_{m} ; \diffp_g). \nonumber
\end{equation}

We define the following function in terms of $\alpha_{i,j} = {k-m \choose i} { m \choose j } q_{2\diffp,\diffp}^{m-j} (1 - q_{2\diffp, \diffp})^j$ when $i \in \{0,\cdots k - m \}$ and $j \in \{0,\cdots, m \}$ with $\alpha_{i,j} = 0$ otherwise.
\[ 
F_\ell(t) \defeq \sum_{n = 0}^\ell\left( 1 - e^{\diffp_g - \diffp(m - n) - t (k-m) } \right) \sum_{ \substack{i + 2j = n \\ i \in \{0,1 \cdots, k-m \} \\ j \in \{0, 1, \cdots, m \} } } \alpha_{i,j} q_{\diffp,t}^{k-m - i} ( 1- q_{\diffp,t})^j
\]

The following result provides a more general version of Lemma 5.7 in \cite{DongDuRo19}, which is central to showing that we need to only consider a few values of $t \in [0,\diffp]$.
\begin{lemma}\label{lem:deriv_zero}
\[
F'_\ell(t) = \frac{1}{1 - e^{-\diffp}} \left( e^{\diffp_g - \diffp(m -\ell) - t(k-m)} - e^{t-\diffp} \right) \sum_{ \substack{i + 2j = \ell \\ i \in \{0,1 \cdots, k-m \} \\ j \in \{0, 1, \cdots, m \} } }  \alpha_{i,j} (k-m - i)  q_{\diffp,t}^{k-m - i - 1} (1 - q_{\diffp,t})^i 
\]
\end{lemma}

We relegate details of the proof to the appendix (\ref{lem:deriv_zero:proof}).

Given this result, we can limit our search for $\sup_{t \in [0,\diffp]}$ to a $\max$ over $k+m$ terms.  We then have the main result of this section, which presents the formula for computing the optimal privacy bound of $k$ non-adaptively selected $\diffp$-DP mechanisms, where $k-m$ of which are $\diffp$-BR.   
\begin{theorem}\label{thm:opt_nonadaptive}
Consider the non-adaptive sequence of $m$ many $\mdp(\diffp)$ mechanisms and $k-m$ many $\mbr(\diffp)$ mechanisms.  We define $t_{\ell} = \frac{\diffp_g + \diffp(\ell+1 -m )}{k-m+1}$ if $t_\ell \in [0,\diffp]$, otherwise we round it to the closest point in $\{0,\diffp \}$, for $\ell \in \{0,\cdots, k+m  \} $.  Then for $\diffp,\diffp_g >0$ we have the following formula for computing $\delta_{\opt} \left( \mbr(\diffp_1) \times \cdots \times \mbr(\diffp_{k-m}) \times  \mdp(\diffp_{k-m + 1}) \times \cdots \times \mdp(\diffp_k);\diffp_g\right) $ with $\diffp_i = \diffp$ for $i \in [k]$.
\begin{align*}
\max_{\ell \in \{0,\cdots, k+m \}} \sum_{i=0}^{k-m}\sum_{j = 0}^m {k-m \choose i} {m \choose j} q_{2\diffp,\diffp}^{m - j}(1 - q_{2\diffp,\diffp})^{j}q_{\diffp,t_\ell}^{k - m - i}  (1 - q_{\diffp,t_\ell})^{i} \left[ 1- e^{\diffp_g - \diffp ( m - 2j - i) - t_\ell (k - m )}\right]_+
\end{align*}
Furthermore, this bound applies when the ordering of mechanisms can be adversarially chosen, prior to any interaction.  
\end{theorem}

Note that when $m = k$ we recover the optimal privacy loss bounds for DP mechanisms from \cite{KairouzOhVi17, MurtaghVa16}, and the expression in Theorem~\ref{thm:opt_nonadaptive} becomes independent of $t_{\ell}$. Hence, we do not need to do a max over $2k+1$ terms, and it can be computed in $O(k)$ time.  When $m = 0$, we recover the expression from \citet{DongDuRo19} and it can be computed in $O(k^2)$ time.  In the case when $m= \Theta(k)$ our formula can take $O(k^3)$ time to calculate.

We present the family of curves for various values of $m$ in Figure~\ref{fig:opt_nonadaptive}.  These bounds allow us to interpolate between the two previous optimal composition bounds for $\diffp$-DP mechanisms from \citet{KairouzOhVi17} (setting $m = k$) and $\diffp$-BR mechanisms from \citet{DongDuRo19} (setting $m= 0$). 
\begin{figure}[h]
\centering
\includegraphics[width=0.48\textwidth]{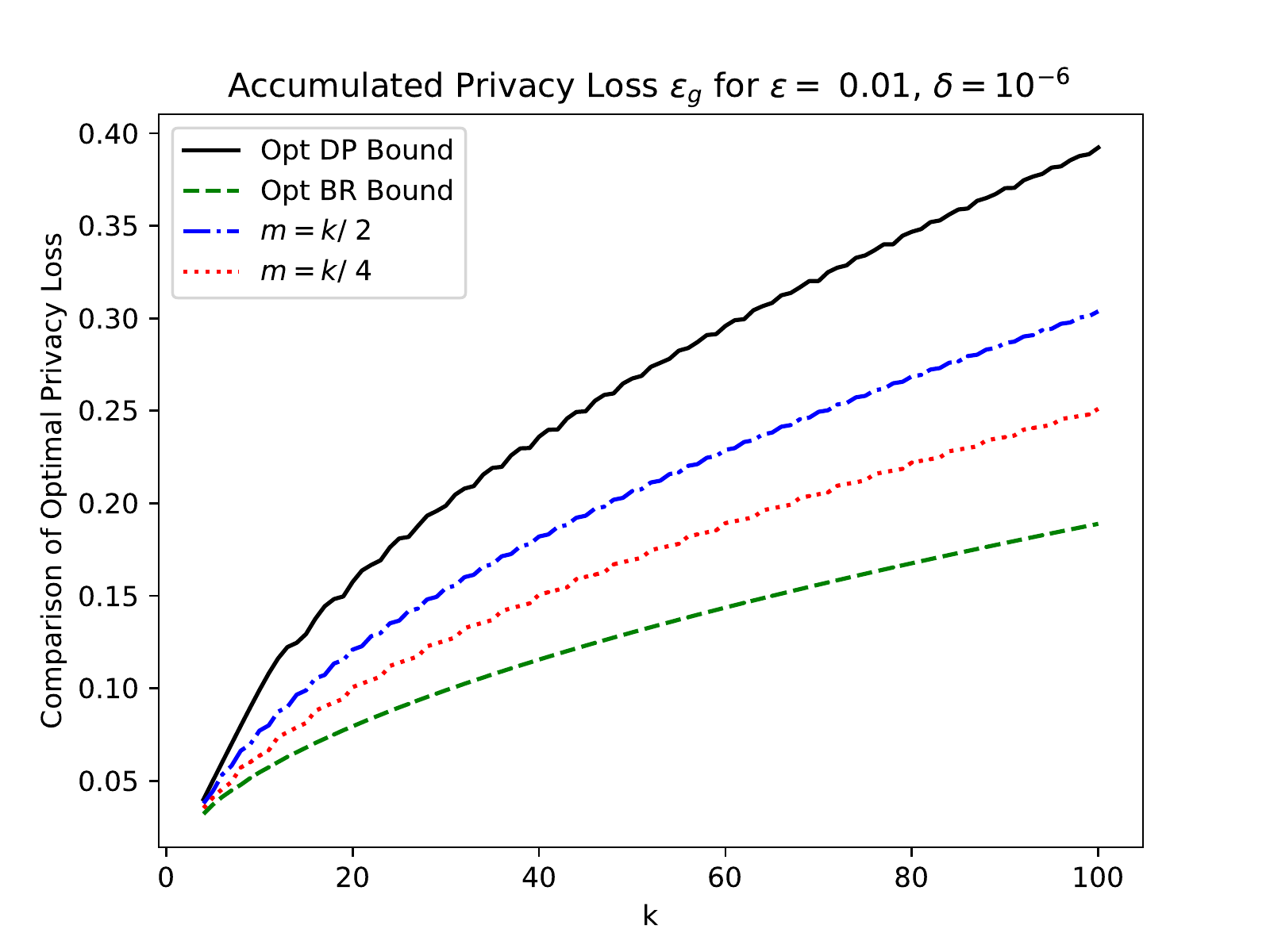}
\includegraphics[width=0.48\textwidth]{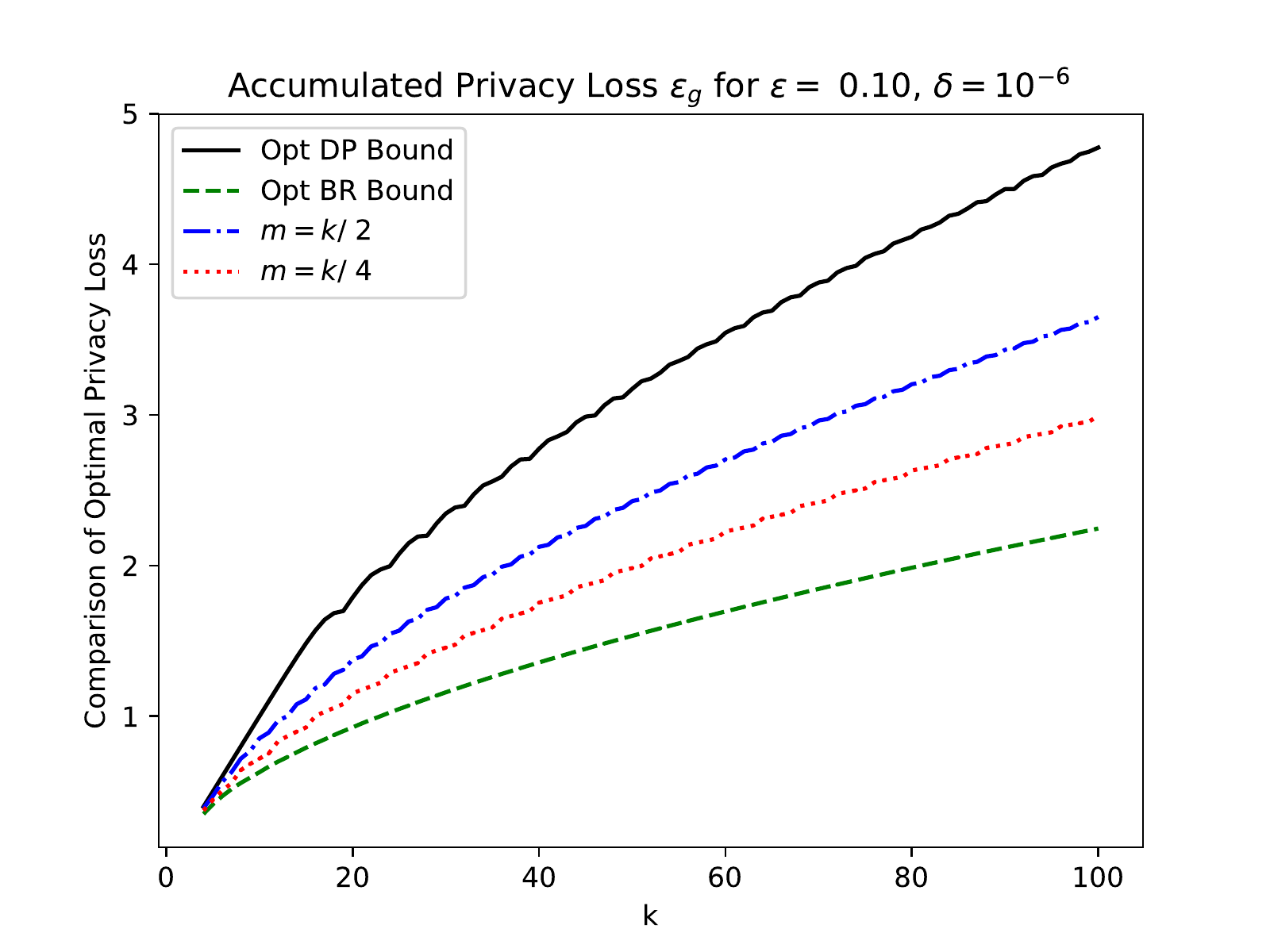}
\caption{We plot the value of $\delta_{\opt}$ from Theorem~\ref{thm:opt_nonadaptive} for $k$ many $\diffp$-DP mechanisms, of which $k-m$ are $\diffp$-BR.  \label{fig:opt_nonadaptive}}
\end{figure}

\section{Optimal Adaptive Composition Bounds  \label{sec:OptAdaptive}}

We next consider the adaptive setting, where each $\diffp$-BR mechanism can be selected as a function of previous outcomes.  From \citet{DongDuRo19}, we know that the privacy loss can strictly increase when compared to the non-adaptive setting. We can use the privacy loss bounds from Lemma~\ref{lem:genbound} here, even in the heterogenous privacy parameter case. However, we can better understand the adaptive setting by examining whether the ordering of DP and BR mechanisms can change the total privacy loss. We have the following formulation of the optimal privacy loss for $\diffp$-BR mechanisms from \citet{DongDuRo19}.
\begin{lemma} \label{lem:adap_recursion_hetero}
Let $\vec{\cM}$ be a sequence of adaptively selected $\mbr(\diffp_i)$ and $\mdp(\diffp_i)$ where $\diffp_i >0$ are fixed in advance for $i \in [k]$.  For any $\diffp_g >0$ and setting $\delta_{\opt}(\emptyset;\diffp_g) = \left[1 - e^{\diffp_g}\right]_+$ we have:
\begin{align*}
\delta_{\opt}(\mbr(\diffp_0), \vec{\cM};\diffp_g)  & = \sup_{t_0 \in [0,\diffp_0]}  \left\{ q_{\diffp_0,t_0} \delta_{\opt}(\vec{\cM}; \diffp_g - t_0)+ (1 - q_{\diffp_0,t_0})  \delta_{\opt}(\vec{\cM}; \diffp_g + \diffp_0 - t_0) \right\}
\\
\delta_{\opt}(\mdp(\diffp_0), \vec{\cM};\diffp_g)  & = q_{2\diffp_0,\diffp_0} \delta_{\opt}(\vec{\cM}; \diffp_g - \diffp_0)+ (1 - q_{2\diffp_0,\diffp_0})  \delta_{\opt}(\vec{\cM}; \diffp_g + \diffp_0 )
\end{align*}
\end{lemma}

As in the previous section, we will only consider a homogenous privacy parameter, $\diffp > 0$, and leave the heterogenous case to future work.

We start with the following result that shows the worst case ordering is to select the BR mechanisms all at the end of the entire interaction. Hence, we might be able to decrease the overall privacy loss if it is known that the BR mechanisms will not all be selected at the end.
\begin{prop}\label{lem:bad_order}
Let $\vec{\cM}, \vec{\cM}'$ be two sequences of adaptively selected $\mbr$ and $\mdp$ mechanisms, where either may be empty.  We then have for any $\diffp_g \geq 0$
\[
\deltaopt(\vec{\cM}, \mbr, \mdp, \vec{\cM}';\diffp_g) \leq \deltaopt(\vec{\cM}, \mdp, \mbr, \vec{\cM}';\diffp_g)
\]
\end{prop}  

\begin{proof}
Let $\vec{\cM}$ consist of $k'$ many $\diffp$-BR mechanisms, which are in positions $\ell_1, \cdots, \ell_{k'}$.  At each level $\ell$, there will be $2^{ \ell - 1}$ many variables to maximize an expression over, which we will write the variables as $t_{j,\ell} \in [0,\diffp]$ where $j \in [2^{\ell - 1}]$.  The full set of these variables can be written as $\cT = \{t_{i,\ell} : i \in [2^{\ell - 1}] ,  \ell \in \{\ell_1, \cdots, \ell_{k'} \} \}$.  Using the recursive formulation from Lemma~\ref{lem:adap_recursion_hetero}, we know that $\deltaopt(\vec{\cM}, \mbr, \mdp, \vec{\cM}';\diffp_g)$ will consist of a summation of terms with the following form, where we take a sup over all $\bbt \in \cT$ for some coefficient $\lambda(\bbt;\diffp)$ and term $\alpha(\bbt;\diffp,\diffp_g)$,
\[ 
\lambda(\bbt;\diffp) \cdot \deltaopt( \mbr, \mdp, \vec{\cM} ;\alpha(\bbt;\diffp,\diffp_g) ).
\]
We then expand this term using our recurrence formula,
\begin{align*}
& \deltaopt( \mbr, \mdp, \vec{\cM} ;\alpha(\bbt;\diffp,\diffp_g) ) \\
& = \sup_{s \in [0,\diffp]} \left\{ q_{\diffp,s}   \deltaopt( \mdp, \vec{\cM} ;\alpha(\bbt;\diffp,\diffp_g) -s) + (1 - q_{\diffp,s})  \deltaopt( \mdp, \vec{\cM} ;\alpha(\bbt;\diffp,\diffp_g) +\diffp - s) \right\} \\
& = \sup_{s \in [0,\diffp]} \left\{ q_{\diffp,s}  q_{2\diffp,\diffp}  \deltaopt(\vec{\cM} ;\alpha(\bbt;\diffp,\diffp_g) -s - \diffp)  + q_{\diffp,s}  (1 - q_{2\diffp,\diffp})  \deltaopt(\vec{\cM} ;\alpha(\bbt;\diffp,\diffp_g) -s + \diffp) \right. \\
& \qquad\qquad \left.+ (1 - q_{\diffp,s}) q_{2\diffp,\diffp} \deltaopt( \vec{\cM} ;\alpha(\bbt;\diffp,\diffp_g) - s) + (1 - q_{\diffp,s}) (1 - q_{2\diffp,\diffp}) \deltaopt( \vec{\cM} ;\alpha(\bbt;\diffp,\diffp_g) +2\diffp - s) \right\}.
\end{align*}

Similarly, $\deltaopt(\vec{\cM}, \mdp, \mbr, \vec{\cM}';\diffp_g)$ will consist of a summation of terms with the following form, with the same terms $\lambda(\bbt;\diffp)$ and $\alpha(\bbt;\diffp,\diffp_g)$
\[ 
\lambda(\bbt;\diffp) \cdot \deltaopt( \mdp, \mbr, \vec{\cM} ;\alpha(\bbt;\diffp,\diffp_g) ).
\]
Again, we use our recurrence formula to get the following
\begin{align*}
& \deltaopt( \mdp, \mbr, \vec{\cM} ;\alpha(\bbt;\diffp,\diffp_g) ) \\
& = q_{2\diffp,\diffp}   \deltaopt( \mbr, \vec{\cM} ;\alpha(\bbt;\diffp,\diffp_g) -\diffp) + (1 - q_{2\diffp,\diffp})  \deltaopt( \mbr, \vec{\cM} ;\alpha(\bbt;\diffp,\diffp_g) +\diffp) \\
& =  q_{2\diffp,\diffp} \sup_{s \in [0,\diffp]} \left\{  q_{\diffp,s}   \deltaopt(\vec{\cM} ;\alpha(\bbt;\diffp,\diffp_g) -s - \diffp)  +  (1 - q_{\diffp,s}) \deltaopt(\vec{\cM} ;\alpha(\bbt;\diffp,\diffp_g) -s ) \right\} \\
& \qquad + (1 - q_{2\diffp,\diffp}) \sup_{s' \in [0,\diffp]} \left\{  q_{\diffp,s'} \deltaopt( \vec{\cM} ;\alpha(\bbt;\diffp,\diffp_g) +\diffp - s') +  (1 - q_{\diffp,s'})  \deltaopt( \vec{\cM} ;\alpha(\bbt;\diffp,\diffp_g) +2\diffp - s') \right\}  \\
& = \sup_{s,s' \in [0,\diffp]} \left\{ q_{\diffp,s}  q_{2\diffp,\diffp}  \deltaopt(\vec{\cM} ;\alpha(\bbt;\diffp,\diffp_g) -s - \diffp)  + q_{\diffp,s'}  (1 - q_{2\diffp,\diffp})  \deltaopt(\vec{\cM} ;\alpha(\bbt;\diffp,\diffp_g) -s' + \diffp) \right. \\
& \qquad\qquad \left.+ (1 - q_{\diffp,s}) q_{2\diffp,\diffp} \deltaopt( \vec{\cM} ;\alpha(\bbt;\diffp,\diffp_g) - s) + (1 - q_{\diffp,s'}) (1 - q_{2\diffp,\diffp}) \deltaopt( \vec{\cM} ;\alpha(\bbt;\diffp,\diffp_g) +2\diffp - s') \right\}
\end{align*}
Hence, $\deltaopt( \mbr, \mdp, \vec{\cM} ;\alpha(\bbt;\diffp,\diffp_g) )$ and $\deltaopt( \mdp, \mbr, \vec{\cM} ;\alpha(\bbt;\diffp,\diffp_g) )$ consists of the same terms, except the former has a single $\sup$ and the latter takes a $\sup$ over two terms.  Hence, we must have 
\[ 
\lambda(\bbt) \cdot \deltaopt( \mbr, \mdp, \vec{\cM} ;\alpha(\bbt;\diffp,\diffp_g) ) \leq \lambda(\bbt;\diffp) \cdot \deltaopt( \mdp, \mbr, \vec{\cM} ;\alpha(\bbt;\diffp,\diffp_g) ).
\]
Because this applies for each term in both $\deltaopt(\vec{\cM}, \mbr, \mdp, \vec{\cM}';\diffp_g)$ and when we switch the BR and DP order $\deltaopt(\vec{\cM}, \mdp, \mbr, \vec{\cM}';\diffp_g)$, we have our result.
\end{proof}
\
Although we have a worst case ordering of $\mbr$ and $\mdp$ mechanisms, we still want to know if the ordering of these mechanisms leads to \emph{strictly} larger privacy losses.  


\subsection{Single BR Mechanism}
We first consider a single BR mechanism $\mbr$ and $k-1$ $\mdp$ mechanisms where each mechanism can be selected adaptively.  It turns out that a single $\diffp$-BR Mechanism is insufficient to induce a difference in $\delta_{\opt}$ for different orderings of the mechanisms.
\begin{prop}\label{prop:sbr-ord}
Let $\diffp>0, N\in\naturals $, and let $\vec{\cA}_{N}, \vec{\cB}_{N}$ be sequences of mechanisms ($\cA_{1} \ldots \cA_{N}$), ($\cB_{1} \ldots \cB_{N})$ for which $\exists k_{\cA}, k_{\cB}\in[N]$ s.t. $\cA_{k_{\cA}}, \cB_{k_{\cB}} = \mbr$ and for $i\neq k_{\cA}, j\neq k_{\cB}$ $\cA_{i}, \cB_{j}= \mdp$.
Then for all $\diffp_{g} \in \reals$ and all $\vec{\cA}_{N}, \vec{\cB}_{N}$:
\begin{align*}
& \delta_{\opt}\left(\vec{\cA}_{N}; \diffp_{g}\right) = \delta_{\opt}\left(\vec{\cB}_{N};\diffp_{g}\right).
\end{align*}
\end{prop}
The proof follows from definitions and case analysis. Details are in appendix (\ref{lem-sbr-ord-1}).

\begin{lemma}\label{lem-sbr-ord-2}
For $\ell\in \naturals$, $x\in\reals$, $\diffp \geq 0$, and $\diffp_{i} = \diffp$ for $i\in\{0, \ldots, \ell\}$ we define:
\begin{align*}
\delta_{\ell}\left(x\right) \defeq \delta_{\opt}\left(\mdp(\diffp_{1}),\ldots,\mdp(\diffp_{\ell}); x\right).\
\end{align*}
Then we have the following identity for some constants $\lambda_{\ell, i} \in \reals$:
\begin{align*}
& \delta_{\ell}\left(x\right) = \sum_{i\in\left\{0,\ldots,\ell\right\}} \lambda_{\ell, i}\left[1-e^{\left(2i-\ell\right)\diffp + x}\right]_{+} .
\end{align*}

\end{lemma}

The proof is by applying Lemma~\ref{lem:adap_recursion_hetero} and induction on $\ell$. Details are in appendix (\ref{lem-sbr-ord-2})

We now provide a proof sketch for Proposition~\ref{prop:sbr-ord}, details are in appendix \ref{lem:sbr-ord:proof}.

\begin{proof}[Proof Sketch for Proposition~\ref{prop:sbr-ord}] 
To prove this we induct on the number, $N$, of $\diffp$-DP mechanisms. The base case $\left(N=2\right)$ follows from expanding both $\deltaopt(\vec{\cA}_{2}; \diffp_{g})$ and $\deltaopt(\vec{\cB}_{2}; \diffp_{g})$ using Lemma ~\ref{lem:adap_recursion_hetero}, eliminating zero terms, and comparing. (See supplementary file)

Now, consider $N>2$. Suppose that for $k<N$ we have $\deltaopt(\vec{\cA}_{k}; \diffp_{g})= \deltaopt(\vec{\cB}_{k}; \diffp_{g})$ for all $\vec{\cA}_{k}, \vec{\cB}_{k}$. Then to show ~$\deltaopt(\vec{\cA}_{N}; \diffp_{g}) = \deltaopt(\vec{\cB}_{N}; \diffp_{g})$ there are two nontrivial cases to consider:
\begin{itemize}
\item[1.] $\cA_{1}, \cB_{1} = \mdp$. Here we simply expand the first $\diffp$-DP mechanism on each side using definitions and apply the inductive hypothesis.
\item[2.] $\cA_{1} = \mbr$ and $\cB_{1} = \mdp$. The proof here is slightly more involved. First we expand the first two terms of each side:
\end{itemize}

\begin{align*}
&\deltaopt\left(\vec{\cA}_{N}; \diffp_{g}\right) = \\
& \sup_{t\in[0,\diffp]}\biggl\{ q_{t, \diffp}\left[ \qedp \deltaopt\left(\mdp,\ldots,\mdp; \diffp_{g} - \diffp - t\right) + \left(1-\qedp\right)  \deltaopt\left(\mdp,\ldots,\mdp; \diffp_{g} + \diffp - t\right) \right]\\
&\vphantom{\bigg\{} + \left(1-q_{t, \diffp}\right) \left[\qedp \deltaopt\left(\mdp,\ldots,\mdp; \diffp_{g} - t\right) + \left(1-\qedp\right)  \deltaopt\left(\mdp,\ldots,\mdp; \diffp_{g} + 2\diffp - t\right) \right] \biggr\} .\\
&\deltaopt\left(\vec{\cB}_{N}; \diffp_{g}\right) =\\
& \qedp\sup_{\ton\in[0,\diffp]}\biggl\{ q_{\ton, \diffp} \deltaopt\left(\mdp,\ldots,\mdp; \diffp_{g} - \diffp - \ton\right) + \left(1-q_{\ton, \diffp}\right)  \deltaopt\left(\mdp,\ldots,\mdp; \diffp_{g} - \ton\right) \biggr\} + \left(1- \right.  \\
&\left. \qedp\right) \sup_{\ttw\in[0,\diffp]}\biggl\{ q_{\ttw, \diffp} \deltaopt\left(\mdp,\ldots,\mdp; \diffp_{g} + \diffp - \ttw\right) + \left(1-q_{\ttw, \diffp}\right)  \deltaopt\left(\mdp,\ldots,\mdp; \diffp_{g} + 2\diffp - \ttw\right) \biggr\}.
\end{align*}

First we expand the first two terms of each side, applying the inductive hypothesis to $\vec{\cB}_{N}$ after expanding the first term to ensure that the second mechanism is $\diffp$-BR wlog.
Then, by applying the summation and reduction formulas from Lemmas ~\ref{lem-sbr-ord-2} and ~\ref{lem-sbr-ord-1}, for some $C_{1}(\diffp), C_{2}(\diffp) \in \reals$, we can write:

\begin{align}
&\deltaopt\left(\vec{\cA}_{N}; \diffp_{g}\right) =  \sup_{t\in[0,\diffp]}\biggl\{ \qedp \left(C_{1}(\diffp) + q_{t, \diffp} \lambda_{i^{*}} \left[1-e^{ \diffp_{g} + \left(-1 + 2i^{*} - \left(N-2\right)\right)\diffp- t}\right]_{+}\right) \nonumber \\
&  + \left(1-\qedp\right) \left(C_{2}(\diffp) + \quad q_{t, \diffp} \lambda_{i^{*}}\left[1-e^{\diffp_{g} + \left(1+ 2i^{*} - \left(N-2\right)\right)\diffp- t}\right]_{+}\right)\biggr\}. \label{sbr-1} \\
&\deltaopt\left(\vec{\cB}_{N}; \diffp_{g}\right) = \qedp\sup_{\ton\in[0,\diffp]}\biggl\{ C_{1}(\diffp) + q_{t, \diffp} \lambda_{i^{*}} \left[1-e^{ \diffp_{g} + \left(-1 + 2i^{*}- \left(N-2\right)\right)\diffp- \ton}\right]_{+}\biggr\} \nonumber \\
& + \left(1-\qedp\right) \sup_{\ttw\in[0,\diffp]}\biggl\{C_{2}(\diffp) + q_{\ttw, \diffp} \lambda_{i^{*}}\left[1-e^{\diffp_{g} + \left(1+ 2i^{*} - \left(N-2\right)\right)\diffp- \ttw}\right]_{+}\biggr\}. \label{sbr-2}
\end{align}

Then both remaining terms in $\deltaopt(\vec{\cB}_{N}; \diffp_{g})$ achieve the supremum for the same argument $  t^{*}=\ton=\ttw$.  Fixing $t=t^{*}$ in (\ref{sbr-1}) and $\ton=\ttw=t^{*}$ in (\ref{sbr-2}) and comparing it is clear that $\deltaopt(\vec{\cA}_{N}; \diffp_{g}) \geq  \deltaopt(\vec{\cB}_{N}; \diffp_{g})$. By the triangle inequality, we also have $\deltaopt(\vec{\cB}_{N}; \diffp_{g}) \geq \deltaopt(\vec{\cA}_{N}; \diffp_{g})$, which proves the result.
\end{proof}

The following corollary states that we can use the non-adaptive, optimal DP composition formula from Theorem~\ref{thm:opt_nonadaptive} with $m= 1$, even when the single BR mechanism can be adaptively selected at any round.
\begin{corollary}
Let $\vec{\cA}$ be any sequence of $k-1$ adaptively selected $\mdp$ mechanisms and a single $\mbr$ that can be adaptively selected, e.g. $\vec{\cA} = (\mdp, \cdots, \mdp, \mbr, \mdp, \cdots, \mdp)$.  Then $\deltaopt(\vec{\cA}; \diffp_g)$ can be computed with Theorem~\ref{thm:opt_nonadaptive}.
\end{corollary}

\subsection{Ordering of BR and DP Impacts the Privacy Loss}

To see that the privacy loss can strictly increase if we change the ordering of BR and DP mechanisms, we consider the following simple example where we compose $k = 3$ $\diffp$-DP mechanisms adaptively and $2$ of them are $\diffp$-BR.  In this case, we can directly compute $\deltaopt(\cdot;\diffp_g)$ for each of the three possible orderings of mechanisms.
\begin{lemma}\label{lem:xyz}
Let $\diffp>0$ and $\diffp_g \geq 0$.  If $\diffp_g \geq \diffp$, then we have
\begin{align*}
\deltaopt(\mbr, \mbr, \mdp ;\diffp_g) &= \deltaopt(\mbr, \mdp, \mbr ;\diffp_g)
\\
& = \deltaopt(\mdp, \mbr, \mbr ;\diffp_g).
\end{align*}
For $\diffp_g \leq \diffp$ we define the following functions:
\begin{align*}
x(t) & \defeq q_{2\diffp,\diffp}  \left( q_{\diffp,t} q_{\diffp,\tfrac{\diffp_g - t}{2}}^2 \left(1- e^{- \diffp}\right) + ( 1 - q_{\diffp,t } )  q_{\diffp,\tfrac{\diffp_g + \diffp - t}{2}}^2 \left( 1 - e^{ - \diffp} \right) \right) \\
y(t) & \defeq q_{2\diffp,\diffp} \left(q_{\diffp,t} (1 - e^{\diffp_g - \diffp - t}) + (1-  q_{\diffp,t}) q_{\diffp,\tfrac{\diffp_g +\diffp - t}{2}}^2 \left(1- e^{- \diffp}\right) \right) \\
z(t) & \defeq (1 - q_{2\diffp,\diffp})q_{\diffp,t}q_{\diffp,\diffp + \tfrac{\diffp_g - t}{2}}^2 \left( 1 - e^{ -\diffp} \right) .
\end{align*}
We then can write out the following expressions:
\begin{align*}
\deltaopt(\mdp, \mbr, \mbr ;\diffp_g)  & = \max \left\{ \sup_{t \in [0,\diffp_g)}x(t), \sup_{t \in [\diffp_g, \diffp] }y(t) \right\} + \sup_{t' \in [\diffp_g, \diffp]} z(t') \\
 \deltaopt(\mbr, \mdp, \mbr ;\diffp_g)  & = \max \left\{ \sup_{t \in [0,\diffp_g)}x(t), \sup_{t \in [\diffp_g, \diffp] }y(t) +z(t) \right\} \\
\deltaopt(\mbr, \mbr, \mdp ;\diffp_g)  & = \max \left\{ \sup_{t \in [0,\diffp_g) } x(t), \sup_{t \in [\diffp_g,\diffp] } y(t) + z(t)  \right\} .
\end{align*}
\end{lemma}

The proof requires several technical details and is left to the appendix (\ref{lem:xyz:proof}).

From this, we can also obtain the following result, which demonstrates that the ordering of the mechanisms changes the overall privacy loss.

\begin{lemma}\label{lem:ordbrdp2}
Let $0\leq \diffp_g < \diffp$.  Then 
\[
\deltaopt(\mdp, \mbr, \mbr ;\diffp_g) > \deltaopt(\mbr, \mdp, \mbr ;\diffp_g).
\]
Furthermore, for $x(t), y(t), z(t)$ defined in Lemma~\ref{lem:xyz}, we have
\begin{align*}
\deltaopt(\mdp, \mbr, \mbr ;\diffp_g)  & = \left\{ \begin{array}{lr}
													x(\diffp/2) + z\left( \frac{2\diffp + \diffp_g}{3} \right) & \text{ if } \diffp_g \geq \diffp/2 \\
													y\left( \frac{\diffp + \diffp_g}{3} \right) +   z\left( \frac{2\diffp + \diffp_g}{3} \right)& \text { else}
													\end{array}
											\right.
											\\
 \deltaopt(\mbr, \mdp, \mbr ;\diffp_g)  & = \left\{ \begin{array}{lr}
													\max\{ x(\diffp/2), y(\diffp_g) + z(\diffp_g) \} & \text{ if } \diffp_g \geq \diffp/2 \\
													\max\{x(\diffp_g), y(\diffp/2) + z(\diffp/2) \} & \text { else}
													\end{array}
											\right.
\end{align*}
\end{lemma}

Again the proof is relegated to the appendix (\ref{lem:ordbrdp2:proof}).

Knowing that there is a difference in the overall privacy loss when an analyst changes the order of BR and DP mechanisms, we then plot the ratio between $\deltaopt(\mbr, \mdp, \mbr;\diffp_g)$ and $\deltaopt(\mdp, \mbr, \mbr; \diffp_g)$ in Figure~\ref{fig:gap}.

\begin{figure}[h]
\centering
\includegraphics[width=0.48\textwidth]{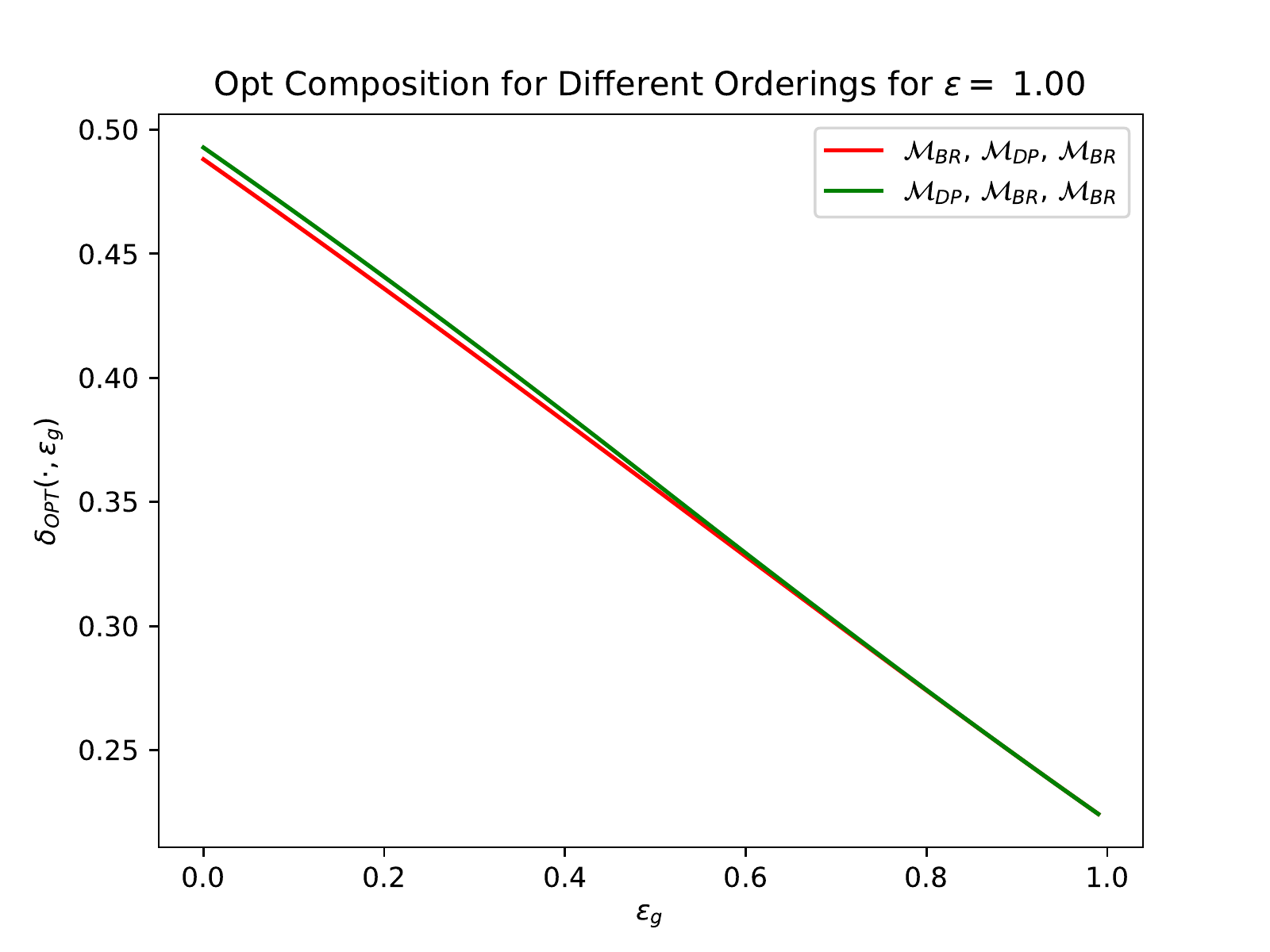}
\includegraphics[width=0.47\textwidth]{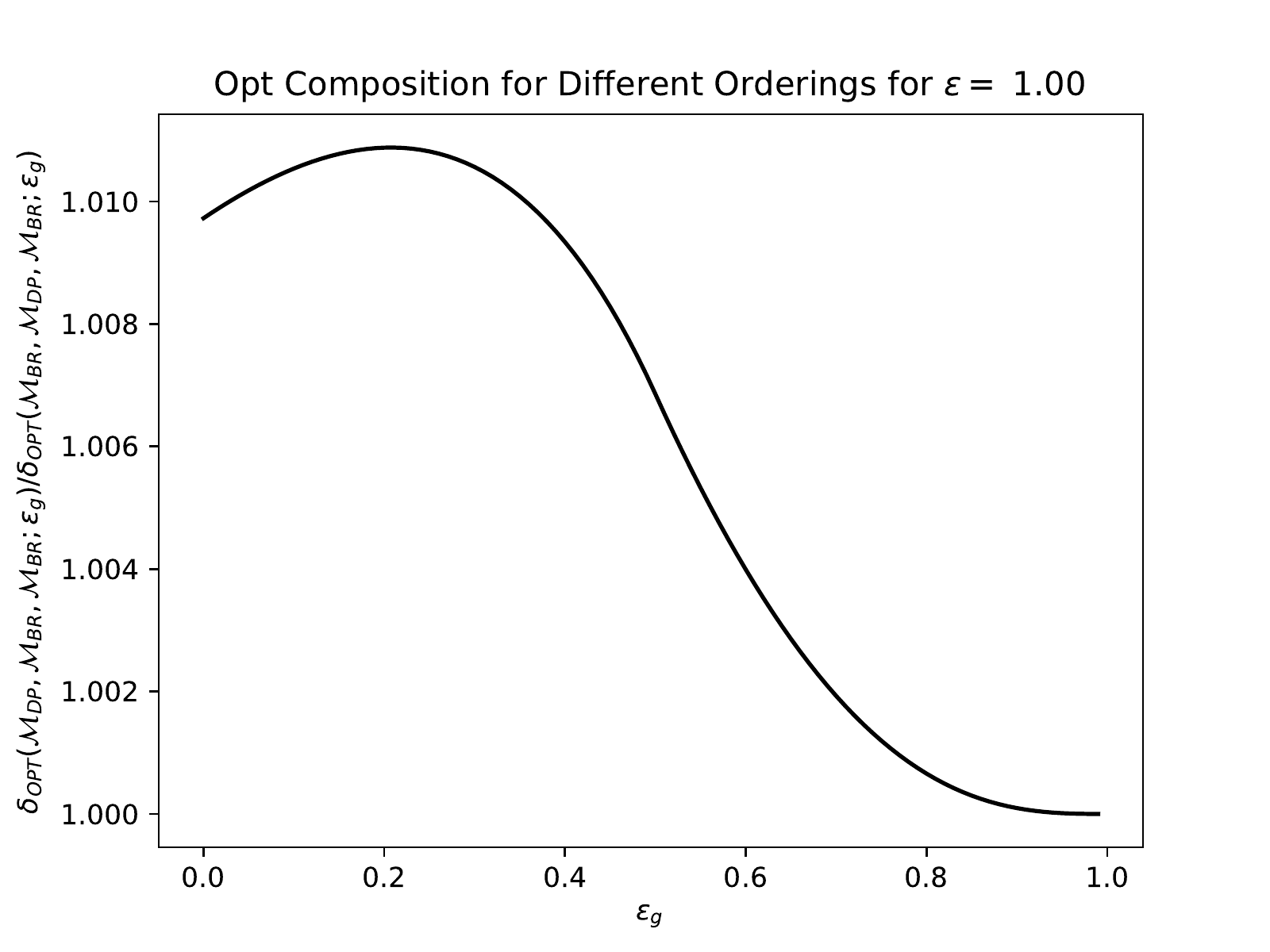}
\caption{The absolute difference (left) and the ratio (right) between $\deltaopt(\mbr, \mdp, \mbr;\diffp_g)$ and $\deltaopt(\mdp, \mbr, \mbr; \diffp_g)$ with $\diffp = 1.0$ and varying $\diffp_g < \diffp$.  \label{fig:gap}}
\end{figure}

\section{Comparing Gaussian and Laplace Mechanisms for Private Histograms \label{sec:GM}} 
In this section, we compare the Laplace mechanism and the Gaussian mechanism over multiple rounds of composition.  Before we consider composition, the main distinguishing feature of which noise to use, Laplace or Gaussian, mainly relies on which sensitivity bound we have on the quantity we want to add noise to.  We now define $\ell_p$-sensitivity of a given function $f: \cX \to \R^d$,
\[
\Delta_p(f) \defeq \max_{x\sim x'}|| f(x) - f(x') ||_p.
\]
Typically, if we have a bound on $\ell_1$-sensitivity we would use the Laplace mechanism, whereas if we have a bound on $\ell_2$-sensitivity we would use the Gaussian mechanism, since each mechanism adds noise with standard deviation proportional to the $\ell_1$ or $\ell_2$ sensitivity, respectively.  In machine learning applications, one typically normalizes high dimensional vectors by some $\ell_2$ bound, so it is then natural to use Gaussian noise, not to mention the \emph{moment accounting} composition that also takes advantage of subsampling at each round of stochastic gradient descent \cite{Abadietal16}.  

However, when one wants to privatize counts from a histogram, one typically does not have only an $\ell_1$ or $\ell_2$ sensitivity bound.  Rather, one has a bound on the number of distinct counts a user can contribute to in the histogram, i.e. an $\ell_0$-sensitivity bound, and a bound on the amount a user can impact a single element's count, i.e, an $\ell_\infty$-sensitivity bound.  Obviously, we can convert these bounds to $\ell_1$ or $\ell_2$ sensitivity bounds.  Let $\tau$ be the $\ell_\infty$-sensitivity and $\Delta$ be the $\ell_0$-sensitivity of the histogram $h \in \N^d$, computed on the input data. We then have 
\[ 
\Delta_1(h) = \tau \Delta, \qquad \Delta_2(h) = \tau \sqrt{\Delta}.
\]

Hence, it would seem better to use the Gaussian mechanism to release counts, since the $\ell_2$-sensitivity can be much lower than the $\ell_1$-sensitivity, thus leading to more accurate counts.  We have the following result from \citet{BalleWa18} that gives the smallest scale of noise for the Gaussian mechanism with a given $\ell_2$-sensitivity and overall privacy guarantee.
\begin{lemma}[Analytic Gauss \cite{BalleWa18}]\label{lem:AGM}
Let $f: \cX \to \R^d$ have $\ell_2$-sensitivity $\Delta_2$, then for any $\diffp>0$ and $\delta \in (0,1]$ we have $M(x) = \Normal{f(x)}{ \Delta_2^2 \sigma^2I_d}$  is $(\diffp,\delta)$-DP if and only if
\[
\Phi\left(\frac{1}{2\sigma} - \diffp\sigma \right) - e^\diffp \Phi\left(- \frac{1}{2 \sigma} - \diffp \sigma \right)\leq \delta
\]
\end{lemma}
We now discuss how we can actually bound the privacy loss of the Laplace mechanism by considering composition of $\Delta$ ($\ell_0$-sensitivity) many DP mechanisms.
\begin{lemma}\label{lem:peeling_lap}
Consider a function $f: \cX \to \R^d$ with $\ell_0$-sensitivity $\Delta$ and $\ell_\infty$-sensitivity $\tau$.  Then the Laplace mechanism with parameter $\diffp>0$
\[
f(x) + (Z_1, \cdots, Z_d), \qquad \{Z_i : i \in [d] \} \stackrel{i.i.d.}{\sim} \lap(\tau/\diffp).
\]
is $(\diffp_g, \deltaopt( \mdp(\diffp_1), \cdots, \mdp(\diffp_{\Delta}) ;\diffp_g))$-DP for $\diffp_g \geq 0$ and $\diffp_i = \diffp$.
\end{lemma}
\begin{proof}
Note that we fix neighboring datasets, $x$, $x'$ which induces two function values $f(x)$, $f(x')$ that differ in at most $\Delta$ positions, and in each position they differ by at most $\tau$.  Since we fix the neighboring datasets, we also know the $\Delta$ positions that have changed.  Hence, we need to only consider the contributions to the overall privacy loss in these $\Delta$ positions, while the other positions contribute zero to the overall privacy loss and can be dropped.  
\end{proof}

Note that a similar argument can be made with the Gaussian mechanisms, given the $\ell_0$-sensitivity.  For this, we will analyze the mechanism using zCDP from Definition~\ref{defn:zCDP}.  

Using a similar argument to Lemma~\ref{lem:peeling_lap} and using the composition property of zCDP from \citet{BunSt16}, we have the following result, which can be optimized over $\alpha > 1$.
\begin{lemma}\label{lem:GM}
Consider a function $f: \cX \to \R^d$ with $\ell_0$-sensitivity $\Delta$ and $\ell_\infty$-sensitivity $\tau$.  Then the Gaussian mechanism 
\[
f(x) + (Z_1, \cdots, Z_d), \qquad \{Z_i : i \in [d] \} \stackrel{i.i.d.}{\sim} \Normal{0}{ \tau^2 \sigma^2}.
\]
is $(\frac{\Delta}{2\sigma^2} + \frac{1}{\sigma} \sqrt{2\Delta \log(1/\delta)}, \delta)$-DP for any $\delta>0$.
\end{lemma} 

\subsection{Results}
Given an $\ell_0$-sensitivity bound, we then want to compare the various bounds of the Gaussian mechanism from Lemma~\ref{lem:AGM} and Lemma~\ref{lem:GM} with our Laplace mechanism bound from Lemma~\ref{lem:peeling_lap}.  In order to compare the utility of the Laplace and Gaussian mechanisms, we will fix the variances to be the same between them and set $\delta >0$ to be the same across.\footnote{Note that the variance of $\lap(b)$ is $2 b^2$.}  We present the comparisons in Figure~\ref{fig:GM_oneshot}.  Note that for relatively small $\ell_0$-sensitivities, we get an improvement in the overall privacy guarantee with the Laplace mechanism, but then Gaussian noise seems to win out as the $\ell_0$-sensitivity increases.
\begin{figure}[h]
\centering
\includegraphics[width=0.48\textwidth]{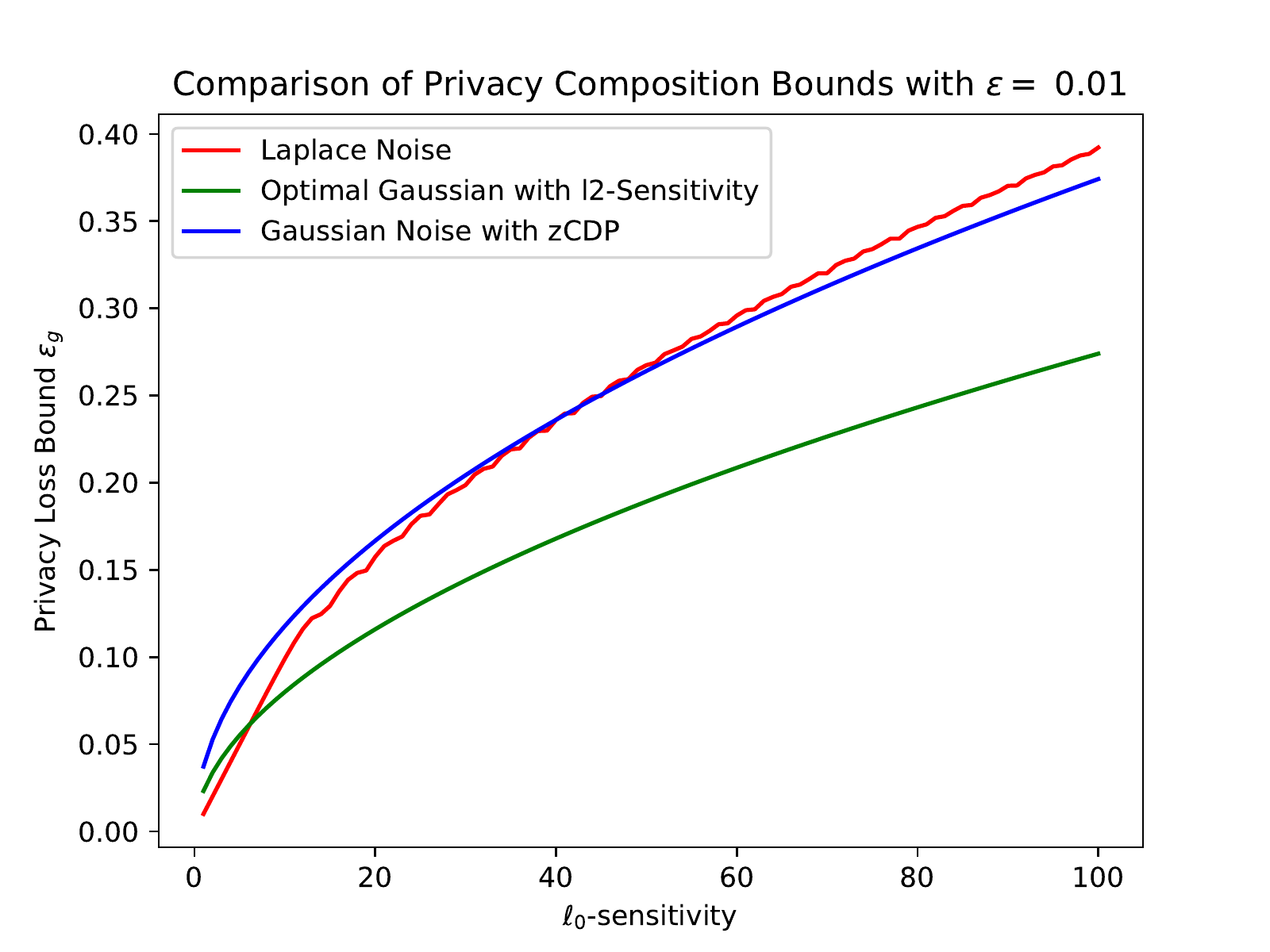}
\includegraphics[width=0.48\textwidth]{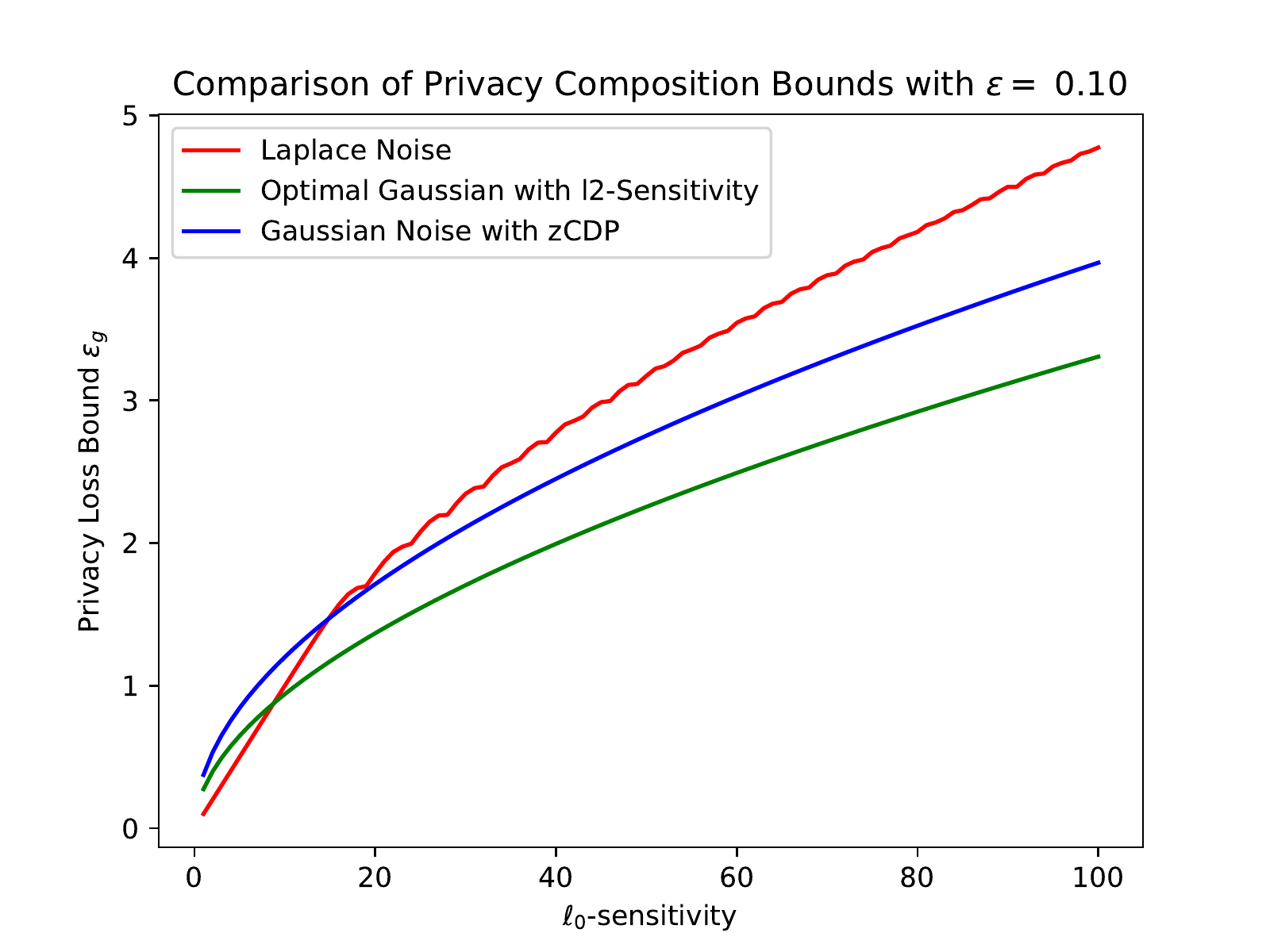}
\caption{Comparison of the overall DP guarantee with $\delta = 10^{-6}$, with ``Laplace Noise" being the bound in Lemma~\ref{lem:peeling_lap}, ``Optimal Gaussian with $\ell_2$-sensitivity" being the bound in Lemma~\ref{lem:AGM}, and ``Gaussian Noise with zCDP" being the bound in Lemma~\ref{lem:GM}.  Note that we equalize the standard deviations between the Laplace and Gaussian noise. \label{fig:GM_oneshot}}
\end{figure}

We are not just concerned with the overall DP parameters for a one-shot mechanism; we also want to consider composition.  In order to apply composition for Lemma~\ref{lem:AGM}, we will use the optimal DP composition bound from \cite{MurtaghVa16} over $k$ different mechanisms with $\ell_0$-sensitivity $\Delta$ and find the smallest $\diffp_g$ value for the given $\sigma$ and overall $\delta$.  We present the results in Figure~\ref{fig:GM_comp}.

\begin{figure}[h]
\centering
\includegraphics[width=0.48\textwidth]{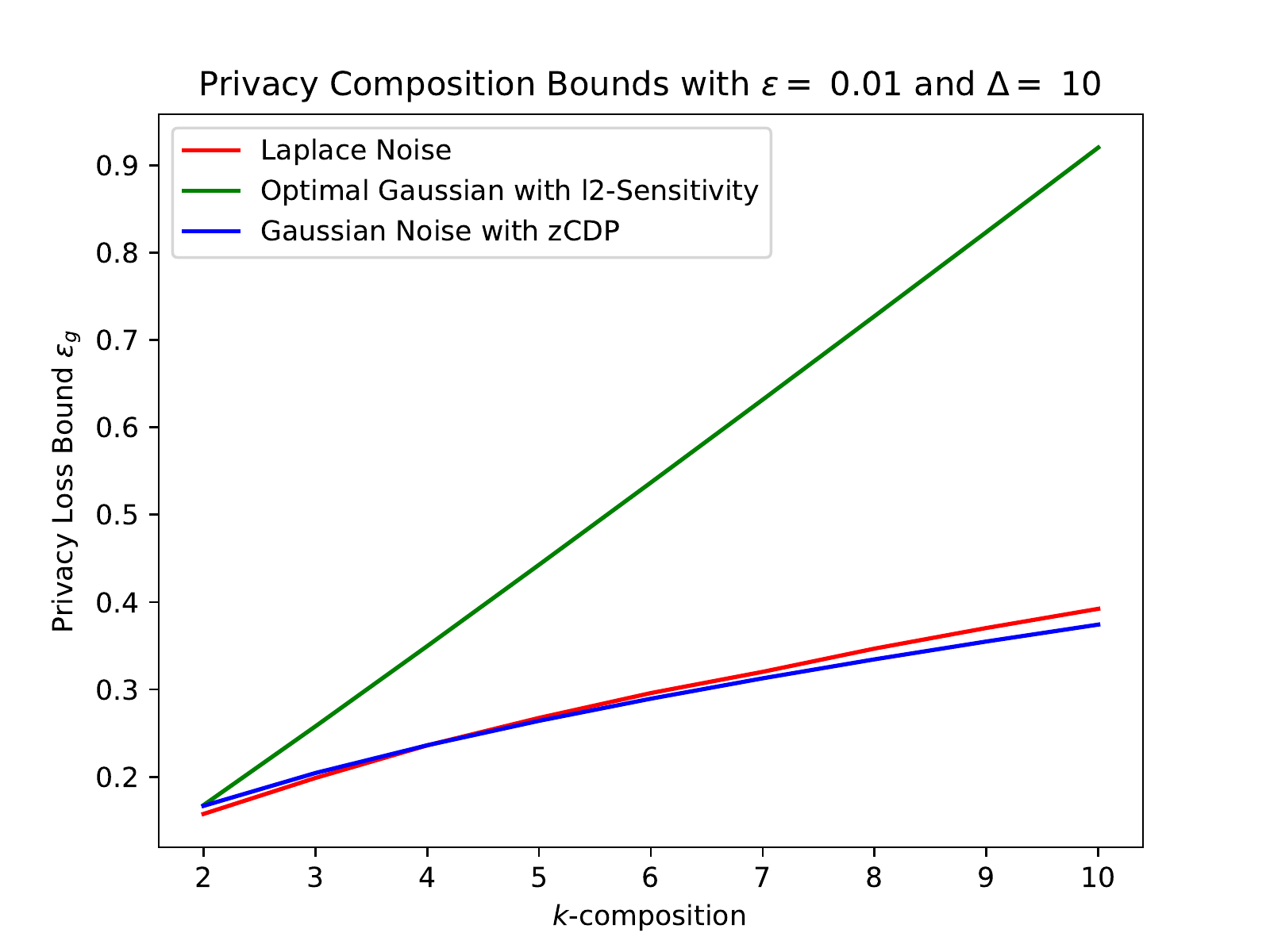}
\includegraphics[width=0.48\textwidth]{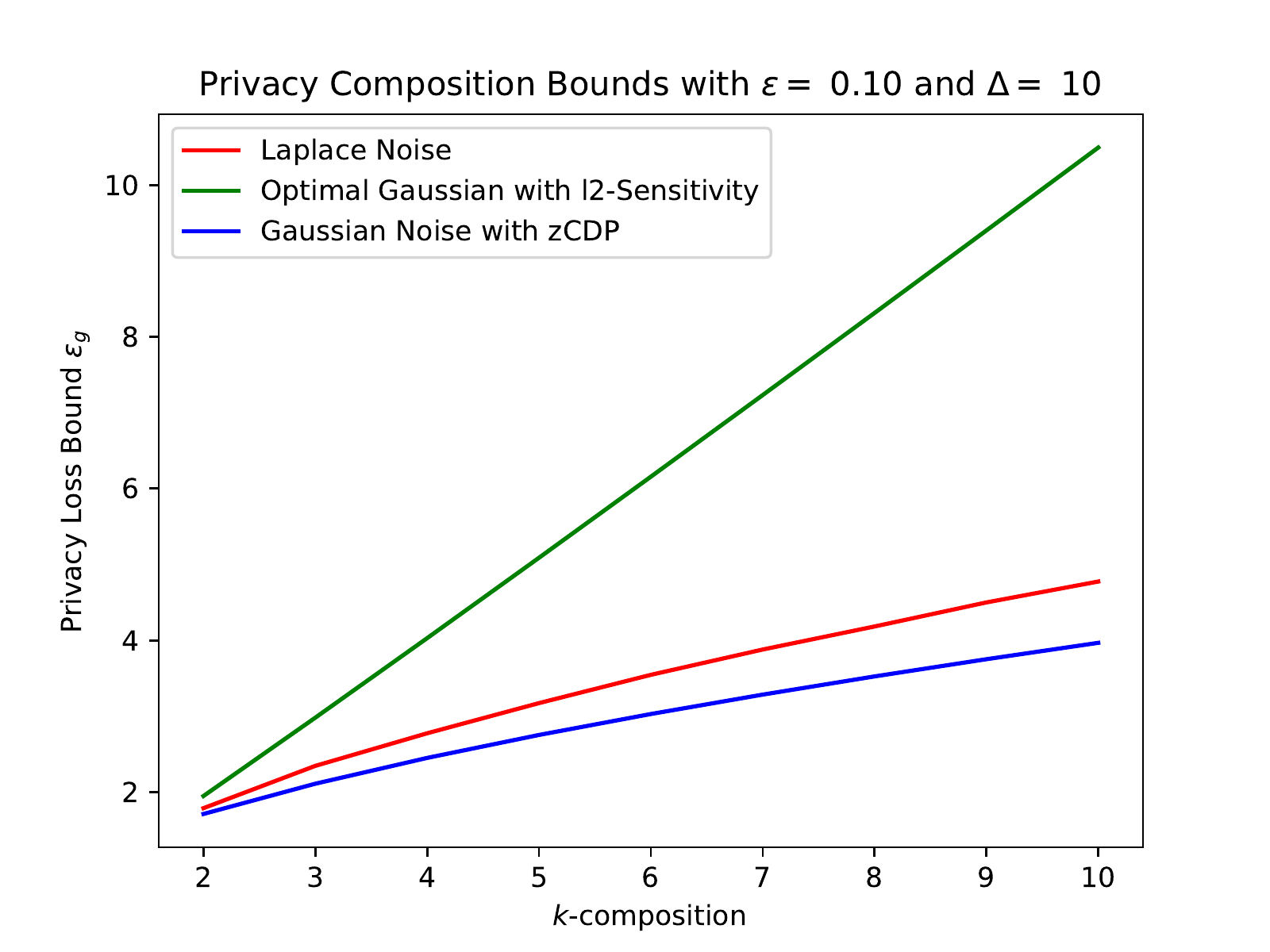}
\caption{Similar to Figure~\ref{fig:GM_oneshot}, we now consider composition with $k$ mechanisms, each adding noise to a $\ell_0$-sensitivity $\Delta=10$ function.   \label{fig:GM_comp}}
\end{figure}

Note the bounds for ``Laplace Noise" and ``Gaussian Noise with zCDP" do not change between Figures~\ref{fig:GM_oneshot} and~\ref{fig:GM_comp}, since the bounds consider composing $k\cdot\Delta$-mechanisms which varies from 20 to 100 in Figure~\ref{fig:GM_comp}.  We see that the best one-shot mechanism then does not provide a small privacy loss when composing several mechanisms.
From the empirical results, we see that the Gaussian mechanisms typically have smaller accumulated privacy loss after a large enough $k$, however for a reasonable number of compositions, Laplace outperforms Gaussian.  

\subsection{Gaussian Based Private Top-$k$ Mechanisms}
Given the results above, we then propose variants of existing mechanisms that use Gaussian noise rather than Laplace noise.  From \citet{RogersSuPeDuLeKaSaAh20}, we have the following table of mechanisms for data analytics based on histogram data.  In the $\Delta$-restricted sensitivity setting we assume that the number of categories a single user can impact is at most $\Delta$ (i.e.the $\ell_0$-sensitivity) and the unrestricted sensitivity setting has no such restriction, but requires limiting the output to the top-$k$.  Furthermore, the known domain setting is where the algorithms are given the set of categories over the histogram (the labels for the $x$-axis), because they cannot be given by the data.  Whereas the unknown domain has no such restriction and must have a parameter $\bar{d}$ for an upper bound on the number of distinct elements the histogram can have.  Each setting requires a $\ell_\infty$-sensitivity bound $\tau$.  

\begin{table}[htbp]
\centering\setcellgapes{4pt}\makegapedcells
\begin{tabular}{ |c|c|c| } 
 \hline
 & $\Delta$-restricted sensitivity & unrestricted sensitivity \\ 
 \hline
 \shortstack{Known \\ Domain} & $\knownLap{\Delta,\tau}$ \cite{DworkMcNiSm06} & $\knownEM{k,\tau}$ \cite{McSherryTa07} \\ 
 \hline
\shortstack{Unknown \\ Domain} & $\rTE{\Delta,\bar{d},\tau}$ \cite{RogersSuPeDuLeKaSaAh20} & $\rT{k,\bar{d},\tau}$ \cite{RogersSuPeDuLeKaSaAh20} \\ 
 \hline
\end{tabular}
\caption{DP algorithms for various data analytics tasks\label{table:tasks}}
\end{table}

Rather than present each algorithm here, we summarize each one and present a variant of it using Gaussian noise.  The $\knownLap{\Delta,\tau}$ mechanism can easily be replaced with a Gaussian mechanism that adds Gaussian noise to each count with $\tau\sigma$ standard deviation to ensure $(\tfrac{\Delta}{2\sigma^2} , \tfrac{\sqrt{\Delta}}{\sigma} )$-CDP \cite{DworkRo16} and $\tfrac{\Delta}{2\sigma^2}$-zCDP \cite{BunSt16} since at most $\Delta$ bins in neighboring histograms can change.  

The unrestricted sensitivity algorithms are based on the exponential mechanism (using Gumbel noise) to first discover the elements in the top-$k$ and then use the Laplace mechanism to release noisy counts on the discovered elements.  A simple modification of this algorithm is to still use exponential mechanisms to discover the elements but then use Gaussian noise on the resulting counts.  Hence, adding Gaussian noise to the count of each discovered element with $\tau\sigma$ standard deviation will guarantee $(\tfrac{k}{2\sigma^2} , \tfrac{\sqrt{k}}{\sigma} )$-CDP and $\tfrac{k}{2\sigma^2}$-zCDP.  However, this ignores a useful detail that the first phase of exponential mechanisms is giving, a ranked list of elements.  Hence, we propose a simple post processing function of the Gaussian mechanism that will respect the order given by the first phase of domain discovery.  We solve a constrained least squares problem to return the maximum likelihood estimator for the true counts given an ordering.  We present the $\LSNoise{\sigma}$ procedure in Algorithm~\ref{algo:least_squares}.
\begin{algorithm}[h!]
	\caption{$\LSNoise{\sigma}$; Return noisy counts subject to a fixed ordering}
	\begin{algorithmic}
		\State \textbf{Input:} Histogram $\bbh = h_1, \cdots, h_k$ of $k$ elements with $\ell_\infty$-sensitivity $\tau$ and an ordering $i_1,i_2, \cdots i_k$.
		\State \textbf{Output:} Noisy counts $\tilde{h}_{i_1}, \tilde{h}_{i_2},\cdots, \tilde {h}_{i_k}$.
		\State Add noise $\hat{\bbh} =  (h_1, \cdots, h_k) + (Z_1, \cdots, Z_k)$ where $Z_i \sim \Normal{0}{\tau^2\sigma^2}$.  
		\State Solve the following and let $\tilde{\bbh} = (\tilde{h}_1,\cdots, \tilde{h}_k)$ be the solution:
		\begin{align}
		\min_{\bbx} \quad & \quad || \hat{\bbh} - \bbx ||_2 \label{eq:least_squares} \\
		\text{s.t.} \quad  & \quad  x_{i_1} \geq \cdots \geq x_{i_k} \geq 0 \nonumber
		\end{align}
		\State Return $(\tilde{h}_{i_1},\cdots, \tilde{h}_{i_k} )$.	
	\end{algorithmic}\label{algo:least_squares}
\end{algorithm}
Because this is a post-processing function of the Gaussian mechanism, the privacy parameters remain the same given the outcome of $\LSNoise{\sigma}$.  Note that we could use an $\ell_1$ loss to find the MLE when using Laplace noise, however this would not guarantee a unique solution.  An additional advantage of using the $\ell_2$ loss is that we can use standard constrained least squares numerical methods, e.g. `nnls' in scipy.

We give the comparison between $\LSNoise{\sigma}$ and simply adding Laplace noise to the counts with the same overall privacy budget in Figure~\ref{fig:NoiseCompare} and $\ell_\infty$-sensitivity of $\tau = 1$.  In particular we consider the top-$25$ words in Macbeth and fix $\diffp = 0.1$ in $\knownLap{25,1}$.  We then compute the overall privacy parameter with $\delta = 10^{-6}$ using the optimal composition bound to get $\diffp'(\delta) = 2.08$.  Solving for $\sigma$ in Lemma~\ref{lem:GM} with the given total $\epsilon'(\delta)$ and $\delta = 10^{-6}$, we get $\sigma = 13.1$.  We then ran the two privatized count algorithms over 100 trials and plot 1 standard deviation from the empirical average.  From the figure, we can see that $\LSNoise{\sigma}$ has smaller error for smaller counts than the basic Laplace mechanism with the same privacy guarantee.

\begin{figure}[h]
\centering
\includegraphics[width=12cm]{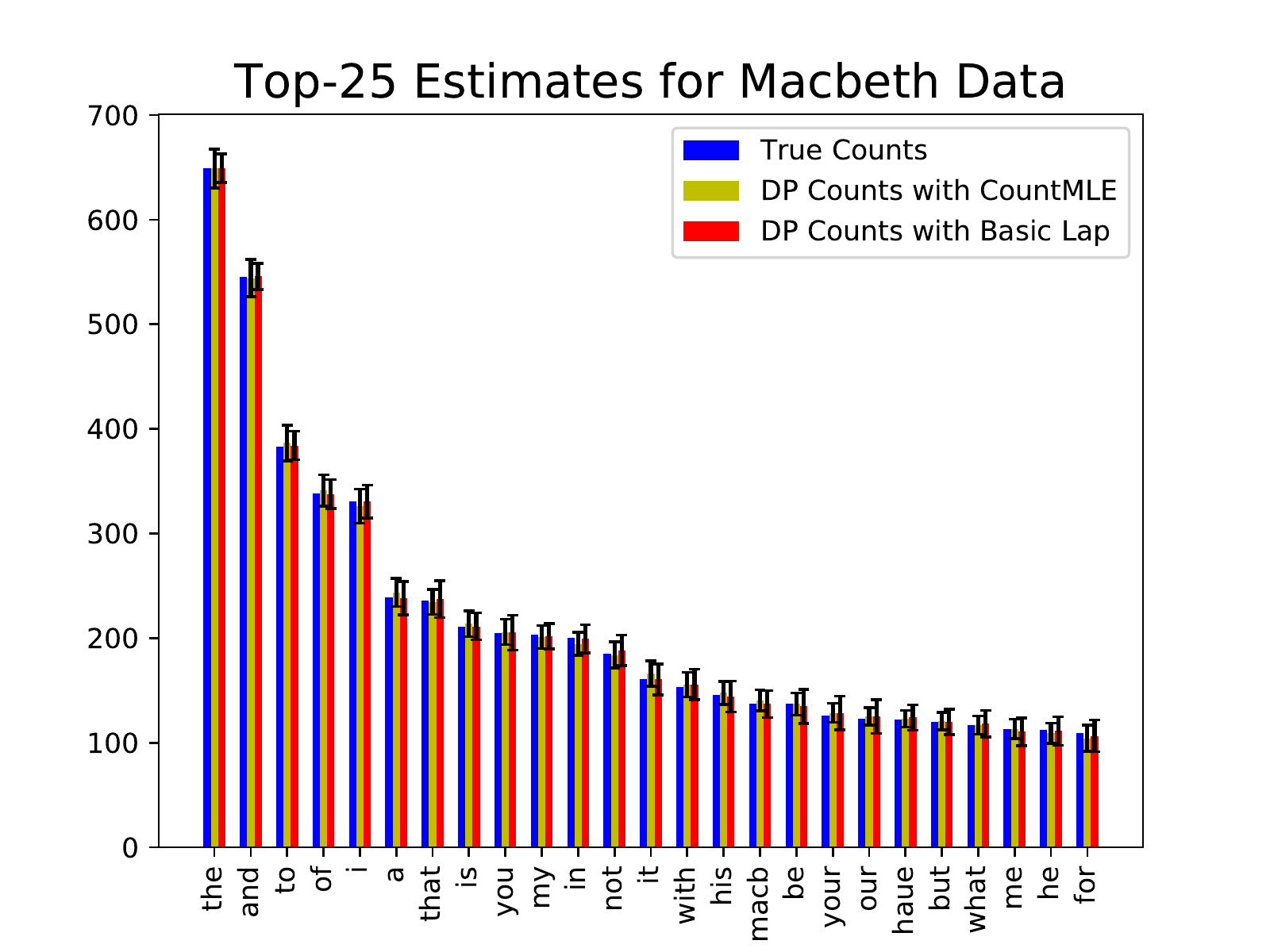}
\caption{Comparison of $(2.08,10^{-6})$-DP noise addition mechanisms given a fixed ordering.\label{fig:NoiseCompare}}
\end{figure}

Next, we consider the unknown domain and $\Delta$-restricted sensitivity setting.  The procedure $\rTE{\Delta,\bar{d},\tau}$ from \cite{RogersSuPeDuLeKaSaAh20} adds Laplace noise to the $\bar{d}$ counts that it has access to from the true histogram and adds Laplace noise proportional to $\Delta\tau/\epsilon$ to each count.  Then it adds a data dependent threshold, which also has Laplace noise added to it and only returns the elements above the noisy threshold.  The resulting algorithm is $(\epsilon,\delta)$-DP.  We propose a variant of this algorithm using a symmetric truncated Gaussian, which we write as $\truncNormal{T}{\mu}{\sigma^2}$ with truncation value $T>0$ that has the following density $f_T$ for $\mu - T\leq z\leq \mu+T $,
\[
f_T(z;\mu,\sigma) = \frac{1}{\sigma} \frac{\phi\left( \frac{z - \mu}{\sigma} \right)}{\Phi\left( \frac{T}{\sigma}\right) - \Phi\left(\frac{-T}{\sigma} \right)}.
\]
and otherwise $f_T(z;\mu,\sigma) = 0$.  Further the parameter $\bar{d}$ needs to be larger than the histogram distribution $d$ to ensure the full histogram is available.  The procedure $\truncGauss{\Delta,\bar{d},\tau}$ is presented in Algorithm~\ref{algo:truncGauss}. In related work, \citet{GengDiGuKu18} studied the privacy and utility of a truncated Laplace and showed the benefits of it over Gaussian noise when satisfying an $(\epsilon,\delta)$-DP claim.

\begin{algorithm}[h!]
	\caption{$\truncGauss{\Delta,\bar{d},\tau}$; Truncated Gaussian mechanism over unknown domain with upper bound $\bar{d}$ elements, $\ell_\infty$-sensitivity $\tau$, and $\Delta$-restricted sensitivity.}
	\begin{algorithmic}
		\State \textbf{Input:} Histogram $\bbh$, $\Delta$ sensitivity, upper bound $\bar{d}$, and $\epsilon,\delta$.
		\State \textbf{Output:} Ordered set of indices and counts.
		\State Sort $h_{(1)} \geq h_{(2)} \geq \cdots \geq h_{(\bar{d}+1 )} = 0$.
		\State Set $T$ to solve the following expression
		\begin{equation}
		\delta = \Delta \left( 1 - \frac{\Phi\left( \frac{T}{\tau\sigma} \right) - \Phi\left( \frac{\tau-T}{\tau\sigma} \right)}{\Phi\left( \frac{T}{\tau\sigma} \right) - \Phi\left( \frac{-T}{\tau\sigma} \right)} \right)
		\label{eq:truncationLevel}
		\end{equation}
		\State $S = \emptyset$
		\For {$i \leq \bar{d}$}
        			\State Set $v_i = \truncNormal{T}{h_{(i)}}{\tau^2\sigma^2}$
			\If {$v_i > \tau + T$}
				\State $S \gets S \cup \{(i,v_i)\}$
			\EndIf
		\EndFor
		\State Return $S$	
	\end{algorithmic}\label{algo:truncGauss}
\end{algorithm}

\begin{lemma}
The procedure $\truncGauss{\Delta,\bar{d},\tau}$ is $\delta$-approximately $\left(0,\tfrac{\Delta}{2\sigma^2}\right)$-zCDP.
\end{lemma}
\begin{proof}
We begin by proving that for a single count $\truncNormal{T}{h_{i}}{\tau^2\sigma^2}$ ensures $\delta/\Delta$-approximate $\rho$-zCDP.  This follows a similar analysis as in \cite{GengDiGuKu18} to prove truncated Laplace noise ensures approximate DP.  WLOG, let $h_{i}' = h_{i} + \tau$.  We consider the event $E = \{h_{i}' - T \leq \truncNormal{T}{h_{i}}{\tau^2\sigma^2} \leq h_{i} + T  \}$ and compute its probability,
\[
\Pr[E] = \int_{h_{i}'- T}^{h_{i} + T} f_T(z;h_{i},\tau\sigma) dz =  \frac{\Phi\left( \frac{T}{\tau\sigma} \right) - \Phi\left( \frac{\tau-T}{\tau\sigma} \right)}{\Phi\left( \frac{T}{\tau\sigma} \right) - \Phi\left( \frac{-T}{\tau\sigma} \right)}.
\]
Further, we consider the event $E' = \{h_{i}' - T \leq \truncNormal{T}{h_{i}'}{\tau^2\sigma^2} \leq h_{i} + T  \}$.  Due to symmetry, we have $\Pr[E] = \Pr[E']$.  Note that we set $T$ so that for a given $\delta>0$ we solve \eqref{eq:truncationLevel}, hence $\Pr[E] = \Pr[E'] \leq \delta/\Delta$.

We then consider the R\'enyi divergence between two truncated Gaussians conditioned on the events that their supports overlap.  Consider the scale $\alpha > 1$ and set constant $c =  \frac{1}{\sqrt{2\pi\tau^2\sigma^2}} \frac{1}{\Phi\left( \frac{T}{\tau\sigma}\right) - \Phi\left(\frac{\tau-T}{\tau\sigma} \right)} $.
\begin{align*}
& \exp\left( (\alpha - 1) D_\alpha(\truncNormal{T}{h_{i}}{\tau^2\sigma^2}|_E || \truncNormal{T}{h_{i}'}{\tau^2\sigma^2}|_{E'} ) \right)  \\
& = c\cdot \int_{h_i'-T}^{h_i +T} \exp\left(-\alpha \left( \frac{(z - h_{i})^2}{2\tau^2\sigma^2} \right) - (1- \alpha)\left( \frac{(z - h_{i}')^2}{2\tau^2\sigma^2} \right) \right) dz \\
& = c\cdot \int_{h_i'-T}^{h_i +T} \exp\left( \frac{-1}{2\tau^2\sigma^2} \left( \left(z - (\alpha h_{i} + (1-\alpha) h_{i}') \right)^2 - (\alpha h_{i} + (1-\alpha) h_{i}')^2 + \alpha h_{i}^2 + (1-\alpha) h_{i}'^2 \right)  \right) dz \\
& = c\cdot \exp\left( \frac{\alpha(\alpha-1)}{2\sigma^2} \right) \cdot \int_{h_i'-T}^{h_i +T} \exp\left( \frac{-1}{2\tau^2\sigma^2} \left(z - (\alpha h_{i} + (1-\alpha) h_{i}') \right)^2  \right) dz \\
& = \exp\left( \frac{\alpha(\alpha-1)}{2\sigma^2} \right) \frac{\Phi\left( \frac{T-(1-\alpha)\tau }{\tau\sigma}\right) - \Phi\left(\frac{\alpha\tau - T}{\tau\sigma} \right) }{\Phi\left( \frac{T}{\tau\sigma}\right) - \Phi\left(\frac{\tau-T}{\tau\sigma} \right)} \leq \exp\left( \frac{\alpha(\alpha-1)}{2\sigma^2} \right).
\end{align*}
The last inequality holds due to the fact that the numerator is smaller than the denominator for all $\alpha > 1$.
We next consider switching the order in the divergence.  Let $c' =  \frac{1}{\sqrt{2\pi\tau^2\sigma^2}} \frac{1}{\Phi\left( \frac{T-\tau}{\tau\sigma}\right) - \Phi\left(\frac{-T}{\tau\sigma} \right)} $
\begin{align*}
& \exp\left( (\alpha - 1) D_\alpha(\truncNormal{T}{h_{i}'}{\tau^2\sigma^2}|_{E'} || \truncNormal{T}{h_{i}}{\tau^2\sigma^2}|_{E} ) \right)  \\
& = c'\cdot \int_{h_i'-T}^{h_i +T} \exp\left(-\alpha \left( \frac{(z - h_{i}')^2}{2\tau^2\sigma^2} \right) - (1- \alpha)\left( \frac{(z - h_{i})^2}{2\tau^2\sigma^2} \right) \right) dz \\
& = c'\cdot \int_{h_i'-T}^{h_i +T} \exp\left( \frac{-1}{2\tau^2\sigma^2} \left(z - (\alpha h_{i}' + (1-\alpha) h_{i}) \right)^2 - \left( (\alpha h_{i}' + (1-\alpha) h_{i})^2 + \alpha h_{i}'^2 + (1-\alpha) h_{i}^2 \right)  \right) dz \\
& = c'\cdot \exp\left( \frac{\alpha(\alpha-1)}{2\sigma^2} \right) \cdot \int_{h_i'-T}^{h_i +T} \exp\left( \frac{-1}{2\tau^2\sigma^2} \left(z - (\alpha h_{i}' + (1-\alpha) h_{i}) \right)^2  \right) dz \\
& = \exp\left( \frac{\alpha(\alpha-1)}{2\sigma^2} \right) \frac{\Phi\left( \frac{T-\alpha\tau }{\tau\sigma}\right) - \Phi\left(\frac{(1-\alpha)\tau - T}{\tau\sigma} \right) }{\Phi\left( \frac{T-\tau}{\tau\sigma}\right) - \Phi\left(\frac{-T}{\tau\sigma} \right)} \leq \exp\left( \frac{\alpha(\alpha-1)}{2\sigma^2} \right) .
\end{align*}
Hence, truncation only reduces the R\'enyi divergence between two shifted Gaussians.  We then have 
\[
D_\alpha(\truncNormal{T}{h_{i}}{\tau^2\sigma^2}|_E || \truncNormal{T}{h_{i}'}{\tau^2\sigma^2}|_{E'} ) \leq \frac{\alpha}{2\sigma^2}.
\]
Now consider the setting where we have access to the full histogram and use the truncated Gaussian noise for each count.  In this case, we can apply zCDP composition on the $\Delta$ counts that changed in neighboring histograms, hence by applying a union bound over the events $E$ and $E'$ for each count, the result is a $\delta$-approximate $\tfrac{\Delta}{2\sigma^2}$-zCDP mechanism.  We then apply post-processing to only return elements that are above the threshold $\tau + T$.  Recall that post-processing does not increase the privacy parameters.  Note that any element whose true count is less than $\tau$ can never appear in the result, due to the truncated noise.  In a neighboring dataset a user can cause an element whose true count is 0 to have new count of at most $\tau$, and thus cannot appear in the result from the neighboring dataset.   
\end{proof}
Note that in repeated calls to $\truncGauss{\Delta_i,\bar{d}_i,\tau_i}$, where each round $i$ may have different parameter values, we can apply zCDP composition to get an overall privacy guarantee over the entire interaction.

See Figure~\ref{fig:truncGauss} for a plot of the truncation level $T$ from \eqref{eq:truncationLevel} when compared to the truncation level in $\rTE{\Delta,\bar{d},\tau}$ for various $\ell_0$-sensitivities $\Delta$ and a given $(\diffp,\delta)$-DP guarantee.  Note that the truncation level in $\rTE{\Delta,\bar{d},\tau}$ has Laplace noise added to it, whereas the truncation level in $\truncGauss{\Delta,\bar{d},\tau}$ is fixed. 
\begin{figure}[h]
\centering
\includegraphics[width=12cm]{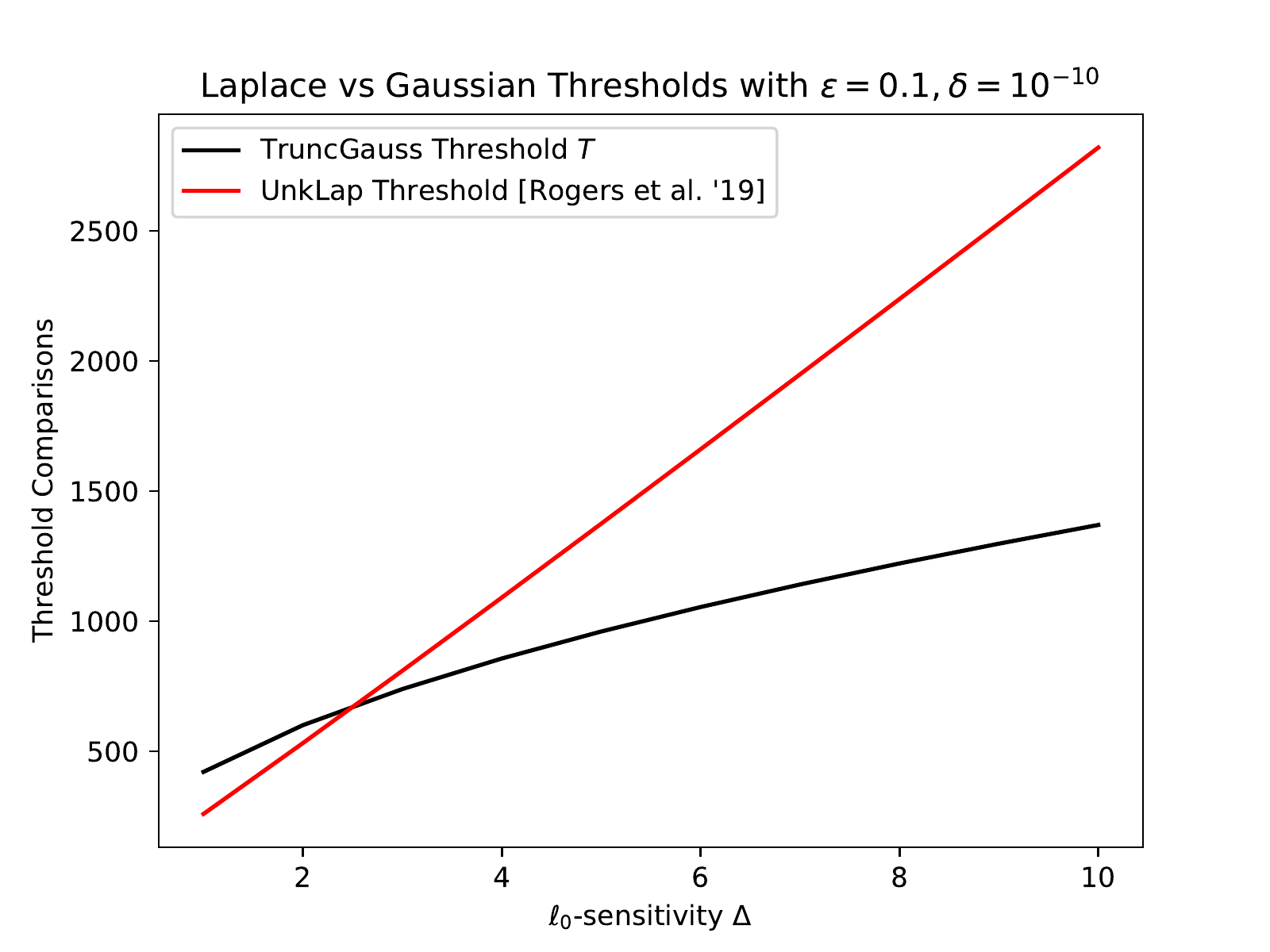}
\caption{Comparing the truncation level $T$ computed in \eqref{eq:truncationLevel} with the threshold given in $\rTE{\Delta,\bar{d},\tau}$.  We fix $\ell_\infty$-sensitivity $\tau = 1$ as well as the overall privacy parameters $(\diffp = 0.1,\delta= 10^{-10})$.  \label{fig:truncGauss}}
\end{figure}

\section{Conclusion \label{sec:conclusion}} 

In this work, we have considered the impact to the overall privacy loss when an analyst can select different types of private mechanisms, including pure DP, BR, and CDP.  We developed more general bounds on the privacy loss when the analyst can select the type of mechanism and the privacy parameters adaptively at each round, as long as the privacy parameters are predetermined and selected without replacement, thus unifying different works on bounding the privacy loss under composition.  We then considered the optimal DP composition bounds when an analyst may select at most $m$ pure DP mechanisms and the remaining $k-m$ mechanisms are BR, both with the same privacy parameter, i.e. the homogeneous case.  In the non-adaptive setting, we computed the optimal DP composition bound, which allows us to smoothly interpolate between the bounds from \cite{KairouzOhVi17}, i.e. $m = k$ and the bounds from \cite{DongDuRo19}, i.e. $m = 0$.   In the adaptive homogeneous setting, we showed that the placement of a single BR mechanism does not change the composition bound, but provided an example with two BR mechanisms where ordering does impact privacy loss.

Although our optimal privacy loss bounds only apply when we combine pure DP and BR mechanisms, we also studied the Gaussian mechanism (a standard CDP mechanism) with different types of analyses and how it compares to the Laplace mechanism (a standard DP mechanism) for releasing private histograms.  Given an $\ell_0$ and $\ell_\infty$-sensitivity for a histogram, the best mechanism, i.e. the one with the smallest privacy loss with a fixed standard deviation, is not obvious, especially if one considers composition.  We then provided new algorithms that utilize Gaussian noise for the existing algorithms in LinkedIn's privacy system.  

We see lots of interesting directions for future work.  As discussed in \citet{DongDuRo19}, determining the computational complexity of the optimal DP composition for adaptively selected BR mechanisms, even in the homogenous case, is incredibly interesting.  We found that with 3 mechanisms, 2 of which are BR, the $\arg\sup$ values were weighted sums of $\diffp$ and $\diffp_g$.  Does this generalize beyond $k = 3$?  Further, for composition of heterogenous DP mechanisms, does the ordering of the privacy parameters matter if they can be adaptively selected?


\paragraph{Acknowledgements}
We would like to thank Adrian Cardoso, Koray Mancuhan, Guillaume Saint Jacques, Reza Hosseini, and Seunghyun Lee for their helpful feedback on this work.  Also, special thanks to David Durfee and Jinshuo Dong for early conversations about this work.

\clearpage

\bibliography{bib}
\bibliographystyle{abbrvnat}

\clearpage
\appendix
\section{Omitted Proofs of Section ~\ref{sec:concentration}\label{app:ConcentrationProof}} 
\begin{proof}[Proof of Lemma~\ref{lem:concentration}]
We first show the following for $\lambda > 0$:
\begin{equation}
\E\left[\exp\left(\lambda \sum_{i=0}^k X_i\right) \right] \leq \exp(\lambda^2 b^2/2)
\label{eq:submartingale}
\end{equation}
By our hypothesis, we have that 
\[
\E[e^{\lambda X_i - \lambda^2 B_i^2/2} \mid \cF_{i-1}] \leq 1 \quad \forall i
\]
Furthermore, we have that $M_k = \exp\left( \lambda \sum_{i=0}^k X_i - \lambda^2/2 \sum_{i=0}^k B_i^2 \right)$ is a supermartingale, due to the following
\begin{align*}
\E[M_k|\cF_{k-1}] &= \E\left[\exp\left(\lambda \sum_{i=0}^k X_i - \lambda^2/2 \sum_{i=0}^k B_i^2\right) \mid \cF_{k-1} \right] \\
& = \exp\left(\lambda \sum_{i=0}^{k-1} X_i - \lambda^2/2 \sum_{i=0}^{k-1} B_i^2\right) \E\left[ e^{\lambda X_k - \lambda^2 B_k^2/2} \mid \cF_{k-1} \right] \\
& \leq  \exp\left(\lambda \sum_{i=0}^{k-1} X_i - \lambda^2/2 \sum_{i=0}^{k-1} B_i^2\right) = M_{k-1}.
\end{align*}
Hence, we have $\E[M_k] \leq 1 $.  Using our assumption that $\sum_{i=0}^k B_i^2 \leq b$, we have the following
\[
\E\left[\exp\left(\lambda \sum_{i=0}^k X_i - \lambda^2 b^2/2\right) \right] \leq \E\left[\exp\left(\lambda \sum_{i=0}^k X_i - \lambda^2 \sum_{i=0}^k B_i^2/2\right) \right] \leq 1
\]
and hence \eqref{eq:submartingale} holds.  We can now prove the result using a standard argument involving Markov's inequality.
\[
\Pr[\sum_{i=0}^k X_i > \beta] = \Pr[e^{\lambda \sum_{i=0}^k X_i - \lambda^2b^2/2} \geq e^{\lambda \beta - \lambda^2b^2/2}] \leq \frac{1}{e^{\lambda \beta - \lambda^2b^2/2}}
\]
We then set $\lambda =\tfrac{\beta}{b^2}$ to get the result.
\end{proof}
\section{Omitted Proofs of Section ~\ref{sec:optimal_nonadaptive_omitted_proofs}} 
\subsection{Proof of Lemma ~\ref{lem:deriv_zero}\label{lem:deriv_zero:proof}}
\begin{proof}
To simplify notation, we use $k' = k - m$ and $q_{\diffp,t} = q_t$.  We have the following recurrence relation for $F_\ell(t) = F_{\ell -1}(t) + f_\ell(t)$ where
\[
f_\ell(t) = \left( 1 - e^{\diffp_g - \diffp(m - \ell) - t k' } \right) \sum_{ \substack{i + 2j = \ell \\ i \in \{0,1 \cdots, k' \} \\ j \in \{0, 1, \cdots, m \} } } \alpha_{i,j} q_{t}^{k' - i} ( 1- q_{t})^j
\]
We then prove the statement by using induction $F_\ell'(t) = F_{\ell - 1}'(t) + f_\ell'(t)$.  We start with the base case,
\begin{align*}
F_0'(t) &= \left( k' \cdot e^{\diffp_g - \diffp m - t k' } \right) \alpha_{0,0} q_t^{k'} + \left( 1-e^{\diffp_g - \diffp m - t k' } \right) \alpha_{0,0} k' q_t^{k'-1} q_t' \\
& = \frac{k' \alpha_{0,0}}{1 - e^{-\diffp}}q_t^{k'-1} \left( e^{\diffp_g - \diffp m - t k' } (1 - e^{t - \diffp}) - e^{t-\diffp} \left( 1-e^{\diffp_g - \diffp m - t k' } \right) \right) \\
& =  \frac{k' \alpha_{0,0}}{1 - e^{-\diffp}}q_t^{k'-1} \left( e^{\diffp - \diffp m - t k'}  - e^{t - \diffp} \right)
\end{align*}
We now present the derivative of $f_\ell(t)$.  Note that we drop the condition in the summation where $i \in \{0,\cdots, k' \}$ and $j \in \{0,\cdots, m \}$ to ease the notation.
\begin{align*}
f_\ell'(t) &= k'e^{\diffp_g - \diffp(m - \ell) - t k'} \sum_{i+2j = \ell} \alpha_{i,j} q_{t}^{k' - i} ( 1- q_{t})^i \\
& \qquad + \left( 1 - e^{\diffp_g - \diffp(m - \ell) - t k' } \right) \sum_{i+2j = \ell} \alpha_{i,j} \left( (k' - i) q_{t}^{k' - i-1} ( 1- q_{t})^i q_t' -i q_t^{k-i} (1 - q_t)^{i-1}q_t'   \right) \\
& = \sum_{i+2j = \ell} \alpha_{i,j}  q_{t}^{k'-i-1} (1 - q_t)^{i-1}\\
& \qquad  \left( k'e^{\diffp_g - \diffp(m - \ell) - t k'}  q_t (1- q_t)  + q_t' \left( 1 - e^{\diffp_g - \diffp(m - \ell) - t k' } \right) \left( (k'-i) (1-q_t) - i q_t \right) \right)
\end{align*}
We now factor out a $1/(1- e^{-\diffp})^2$ and the inner term becomes.
\begin{align*}
(*) \defeq & k'e^{\diffp_g - \diffp(m - \ell) - t k'} \left( 1 - e^{t-\diffp} \right) \left( e^{t-\diffp} - e^{-\diffp} \right) \\
& \qquad - e^{t-\diffp} \left( 1 - e^{\diffp_g - \diffp(m - \ell) - t k' } \right) \left( (k'-i) (e^{t-\diffp} - e^{-\diffp} ) - i ( 1 - e^{t-\diffp}) \right) \\
& = k' e^{\diffp_g - \diffp(m - \ell) - t k'} \left(e^{t-\diffp} - e^{-\diffp} -e^{2(t - \diffp)} + e^{t - 2 \diffp} \right) \\
& \qquad -\left(  e^{t-\diffp}  - e^{\diffp_g - \diffp(m - \ell +1) - t (k'-1) } \right) \left( (k'-i) (e^{t-\diffp} - e^{-\diffp} ) - i ( 1 - e^{t-\diffp}) \right) \\
& =  k' \left( e^{\diffp_g - \diffp(m - \ell+1) - t (k' -1)} - e^{\diffp_g - \diffp(m - \ell+1) - t k'} - e^{\diffp_g - \diffp(m - \ell+2) - t (k' -2)} + e^{\diffp_g - \diffp(m - \ell + 2) - t (k'-1)} \right) \\
& \qquad - (k'-i) \left( e^{2(t-\diffp)}  - e^{\diffp_g - \diffp(m - \ell +2) - t (k'-2) } - e^{t-2\diffp}  + e^{\diffp_g - \diffp(m - \ell +2) - t (k'-1) } \right) \\
& \qquad + i \left( e^{t-\diffp}  - e^{\diffp_g - \diffp(m - \ell +1) - t (k'-1) } - e^{2(t-\diffp)}  + e^{\diffp_g - \diffp(m - \ell +2) - t (k'-2) }\right)
\end{align*}
We now use the inductive claim to prove the statement.  We use the fact that $\alpha_{i,j} (k' - i) = \alpha_{i+1,j} (i+1)$
\begin{align*}
F'_{\ell}(t) & = \frac{1}{1 - e^{-\diffp}} \left( e^{\diffp_g - \diffp(m -\ell + 1) - tk'} - e^{t-\diffp} \right) \sum_{ i + 2j = \ell -1 }  \alpha_{i,j} (k' - i)  q_{\diffp,t}^{k' - i - 1} (1 - q_{\diffp,t})^i \\
& \qquad + \sum_{i+2j = \ell} \alpha_{i,j}  \frac{q_{t}^{k'-i-1} (1 - q_t)^{i-1}}{(1 - e^{-\diffp})^2}\cdot(*) \\
& = \frac{1}{(1 - e^{-\diffp})^2} \sum_{i + 2j = \ell-1} (i+1) \alpha_{i+1,j} \left( e^{\diffp_g - \diffp(m -\ell +1) - tk'} - e^{t-\diffp} \right)q_t^{k' - i - 1} (1 - q_t)^{i} (1 - e^{-\diffp}) \\
& \qquad+  \sum_{i+2j = \ell} \alpha_{i,j}  \frac{q_{t}^{k'-i-1} (1 - q_t)^{i-1}}{(1 - e^{-\diffp})^2}\cdot(*) \\
& = \frac{1}{(1 - e^{-\diffp})^2} \sum_{i+2j = \ell} \alpha_{i,j}  \left( e^{\diffp_g - \diffp(m -\ell+1) - tk'} - e^{t-\diffp} \right)q_t^{k' - i} (1 - q_t)^{i-1} (1 - e^{-\diffp}) i \\
& \qquad + \frac{1}{(1 - e^{-\diffp)^2}} \sum_{i+2j = \ell} \alpha_{i,j}  q_{t}^{k'-i-1} (1 - q_t)^{i-1}\cdot(*) \\
& =  \frac{1}{(1 - e^{-\diffp})^2} \sum_{i+2j = \ell} \alpha_{i,j}  q_{t}^{k'-i-1} (1 - q_t)^{i-1} \left( i \left( e^{\diffp_g - \diffp(m -\ell+1) - tk'} - e^{t-\diffp} \right) ( 1- e^{t-\diffp}) +(*) \right)
\end{align*}
We now expand the inner term by combining like terms with the $i$ coefficient.  
\begin{align*}
& i \left( \left(e^{\diffp_g - \diffp(m -\ell+1) - tk'} - e^{t-\diffp} \right) ( 1- e^{t-\diffp}) + \left( e^{t-\diffp}  - e^{\diffp_g - \diffp(m - \ell +1) - t (k'-1) } - e^{2(t-\diffp)}  + e^{\diffp_g - \diffp(m - \ell +2) - t (k'-2) }\right) \right) \\
& = i \left( e^{\diffp_g - \diffp(m -\ell+1) - tk'} - e^{\diffp_g - \diffp(m -\ell+2) - t(k'-1)} - e^{\diffp_g - \diffp(m - \ell +1) - t (k'-1) }+ e^{\diffp_g - \diffp(m - \ell +2) - t (k'-2) } \right)
\end{align*}
Note that this is the same term as the negative of the coefficient on $k'$. We then combine the terms with $(k'-i)$ coefficient.
\begin{align*}
- (k'-i) \left( e^{\diffp_g - \diffp(m -\ell+1) - tk'} - e^{\diffp_g - \diffp(m -\ell+2) - t(k'-1)} - e^{\diffp_g - \diffp(m - \ell +1) - t (k'-1) }+ e^{\diffp_g - \diffp(m - \ell +2) - t (k'-2) }\right) \\
- (k'-i)\left( e^{2(t-\diffp)}  - e^{\diffp_g - \diffp(m - \ell +2) - t (k'-2) } - e^{t-2\diffp}  + e^{\diffp_g - \diffp(m - \ell +2) - t (k'-1) } \right) \\
= - (k'-i) \left(e^{2(t-\diffp)}- e^{t-2\diffp}  + e^{\diffp_g - \diffp(m -\ell+1) - tk'}  - e^{\diffp_g - \diffp(m - \ell +1) - t (k'-1) }  \right) \\
= (k'-i) (e^{t-\diffp} - e^{-\diffp} ) \left( e^{\diffp_g - \diffp(m-\ell) - tk'} - e^{t - \diffp} \right)
\end{align*}
Putting this altogether, we have the following.
\begin{align*}
F'_{\ell}(t) & =  \frac{1}{(1 - e^{-\diffp})^2} \sum_{i+2j = \ell} \alpha_{i,j}  q_{t}^{k'-i-1} (1 - q_t)^{i-1} (k'-i) (e^{t-\diffp} - e^{-\diffp} ) \left( e^{\diffp_g - \diffp(m-\ell) - tk'} - e^{t - \diffp} \right) \\
& =  \frac{1}{1 - e^{-\diffp}} \left( e^{\diffp_g - \diffp(m-\ell) - tk'} - e^{t - \diffp} \right) \sum_{i+2j = \ell} \alpha_{i,j} (k-i) q_{t}^{k'-i-1} (1 - q_t)^{i} 
\end{align*}
This proves the statement.
\end{proof}


\section{Omitted Proofs of Section ~\ref{sec:OptAdaptive}} 

\subsection{Proof of Proposition~\ref{prop:sbr-ord}:\label{lem:sbr-ord:proof}}

Before we prove Proposition~\ref{prop:sbr-ord} we need the following two lemmas.  The first provides an identity for reducing elemental terms to expressions independent of $t$ (similar to an identity used previously in \citet{DongDuRo19}).  The next lemma provides an expansion of $\deltaopt$ for $k$ pure DP mechanisms $\mdp(\diffp)$.

\begin{lemma}[Reduction Identity]\label{lem-sbr-ord-1} Given $\alpha\in\reals, \diffp>0,$ and $t\in\left[0,\diffp\right],$ we have the identity:
\begin{align*}
&q_{t, \diffp} \left[ 1 - e^{\alpha - t}\right]_{+} + \left(1-q_{t, \diffp}\right) \left[1 - e^{\alpha + \diffp - t}\right]_{+} =
\begin{cases}
  0 & \text{if $\alpha - t$ $\geq$ $0$} \\
   q_{t, \diffp} \left(1-e^{\alpha - t }\right) & \text{if $\alpha\leq t\leq\alpha+\diffp$} \\
   1 - e^{\alpha} & \text{otherwise}
\end{cases}.
\end{align*}
\end{lemma}
\begin{proof}
When ~$\alpha - t \geq 0$ we have ~$ 1-e^{\alpha - t } < 0$ and ~$1-e^{\alpha + \diffp - t } < 0$, hence
\begin{align*}
q_{t, \diffp} \left[ 1 - e^{\alpha - t}\right]_{+} + \left(1-q_{t, \diffp}\right) \left[1 - e^{\alpha + \diffp - t}\right]_{+} = 0.
\end{align*}
When ~$\alpha\leq t\leq\alpha+\diffp$ we have ~$1-e^{\alpha + \diffp - t } < 0$, hence 
\begin{align*}
q_{t, \diffp} \left[ 1 - e^{\alpha - t}\right]_{+} + \left(1-q_{t, \diffp}\right) \left[1 - e^{\alpha + \diffp - t}\right]_{+} = q_{t, \diffp} \left[ 1 - e^{\alpha - t}\right]_{+}.
\end{align*}
When $\alpha + \diffp - t \leq 0 $ we have
\begin{align*}
q_{t, \diffp} \left[ 1 - e^{\alpha - t}\right]_{+} + \left(1-q_{t, \diffp}\right) \left[1 - e^{\alpha + \diffp - t}\right]_{+} & = q_{t, \diffp} \left(1 - e^{\alpha - t}\right)+ \left(1-q_{t, \diffp}\right) \left(1 - e^{\alpha + \diffp - t}\right) \\
&= q_{t, \diffp} + 1 - q_{t, \diffp} - \left[q_{t, \diffp} + e^{\diffp}\left(1-q_{t, \diffp}\right)\right] e^{\alpha - t} \\
&= 1 - \left[\frac{1-e^{t-\diffp}}{1-e^{-\diffp}} + e^{\diffp}\left(1 -\frac{1-e^{t-\diffp}}{1-e^{-\diffp}}\right)\right] e^{\alpha - t} \\
&= 1- \frac{e^{t} - e^{t-\diffp}}{1-e^{-\diffp}}  e^{\alpha - t}\\
&= 1-e^{\alpha}.
\end{align*}
\end{proof}

\begin{lemma}\label{lem-sbr-ord-2}
For $\ell\in \naturals$, $x\in\reals$, $\diffp \geq 0$, and $\diffp_{i} = \diffp$ for $i\in\{0, \ldots, \ell\}$ we define:
\begin{align*}
\delta_{\ell}\left(x\right) \defeq \delta_{\opt}\left(\mdp(\diffp_{1}),\ldots,\mdp(\diffp_{\ell}); x\right).\
\end{align*}
Then we have the following identity for some constants $\lambda_{\ell, i} \in \reals$:
\begin{align*}
& \delta_{\ell}\left(x\right) = \sum_{i\in\left\{0,\ldots,\ell\right\}} \lambda_{\ell, i}\left[1-e^{\left(2i-\ell\right)\diffp + x}\right]_{+} .
\end{align*}

\end{lemma}

\begin{proof}
This is a straightforward induction on $\ell$. The base case $\ell=1$ is apparent from the recurrence in Lemma \ref{lem:adap_recursion_hetero}. If we have $\ell > 1$ and the identity holds for $\ell ' < \ell$ then we have:
\begin{align*}
\vphantom{\sum_{i\in\left\{0,\ldots,\ell -1\right\}}} \delta_{\ell}\left(x\right) =& \quad \qedp  \delta_{\ell-1}\left(x-\diffp\right) +  \left(1-\qedp\right) \delta_{\ell-1}\left(x+\diffp\right) \\
= & \quad \qedp  \sum_{i\in\left\{0,\ldots,\ell -1\right\}} \lambda_{\ell - 1, i}\left[1-e^{\left(2i-\ell+1\right)\diffp + x - \diffp}\right]_{+} \\
& \qquad + \quad \left(1-\qedp\right)\sum_{i\in\left\{0,\ldots,\ell-1\right\}} \lambda_{\ell -1 , i}\left[1-e^{\left(2i-\ell +1\right)\diffp + x + \diffp}\right]_{+} \\
= & \quad \qedp \lambda_{\ell-1, 0} \left[1 - e^{\left(-\ell\diffp + x\right)}\right]_{+} \\
& \qquad + \quad \sum_{i\in\{1,\ldots,\ell-1\}} \left( \qedp\lambda_{\ell-1,i}+ \left(1-\qedp \right)\lambda_{\ell-1, i-1} \right) \left[1-e^{\left(2i-\ell+1\right)\diffp + x - \diffp}\right]_{+} \\
& \qquad + \quad (1-\qedp) \lambda_{\ell-1, \ell-1} \left[1 - e^{\left(\ell\diffp + x\right)}\right]_{+} \\
= & \quad \sum_{i\in\left\{0,\ldots,\ell\right\}} \lambda_{\ell, i}\left[1-e^{\left(2i-\ell\right)\diffp + x}\right]_{+} .
\end{align*}
\end{proof}

We can now complete the proof of Proposition~\ref{prop:sbr-ord}.

\begin{proof}[Proof of Proposition~\ref{prop:sbr-ord}] 
To prove this we induct on the number, $N$, of $\diffp$-DP mechanisms. For the base case $\left(N=2\right)$ consider the expansion of the first two terms (using Lemma \ref{lem:adap_recursion_hetero}):
\begin{align*}
 \deltaopt\left(\vec{\cA}_{2}; \diffp_{g}\right) =\quad& \sup_{t\in[0,\diffp]}\biggl\{ q_{t, \diffp}\left[ \qedp \left[1 - e^{\diffp_{g} - \diffp - t}\right]_{+} + \left(1-\qedp\right)  \left[1 - e^{\diffp_{g} + \diffp - t}\right]_{+}\right]\\
& \qquad+ \left(1-q_{t, \diffp}\right) \left[\qedp \left[1 - e^{\diffp_{g} - t}\right]_{+} + \left(1-\qedp\right)  \left[1 - e^{\diffp_{g} + 2\diffp - t}\right]_{+} \right] \biggr\} \\
\deltaopt\left(\vec{\cB}_{2}; \diffp_{g}\right)=\quad& \qedp\sup_{\ton\in[0,\diffp]}{\biggl\{ q_{\ton, \diffp} \left[1 - e^{\diffp_{g} - \diffp - \ton}\right]_{+} + \left(1-q_{\ton, \diffp}\right)  \left[1 - e^{\diffp_{g} - \ton}\right]_{+} \biggr\}}  \hspace{0.5em} \\
&+ \left(1-\qedp\right) \sup_{\ttw\in[0,\diffp]}\biggl\{ {q_{\ttw, \diffp} \left[1 - e^{\diffp_{g} + \diffp - \ttw}\right]_{+} + \left(1-q_{\ttw, \diffp}\right)  \left[1 - e^{\diffp_{g} + 2\diffp - \ttw}\right]_{+} \biggr\}}
\end{align*}

Where the terms with positive coefficients that depend on $\diffp$ must be $0$, so the $\sup\limits_{\ttw\in[0,\diffp]}$ term in $\deltaopt\left(\cB_{2}; \diffp_{g}\right)$ disappears and equality is apparent.

Now, consider some $N>2$. Suppose that for $k<N$ we have $\deltaopt\left(\vec{\cA}_{k}; \diffp_{g}\right)= \deltaopt\left(\vec{\cB}_{k}; \diffp_{g}\right)$ for all $\vec{\cA}_{k}, \vec{\cB}_{k}$. Then to show ~$\deltaopt\left(\vec{\cA}_{N}; \diffp_{g}\right) = \deltaopt\left(\vec{\cB}_{N}; \diffp_{g}\right)$ there are two nontrivial cases to consider:
\begin{itemize}
\item[1.] $\cA_{1}, \cB_{1} = \mdp(\diffp)$. Here we simply expand the first $\diffp$-DP mechanism on each side using definitions and apply the inductive hypothesis.
\item[2.] $\cA_{1} = \mbr(\diffp)$ and $\cB_{1} = \mdp(\diffp)$. Here the proof is slightly more involved. First we expand the first two terms of each side:
\end{itemize}

Now consider the second case.
\begin{align*}
\deltaopt\left(\vec{\cA}_{N}; \diffp_{g}\right) = &\quad \sup_{t\in[0,\diffp]}\biggl\{ q_{t, \diffp}\left[ \qedp \deltaopt\left(\mdp(\diffp),\ldots,\mdp(\diffp); \diffp_{g} - \diffp - t\right) \hspace{0.25em} \right. \\
&\vphantom{\bigg\{} \left.\qquad\quad+ \left(1-\qedp\right)  \deltaopt\left(\mdp(\diffp),\ldots,\mdp(\diffp); \diffp_{g} + \diffp - t\right) \right] \hspace{0.25em}\\
&\vphantom{\bigg\{} \qquad\quad+ \left(1-q_{t, \diffp}\right) \left[\qedp \deltaopt\left(\mdp(\diffp),\ldots,\mdp(\diffp); \diffp_{g} - t\right) \hspace{0.25em} \right.\\
& \left. \qquad\quad+ \left(1-\qedp\right)  \deltaopt\left(\mdp(\diffp),\ldots,\mdp(\diffp); \diffp_{g} + 2\diffp - t\right) \right] \biggr\} .\\
\deltaopt\left(\vec{\cB}_{N}; \diffp_{g}\right) =& \quad \qedp\sup_{\ton\in[0,\diffp]}\biggl\{ q_{\ton, \diffp} \deltaopt\left(\mdp(\diffp),\ldots,\mdp(\diffp); \diffp_{g} - \diffp - \ton\right)\hspace{0.25em} \\
&\qquad\qquad \quad + \left(1-q_{\ton, \diffp}\right)  \deltaopt\left(\mdp(\diffp),\ldots,\mdp(\diffp); \diffp_{g} - \ton\right) \biggr\} \hspace{0.5em} \\
&\quad+\quad \left(1-\qedp\right) \sup_{\ttw\in[0,\diffp]}\biggl\{ q_{\ttw, \diffp} \deltaopt\left(\mdp(\diffp),\ldots,\mdp(\diffp); \diffp_{g} + \diffp - \ttw\right) \hspace{0.25em} \\
 &\qquad\qquad\quad + \left(1-q_{\ttw, \diffp}\right)  \deltaopt\left(\mdp(\diffp),\ldots,\mdp(\diffp); \diffp_{g} + 2\diffp - \ttw\right) \biggr\}.
\end{align*}

Where we have applied the inductive hypothesis to $\vec{\cB}_{N}$ after expanding the first term to ensure that the second mechanism is $\diffp$-BR wlog. Now, applying our formula from Lemma~\ref{lem-sbr-ord-2} we have:
\begin{align*}
\deltaopt\left(\vec{\cA}_{N}; \diffp_{g}\right) = \quad & \sup_{t\in[0,\diffp]}\biggl\{ q_{t, \diffp}\left( \qedp \sum_{i\in\left[\ell\right]} \lambda_{i}\left[1-e^{ \diffp_{g} + \left(-1 + 2i - \left(N-2\right)\right)\diffp- t}\right]_{+} \right. \\
&\qquad+ \left. \left(1-\qedp\right)  \sum_{i\in\left[\ell\right]} \lambda_{i}\left[1-e^{ \diffp_{g} + \left(1 + 2i - \left(N-2\right)\right)\diffp- t}\right]_{+} \right) \\
&\qquad+ \left(1-q_{t, \diffp}\right) \left(\qedp \sum_{i\in\left[\ell\right]} \lambda_{i}\left[1-e^{ \diffp_{g} + \left(2i - \left(N-2\right)\right)\diffp- t}\right]_{+} \right.\\
&\qquad+ \left. \left(1-\qedp\right)  \sum_{i\in\left[\ell\right]} \lambda_{i}\left[1-e^{ \diffp_{g} + \left(2 + 2i - \left(N-2\right)\right)\diffp- t}\right]_{+} \right) \biggr\}. \\
\deltaopt\left(\vec{\cB}_{N}; \diffp_{g}\right) =\quad &  \qedp\sup_{\ton\in[0,\diffp]}\biggl\{ q_{\ton, \diffp} \sum_{i\in\left[\ell\right]} \lambda_{i}\left[1-e^{ \diffp_{g} + \left(-1 + 2i - \left(N-2\right)\right)\diffp- \ton}\right]_{+} \\
& \qquad\qquad\quad+ \left(1-q_{\ton, \diffp}\right)  \sum_{i\in\left[\ell\right]} \lambda_{i}\left[1-e^{ \diffp_{g} + \left(2i - \left(N-2\right)\right)\diffp- \ton}\right]_{+}  \biggr\} \hspace{0.5em} \\
& + \left(1-\qedp\right) \sup_{\ttw\in[0,\diffp]}\biggl\{ q_{\ttw, \diffp} \sum_{i\in\left[\ell\right]} \lambda_{i}\left[1-e^{ \diffp_{g} + \left(1 + 2i - \left(N-2\right)\right)\diffp- \ttw}\right]_{+}  \\
& \qquad\qquad\quad+ \left(1-q_{\ttw, \diffp}\right)  \sum_{i\in\left[\ell\right]} \lambda_{i}\left[1-e^{ \diffp_{g} + \left(2 + 2i - \left(N-2\right)\right)\diffp- \ttw}\right]_{+} \biggr\}.
\end{align*}

Comparing terms with like coefficients we see that the reduction identity applies and so each pair of sums can be reduced. The reduction identity yields an expression that depends on $t$ only if $t \in \left(\alpha,\alpha + \diffp\right)$, therefore only one reduced term from each sum can depend on $t$ as the $\alpha = \diffp_{g} + \ell\diffp -t $ terms increase in increments of size $2\diffp$. Furthermore, because the coefficients of $\diffp$ in the first exponent in each pair of sums (to which we are applying the reduction) have the same parity, the same term $i^{*}$ depends on $t$ in both. Thus, collecting the terms constant in $t$ into constants $C_{1}(\diffp), C_{2}(\diffp) \in \reals$, we can write:

\begin{align}
\deltaopt\left(\vec{\cA}_{N}; \diffp_{g}\right) =&\quad \sup_{t\in[0,\diffp]}\biggl\{ \qedp \left(C_{1}(\diffp) + q_{t, \diffp} \lambda_{i^{*}}\left[1-e^{ \diffp_{g} + \left(-1 + 2i^{*} - \left(N-2\right)\right)\diffp- t}\right]_{+}\right) \hspace{0.5em} \nonumber \\
& \qquad\qquad+ \left(1-\qedp\right) \left(C_{2}(\diffp) + q_{t, \diffp} \lambda_{i^{*}}\left[1-e^{\diffp_{g} + \left(1+ 2i^{*} - \left(N-2\right)\right)\diffp- t}\right]_{+}\right)\biggr\}. \label{sbr-1} \\
\deltaopt\left(\vec{\cB}_{N}; \diffp_{g}\right) =&\quad \qedp\sup_{\ton\in[0,\diffp]}{\biggl\{ C_{1}(\diffp) + q_{t, \diffp} \lambda_{i^{*}}\left[1-e^{ \diffp_{g} + \left(-1 + 2i^{*}- \left(N-2\right)\right)\diffp- \ton}\right]_{+}\biggr\}} \hspace{0.5em} \nonumber \\
&\qquad\qquad+ \left(1-\qedp\right) \sup_{\ttw\in[0,\diffp]}\biggl\{ C_{2}(\diffp) + q_{\ttw, \diffp} \lambda_{i^{*}}\left[1-e^{\diffp_{g} + \left(1+ 2i^{*} - \left(N-2\right)\right)\diffp- \ttw}\right]_{+}\biggr\}. \label{sbr-2}
\end{align}

Then both remaining terms in $\deltaopt\left(\vec{\cB}_{N}; \diffp_{g}\right)$ achieve the supremum for the same argument $  t^{*}=\ton=\ttw$.  Fixing $t=t^{*}$ in (\ref{sbr-1}) and $\ton=\ttw=t^{*}$ in (\ref{sbr-2}) and comparing it is clear that $\deltaopt\left(\vec{\cA}_{N}; \diffp_{g}\right) \geq  \deltaopt\left(\vec{\cB}_{N}; \diffp_{g}\right)$. By the triangle inequality, we also have $\deltaopt\left(\vec{\cB}_{N}; \diffp_{g}\right) \geq \deltaopt\left(\vec{\cA}_{N}; \diffp_{g}\right)$, which proves the result.
\end{proof}

\subsection {Proof of Lemma~\ref{lem:xyz} \label{lem:xyz:proof}}
\begin{proof}
We start with the ordering with the DP mechanism at the end.
\begin{align*}
& \deltaopt(\mbr(\diffp), \mbr(\diffp), \mdp(\diffp) ;\diffp_g) \\
& = \sup_{t_{1,1} \in [0,\diffp]} \left\{q_{\diffp,t_{1,1}}\deltaopt(\mbr(\diffp), \mdp(\diffp) ;\diffp_g- t_{1,1}) + ( 1- q_{\diffp,t_{1,1}}) \deltaopt(\mbr(\diffp), \mdp(\diffp) ;\diffp_g + \diffp - t_{1,1}) \right\}\\
& = \sup_{t_{1,1}} \left\{q_{\diffp,t_{1,1}} \sup_{t_{1,2}} \left\{ q_{\diffp,t_{1,2}} \deltaopt(\mdp(\diffp) ;\diffp_g - t_{1,1} - t_{1,2}) + (1 -  q_{\diffp,t_{1,2}}) \deltaopt(\mdp(\diffp) ;\diffp_g + \diffp - t_{1,1} - t_{1,2}) \right\} \right. \\
& \left. + ( 1- q_{\diffp,t_{1,1}})\sup_{t_{2,2}} \left\{ q_{\diffp,t_{2,2}} \deltaopt(\mdp(\diffp) ;\diffp_g + \diffp - t_{1,1} - t_{2,2})  + (1 - q_{\diffp,t_{2,2}} )\deltaopt(\mdp(\diffp) ;\diffp_g + 2\diffp - t_{1,1} - t_{2,2}) \right\}  \right\} \\
& = \sup_{t_{1,1}, t_{1,2}, t_{2,2}} \left\{ q_{\diffp,t_{1,1}} q_{\diffp,t_{1,2}} q_{2\diffp,\diffp} [1 - e^{\diffp_g - \diffp - t_{1,1} - t_{1,2} }]_+ + q_{\diffp,t_{1,1}} q_{\diffp,t_{1,2}} (1 - q_{2\diffp,\diffp})[1 - e^{\diffp_g + \diffp - t_{1,1} - t_{1,2}}]_+ \right. \\
& \qquad\qquad \left.  + q_{\diffp,t_{1,1}} (1 -  q_{\diffp,t_{1,2}}) q_{2\diffp,\diffp} [1 - e^{\diffp_g - t_{1,1} - t_{1,2} }]_+   +  (1 - q_{\diffp,t_{1,1}}) q_{\diffp,t_{2,2}} q_{2\diffp,\diffp} [1 - e^{\diffp_g - t_{1,1} - t_{2,2} }]_+ \right. \\
& \qquad \qquad\left. + (1 - q_{\diffp,t_{1,1}}) ( 1- q_{\diffp,t_{2,2}}) q_{2\diffp,\diffp} [1 - e^{\diffp_g + \diffp - t_{1,1} - t_{2,2} }]_+ \right\}
\end{align*}
\begin{align*}
& \deltaopt(\mdp(\diffp), \mbr(\diffp), \mbr(\diffp) ;\diffp_g) \\
& = q_{2\diffp,\diffp} \deltaopt(\mbr(\diffp), \mbr(\diffp) ;\diffp_g - \diffp) + ( 1- q_{2\diffp,\diffp} ) \deltaopt( \mbr(\diffp), \mbr(\diffp) ;\diffp_g+ \diffp) \\
& = q_{2\diffp,\diffp} \sup_{t_{1,2} \in [0,\diffp]} \left\{ q_{\diffp,t_{1,2}}\deltaopt(\mbr(\diffp) ;\diffp_g - \diffp- t_{1,2}) + ( 1- q_{\diffp,t_{1,2}})\deltaopt(\mbr(\diffp) ;\diffp_g - t_{1,1}) \right\} \\
& \quad + (1 - q_{2\diffp,\diffp}) \sup_{t_{2,2} \in [0,\diffp]} \left\{ q_{\diffp, t_{2,2}}\deltaopt(\mbr(\diffp) ;\diffp_g + \diffp - t_{2,2}) + (1 -  q_{\diffp, t_{2,2}}) \underbrace{\deltaopt(\mbr(\diffp) ;\diffp_g + 2\diffp - t_{2,2})}_{ = 0} \right\} \\
& = q_{2\diffp,\diffp} \sup_{t_{1,2}} \left\{ q_{\diffp,t_{1,2}} \sup_{t_{1,3}} \left\{ q_{\diffp,t_{1,3}} [1 - e^{\diffp_g - \diffp- t_{1,2} - t_{1,3}}]_+ + (1 - q_{\diffp,t_{1,3}} ) [ 1- e^{\diffp_g - t_{1,2} - t_{1,3} }]_+ \right\} \right. \\
&\qquad\qquad\qquad  \left.+ ( 1- q_{\diffp,t_{1,2}}) \sup_{t_{2,3}} \left\{ q_{\diffp,t_{2,3}} [1 - e^{\diffp_g - t_{1,2} - t_{2,3}}]_+ + (1 - q_{\diffp,t_{2,3}}) [ 1 - e^{\diffp_g + \diffp - t_{1,2} - t_{2,3}}]_+ \right\} \right\} \\ 
& \quad + ( 1- q_{2 \diffp,\diffp} ) \sup_{t_{2,2} } \left\{ q_{\diffp, t_{2,2}} \sup_{t_{3,3}} \left\{ q_{\diffp,t_{3,3}} [1 - e^{\diffp_g + \diffp - t_{2,2} - t_{3,3}}]_+ + (1 - q_{\diffp,t_{3,3}}) \underbrace{[1 - e^{\diffp_g + 2\diffp - t_{2,2}- t_{3,3}}]_+}_{=0}\right\} \right\} \\ 
& = \sup_{t_{1,2}, t_{2,2}, t_{1,3}, t_{2,3}, t_{3,3} } \left\{ q_{\diffp,t_{1,2}} q_{\diffp,t_{1,3}} q_{2\diffp,\diffp} [1 - e^{\diffp_g - \diffp - t_{1,2} - t_{1,3} }]_+ + q_{\diffp,t_{2,2}} q_{\diffp,t_{3,3}} (1 - q_{2\diffp,\diffp})[1 - e^{\diffp_g + \diffp - t_{2,2} - t_{3,3}}]_+ \right. \\
& \qquad\qquad \left.  + q_{\diffp,t_{1,2}} (1 -  q_{\diffp,t_{1,3}}) q_{2\diffp,\diffp} [1 - e^{\diffp_g - t_{1,2} - t_{1,3} }]_+   +  (1 - q_{\diffp,t_{1,2}}) q_{\diffp,t_{2,3}} q_{2\diffp,\diffp} [1 - e^{\diffp_g - t_{1,2} - t_{2,3} }]_+ \right. \\
& \qquad \qquad\left. + (1 - q_{\diffp,t_{1,2}}) ( 1- q_{\diffp,t_{2,3}}) q_{2\diffp,\diffp} [1 - e^{\diffp_g + \diffp - t_{1,2} - t_{2,3} }]_+ \right\}
\end{align*}

We next consider the ordering that alternates between BR and DP, as one would do with using the exponential mechanism to discover the $k$ most frequent elements and then adding Laplace noise to the counts of the elements that were discovered.  
\begin{align*}
& \deltaopt(\mbr(\diffp), \mdp(\diffp), \mbr(\diffp) ;\diffp_g) \\
& = \sup_{t_{1,1} \in [0,\diffp]} \left\{ q_{\diffp,t_{1,1}} \deltaopt(\mdp(\diffp), \mbr(\diffp); \diffp_g - t_{1,1})  + (1 - q_{\diffp,t_{1,1}})  \deltaopt(\mdp(\diffp), \mbr(\diffp); \diffp_g + \diffp - t_{1,1})  \right\} \\
&= \sup_{t_{1,1} \in [0,\diffp]} \left\{ q_{\diffp,t_{1,1}} (q_{2\diffp,\diffp} \deltaopt(\mbr(\diffp); \diffp_g - t_{1,1} - \diffp) + ( 1 - q_{2\diffp,\diffp}) \deltaopt(\mbr(\diffp); \diffp_g - t_{1,1} + \diffp) ) \right. \\
& \qquad \left. + (1 - q_{\diffp,t_{1,1}})  \left( q_{2\diffp,\diffp} \deltaopt(\mbr(\diffp); \diffp_g  - t_{1,1}) + ( 1 - q_{2\diffp,\diffp}) \underbrace{\deltaopt(\mbr(\diffp); \diffp_g  - t_{1,1} + 2\diffp)}_{=0}  \right) \right\} \\
& =  \sup_{t_{1,1}, t_{1,3}, t_{2,3}, t_{3,3} } \left\{ q_{\diffp,t_{1,1}} q_{\diffp,t_{1,3}} q_{2\diffp,\diffp}  [1 - e^{\diffp_g - t_{1,1} - \diffp - t_{1,3}}]_+ 
+ q_{\diffp,t_{1,1}} q_{\diffp,t_{2,3}} ( 1 - q_{2\diffp,\diffp}) [1 - e^{\diffp_g - t_{1,1} + \diffp - t_{2,3}}]_+  \right. \\
& \qquad\qquad\qquad\qquad \left. + q_{\diffp,t_{1,1}} (1 - q_{\diffp,t_{1,3}}) q_{2\diffp,\diffp} [1 - e^{\diffp_g - t_{1,1} - t_{1,3}}]_+ + (1 - q_{\diffp,t_{1,1}})  q_{\diffp,t_{3,3}} q_{2\diffp,\diffp} [1 - e^{\diffp_g  - t_{1,1} - t_{3,3}}]_+ \right. \\
& \qquad\qquad\qquad\qquad \left.  + (1 - q_{\diffp, t_{1,1}}) (1 - q_{\diffp,t_{3,3}})  q_{2\diffp,\diffp}[1 - e^{\diffp_g + \diffp - t_{1,1} - t_{3,3}}]_+\right\}
\end{align*}

In the case when $\diffp_g \geq \diffp$, all the terms with $[1 - e^{\diffp_g + \diffp - t - t'}]_+ = 0$, so all three become equal, with a sup over three terms.
Hence, in the case when $\diffp_g \geq \diffp$ we have
\[ 
\deltaopt(\mdp(\diffp), \mbr(\diffp), \mbr(\diffp) ;\diffp_g)= \deltaopt(\mbr(\diffp), \mdp(\diffp), \mbr(\diffp) ;\diffp_g)
\]
We then assume for the rest of the proof that $\diffp_g < \diffp$.  There are many similar terms in each of the expressions.  We start by focusing on the following term when $t < \diffp_g$,
\begin{align*}
 \sup_{t' \in [0,\diffp]} & \left\{ q_{\diffp,t'} [1 - e^{\diffp_g - \diffp- t - t'}]_+ + (1 - q_{\diffp,t'} ) [ 1- e^{\diffp_g - t - t' }]_+ \right\}  \\
 & = \max\left\{ q_{\diffp,\tfrac{\diffp_g - t}{2}}^2 \left(1- e^{- \diffp}\right), 1 - e^{\diffp_g - \diffp - t} \right\}.
\end{align*}
To determine which term attains the maximum, we consider each term,
\begin{align*}
q_{\diffp,\tfrac{\diffp_g - t}{2}}^2 \left(1- e^{- \diffp}\right) & = \frac{\left(1 - e^{\tfrac{\diffp_g - t}{2} - \diffp}\right)^2}{1 - e^{-\diffp}} = \frac{1 -2 e^{\tfrac{\diffp_g - t}{2} - \diffp} + e^{\diffp_g - t -2\diffp}}{1 - e^{-\diffp}} \\
1 - e^{\diffp_g - \diffp - t} & = \frac{1 - e^{- \diffp} - e^{\diffp_g  - t - \diffp} + e^{\diffp_g  - t - 2\diffp} }{1 - e^{-\diffp}}
\end{align*}
Hence, for $t < \diffp_g$, we have 
\begin{align*}
& \max\left\{ q_{\diffp,\tfrac{\diffp_g - t}{2}}^2 \left(1- e^{- \diffp}\right), 1 - e^{\diffp_g - \diffp - t} \right\} = q_{\diffp,\tfrac{\diffp_g - t}{2}}^2 \left(1- e^{- \diffp}\right) \\
& \iff 0 \leq e^{\diffp_g - t} - 2e^{\tfrac{\diffp_g - t}{2}} + 1 = \left(1 - e^{\tfrac{\diffp_g - t}{2}}\right)^2.
\end{align*}
In the case where $t \in [\diffp_g, \diffp]$, we have
\begin{align*}
 \sup_{t' \in [0,\diffp]} & \left\{ q_{\diffp,t'} [1 - e^{\diffp_g - \diffp- t - t'}]_+ + (1 - q_{\diffp,t'} ) [ 1- e^{\diffp_g - t - t' }]_+ \right\} = 1 - e^{\diffp_g - \diffp - t}.
\end{align*}

We next focus on another common term in each expression, considering first when $t < \diffp_g$.
\begin{align*}
 \sup_{t' \in [0,\diffp]} & \left\{ q_{\diffp,t'} [1 - e^{\diffp_g - t - t'}]_+ + (1 - q_{\diffp,t'}) [ 1 - e^{\diffp_g + \diffp - t - t'}]_+ \right\} = q_{\diffp,\tfrac{\diffp_g +\diffp - t}{2}}^2 \left(1- e^{- \diffp}\right)
\end{align*}  
On the other hand, if $t \in [\diffp_g, \diffp]$, we have
\begin{align*}
 \sup_{t'\in [0,\diffp]} &  \left\{ q_{\diffp,t'} [1 - e^{\diffp_g - t - t'}]_+ + (1 - q_{\diffp,t'}) [ 1 - e^{\diffp_g + \diffp - t - t'}]_+ \right\} \\
 & =  \max\left\{ q_{\diffp,\tfrac{\diffp_g +\diffp - t}{2}}^2 \left(1- e^{- \diffp}\right), 1 - e^{\diffp_g - t} \right\}
\end{align*}  
Once again, to determine which term attains the maximum, we consider each term,
\begin{align*}
q_{\diffp,\tfrac{\diffp_g + \diffp - t}{2}}^2 \left(1- e^{- \diffp}\right) & = \frac{\left(1 - e^{\tfrac{\diffp_g - \diffp - t}{2}}\right)^2}{1 - e^{-\diffp}} = \frac{1 -2 e^{\tfrac{\diffp_g - \diffp - t}{2} } + e^{\diffp_g - \diffp - t }}{1 - e^{-\diffp}} \\
1 - e^{\diffp_g - t} & = \frac{1 - e^{- \diffp} - e^{\diffp_g  - t} + e^{\diffp_g - \diffp - t} }{1 - e^{-\diffp}}
\end{align*}
Hence, for $t \in  [\diffp_g,\diffp]$, we have 
\begin{align*}
& \max\left\{ q_{\diffp,\tfrac{\diffp_g + \diffp - t}{2}}^2 \left(1- e^{- \diffp}\right), 1 - e^{\diffp_g - t} \right\} = q_{\diffp,\tfrac{\diffp_g + \diffp - t}{2}}^2 \left(1- e^{- \diffp}\right) \\
& \iff 0 \leq e^{\diffp_g - \diffp - t} - 2e^{\tfrac{\diffp_g- \diffp - t}{2}} + 1 = \left(1 - e^{\tfrac{\diffp_g - \diffp - t}{2}}\right)^2
\end{align*}

Putting this together, we have
\begin{align*}
& \deltaopt(\mdp(\diffp), \mbr(\diffp), \mbr(\diffp) ;\diffp_g) \\
&= q_{2\diffp,\diffp} \max \left\{  \sup_{t_{1,2} \in [0, \diffp_g)} \left\{ q_{\diffp,t_{1,2}} q_{\diffp,\tfrac{\diffp_g - t_{1,2}}{2}}^2 \left(1- e^{- \diffp}\right) + ( 1 - q_{\diffp,t_{1,2} } )  q_{\diffp,\tfrac{\diffp_g + \diffp - t_{1,2}}{2}}^2 \left( 1 - e^{ - \diffp} \right) \right\},  \right. \\
& \qquad\qquad\qquad \left.  \sup_{t_{1,2} \in [\diffp_g,\diffp]} \left\{ q_{\diffp,t_{1,2}} (1 - e^{\diffp_g - \diffp - t_{1,2}}) + (1-  q_{\diffp,t_{1,2}}) q_{\diffp,\tfrac{\diffp_g +\diffp - t_{1,2}}{2}}^2 \left(1- e^{- \diffp}\right) \right\}  \right\} \\
&\qquad\qquad\qquad\qquad +  ( 1- q_{2 \diffp,\diffp} ) \sup_{t_{2,2} \in [\diffp_g,\diffp] } \left\{ q_{\diffp,t_{2,2}}q_{\diffp,\diffp + \tfrac{\diffp_g - t_{2,2}}{2}}^2 \left( 1 - e^{ -\diffp} \right) \right\} \\
& = \max\left\{ \sup_{t_{1,2} \in [0,\diffp_g)} x(t_{1.2}) , \sup_{t_{1,2} \in [\diffp_g,\diffp]} y(t_{1,2})\right\} + \sup_{t_{2,2} \in [\diffp_g,\diffp] } z(t_{2,2})
\end{align*}

We now focus on  $\deltaopt(\mbr(\diffp), \mdp(\diffp), \mbr(\diffp) ;\diffp_g)$, 
\begin{align*}
& \deltaopt(\mbr(\diffp), \mdp(\diffp), \mbr(\diffp) ;\diffp_g) \\
&= \max \left\{ q_{2\diffp,\diffp} \sup_{t_{1,1} \in [0, \diffp_g)} \left\{ q_{\diffp,t_{1,1}} q_{\diffp,\tfrac{\diffp_g - t_{1,1}}{2}}^2 \left(1- e^{- \diffp}\right) + ( 1 - q_{\diffp,t_{1,1} } )  q_{\diffp,\tfrac{\diffp_g + \diffp - t_{1,1}}{2}}^2 \left( 1 - e^{ - \diffp} \right) \right\},  \right. \\
& \qquad \qquad q_{2\diffp,\diffp} \sup_{t_{1,1} \in [\diffp_g,\diffp]} \left\{ q_{\diffp,t_{1,1}} (1 - e^{\diffp_g - \diffp - t_{1,1}}) + (1-  q_{\diffp,t_{1,1}}) q_{\diffp,\tfrac{\diffp_g +\diffp - t_{1,1}}{2}}^2 \left(1- e^{- \diffp}\right) \right. \\
&\qquad\qquad\qquad\qquad \left. \left.+  ( 1- q_{2 \diffp,\diffp} ) q_{\diffp,t_{1,1}}q_{\diffp,\diffp + \tfrac{\diffp_g - t_{1,1}}{2}}^2 \left( 1 - e^{ -\diffp} \right) \right\} \right\} \\
& = \max\left\{\sup_{t_{1,1} \in [0,\diffp_g)}x(t_{1,1}) , \sup_{t_{1,1} \in [\diffp_g, \diffp]} y(t_{1,1}) + z(t_{1,1}) \right\}
\end{align*}

Lastly, we focus on $\deltaopt(\mbr(\diffp), \mbr(\diffp), \mdp(\diffp) ;\diffp_g)$.  Note that when $t_{1,1} < \diffp_g$, we always have $\diffp_g + \diffp - t_{1,1} - t_{1,2} \geq 0$ for all $t_{1,2} \in [0,\diffp]$, and for $t_{1,1}  \in [\diffp_g, \diffp]$ we have $\diffp_g - t_{1,1} - t_{1,2} < 0$ for any $t_{1,2} >0$.
\begin{align*}
& \deltaopt(\mbr(\diffp), \mdp(\diffp), \mbr(\diffp) ;\diffp_g) \\
& = \sup_{t_{1,1}, t_{1,2}} \left\{ q_{\diffp,t_{1,1}} q_{\diffp,t_{1,2}} q_{2\diffp,\diffp} [1 - e^{\diffp_g - \diffp - t_{1,1} - t_{1,2} }]_+ + q_{\diffp,t_{1,1}} q_{\diffp,t_{1,2}} (1 - q_{2\diffp,\diffp})[1 - e^{\diffp_g + \diffp - t_{1,1} - t_{1,2}}]_+ \right. \\
& \qquad\qquad \left.  + q_{\diffp,t_{1,1}} (1 -  q_{\diffp,t_{1,2}}) q_{2\diffp,\diffp} [1 - e^{\diffp_g - t_{1,1} - t_{1,2} }]_+   +  (1 - q_{\diffp,t_{1,1}}) q_{2\diffp,\diffp} q_{\diffp,\tfrac{\diffp_g + \diffp - t_{1,1}}{2}}^2 (1 - e^{-\diffp)} \right\}
\\
&  = \max\left\{ \sup_{t_{1,1} \in [0,\diffp_g)} x(t_{1,1}), \sup_{t_{1,1} \in [\diffp_g, \diffp]} y(t_{1,1}) + z(t_{1,1} ) \right\} 
\end{align*}
Summarizing this, we then have for the nonnegative functions $x(t)$, $y(t)$, and $z(t)$ defined above, 
\begin{align*}
\deltaopt(\mdp(\diffp), \mbr(\diffp), \mbr(\diffp) ;\diffp_g)  & = \max \left\{ \sup_{t \in [0,\diffp_g)}x(t), \sup_{t \in [\diffp_g, \diffp] }y(t) \right\} + \sup_{t' \in [\diffp_g, \diffp]} z(t') \\
 \deltaopt(\mbr(\diffp), \mdp(\diffp), \mbr(\diffp) ;\diffp_g)  & = \max \left\{ \sup_{t \in [0,\diffp_g)}x(t), \sup_{t \in [\diffp_g, \diffp] }y(t) +z(t) \right\}  \\
 & = \deltaopt(\mbr(\diffp), \mbr(\diffp), \mdp(\diffp) ;\diffp_g)  
\end{align*}
\end{proof}

\subsection {Proof of Lemma~\ref{lem:ordbrdp2} \label{lem:ordbrdp2:proof}}
\begin{proof}
We start by computing the derivatives of the necessary functions,
\begin{align*}
x'(t) & = \tfrac{1-e^{\diffp/2}}{(1-e^{-\diffp})^{2}} \left(e^{t} - e^{\diffp/2} \right) e^{\tfrac{\diffp_{g} - t}{2} - 2\diffp}. \\
y'(t) & = e^{-2\diffp - t} \left(e^{t/2} - e^{\tfrac{\diffp_g + \diffp}{2}} \right) \left( e^{3t/2} - e^{\tfrac{\diffp_g + \diffp}{2}} \right)\\
z'(t) & = -e^{-\diffp - t} ( e^{t/2}-e^{\diffp_g/2}) \left(e^{3 t/2} - e^{\diffp + \diffp_g/2} \right)\\
(y+z)'(t) &= (e^{\diffp/2} - 1)^2 (e^{\diffp/2} + 1) (e^{\diffp/2} - e^t) e^{\tfrac{\diffp_g -6 \diffp - t}{2}}.
\end{align*}
Note that for any $t \in [0,\diffp]$, we have
\begin{align*}
& q_{\diffp,\tfrac{\diffp_g - t}{2}}^2 \left(1- e^{- \diffp}\right) > (1 - e^{\diffp_g - \diffp - t}) \\
& \iff 0 < \left(1 - e^{\tfrac{\diffp_g - t}{2}} \right)^2
\end{align*}
Hence, as long as $t \neq \diffp_g$, we have $x(t) > y(t)$.  
We break up the analysis into two cases, depending on whether $\diffp_g$ is larger or smaller than $\diffp/2$.
\paragraph{Case 1:} First assume that $\diffp/2< \diffp_g$.  In this case, we have,
\[
\sup_{t \in [\diffp_g,\diffp]} y(t) < x(\diffp/2) = \sup_{t \in [0,\diffp_g)} x(t).
\]
Furthermore, we have for any $t \in [\diffp_g,\diffp]$
\[
y(t) + z(t) < x(\diffp/2) + z(t) < x(\diffp/2) + z(t) \leq x(\diffp/2) + \sup_{t' \in [\diffp_g,\diffp]}z(t')
\]
\paragraph{Case 2:} When $\diffp/2 = \diffp_g$, we have for any $t \in (\diffp_g,\diffp]$
\begin{align*}
y(t) + z(t) & < x(t) + z(t) < x(\diffp_g) + z(t) \leq  \sup_{t' \in [0,\diffp_g)} x(t') + \sup_{t' \in [\diffp_g,\diffp]} z(t')  \\
\text{ and } y(\diffp_g) + \underbrace{z(\diffp_g)}_{=0} &= x(\diffp_g) < x(\diffp_g) + \sup_{t' \in [\diffp_g,\diffp]} z(t') = \sup_{t' \in [0,\diffp_g)} x(t') + \sup_{t' \in [\diffp_g,\diffp]} z(t') 
\end{align*}
Note that we are taking $\sup$ over compact sets, including the element $\diffp_g$.  Thus, we have, in the case of $\diffp/2 \leq \diffp_g$ that $\deltaopt(\mbr(\diffp), \mdp(\diffp), \mbr(\diffp) ;\diffp_g)  < \deltaopt(\mdp(\diffp), \mbr(\diffp), \mbr(\diffp) ;\diffp_g) $.  

\paragraph{Case 3: }Now we consider $\diffp/2 > \diffp_g$.  
In this case, 
\[
\sup_{t \in [0,\diffp_g)} x(t) = x(\diffp_g) = y(\diffp_g) < y\left( \tfrac{\diffp_g + \diffp}{3}\right) =   \sup_{t' \in [\diffp_g, \diffp]} y(t').
\]
Looking at the derivatives of $z(t)$ and $y(t)$, computed above, the term that maximizes $y(t)$ is not the same as $z(t)$, so that 
\[
\sup_{t \in [\diffp_g, \diffp]} y(t) + z(t) < \sup_{t \in [\diffp_g,\diffp]}y(t) + \sup_{t' \in [\diffp_g, \diffp] } z(t')
\]
This concludes the proof.
\end{proof}

\end{document}